\RequirePackage[l2tabu,orthodox]{nag}
\documentclass
[11pt,letterpaper]
{article}

\usepackage[utf8]{inputenc}
\usepackage{etex}
\usepackage{xspace,enumerate}
\usepackage[dvipsnames]{xcolor}
\usepackage[T1]{fontenc}
\usepackage[full]{textcomp}
\usepackage[american]{babel}
\usepackage{mathtools}
\usepackage{titling}
\usepackage{amsthm}
\newtheorem{theorem}{Theorem}[section]
\newtheorem*{theorem*}{Theorem}

\newtheorem{proposition}[theorem]{Proposition}
\newtheorem*{proposition*}{Proposition}
\newtheorem{lemma}[theorem]{Lemma}
\newtheorem*{lemma*}{Lemma}
\newtheorem{corollary}[theorem]{Corollary}
\newtheorem*{conjecture*}{Conjecture}
\newtheorem{fact}[theorem]{Fact}
\newtheorem*{fact*}{Fact}

\newtheorem*{hypothesis*}{Hypothesis}

\theoremstyle{definition}
\newtheorem{definition}[theorem]{Definition}
\newtheorem*{definition*}{Definition}

\newtheorem{algorithm}[theorem]{Algorithm}

\theoremstyle{remark}

\newtheorem*{claim*}{Claim}
\newtheorem{remark}[theorem]{Remark}
\newtheorem*{remark*}{Remark}

\newtheorem*{observation*}{Observation}

\usepackage[
letterpaper,
top=1.2in,
bottom=1.2in,
left=1in,
right=1in]{geometry}
\usepackage{newpxtext} 
\usepackage{textcomp} 
\usepackage[varg,bigdelims]{newpxmath}
\usepackage[scr=rsfso]{mathalfa}
\usepackage{bm} 
\let\mathbb\varmathbb
\usepackage{microtype}
\usepackage[
pagebackref,
colorlinks=true,
urlcolor=blue,
linkcolor=blue,
citecolor=OliveGreen,
]{hyperref}
\usepackage[capitalise,nameinlink]{cleveref}
\crefname{lemma}{Lemma}{Lemmas}
\crefname{fact}{Fact}{Facts}
\crefname{theorem}{Theorem}{Theorems}
\crefname{corollary}{Corollary}{Corollaries}
\crefname{claim}{Claim}{Claims}
\crefname{example}{Example}{Examples}
\crefname{algorithm}{Algorithm}{Algorithms}
\crefname{problem}{Problem}{Problems}
\crefname{definition}{Definition}{Definitions}
\usepackage{paralist}
\usepackage{turnstile}
\usepackage{mdframed}
\usepackage{tikz}
\usepackage{caption}
\DeclareCaptionType{Algorithm}
\usepackage{newfloat}

\newcommand{\Authornoteb}[2]{}
 \newcommand{\Authornoter}[2]{}

\newcommand{\Authornotecolored}[3]{}
\newcommand{\Authorcomment}[2]{}
\newcommand{\Authorfnote}[2]{}

\newcommand{\Pnote}{\Authornoter{P}}
\newcommand{\Anote}{\Authornoteb{A}}

\definecolor{forestgreen(traditional)}{rgb}{0.0, 0.27, 0.13}

\usepackage{boxedminipage}
\newcommand{\paren}[1]{(#1)}
\newcommand{\Paren}[1]{\left(#1\right)}

\newcommand{\abs}[1]{\lvert#1\rvert}

\newcommand{\set}[1]{\{#1\}}
\newcommand{\Set}[1]{\left\{#1\right\}}

\newcommand{\norm}[1]{\lVert#1\rVert}
\newcommand{\Norm}[1]{\left\lVert#1\right\rVert}

\newcommand{\iprod}[1]{\langle#1\rangle}
\newcommand{\Iprod}[1]{\left\langle#1\right\rangle}

\newcommand{\Esymb}{\mathbb{E}}
\newcommand{\Psymb}{\mathbb{P}}

\DeclareMathOperator*{\E}{\Esymb}

\DeclareMathOperator*{\ProbOp}{\Psymb}
\renewcommand{\Pr}{\ProbOp}

\newcommand{\seteq}{\mathrel{\mathop:}=}
\newcommand{\from}{\colon}

\newcommand{\mper}{\,.}
\newcommand{\mcom}{\,,}
\newcommand\bdot\bullet

\DeclareMathOperator{\Tr}{Tr}

\DeclareMathOperator{\poly}{poly}

\DeclareMathOperator{\argmax}{argmax}

\newcommand{\N}{\mathbb N}
\newcommand{\R}{\mathbb R}

\newcommand{\cA}{\mathcal A}
\newcommand{\cB}{\mathcal B}
\newcommand{\cC}{\mathcal C}
\newcommand{\cD}{\mathcal D}
\newcommand{\cE}{\mathcal E}

\newcommand{\cJ}{\mathcal J}

\newcommand{\cN}{\mathcal N}

\newcommand{\cS}{\mathcal S}

\newcommand{\cV}{\mathcal V}

\def\II{\mathbb{I}}

\newcommand{\tmu}{\tilde{\mu}}

\renewcommand{\S}{\mathbb{S}}

\renewcommand{\leq}{\leqslant}

\renewcommand{\geq}{\geqslant}
\renewcommand{\ge}{\geqslant}

\let\epsilon=\varepsilon
\numberwithin{equation}{section}
\newcommand\MYcurrentlabel{xxx}
\newcommand{\MYstore}[2]{%
  \global\expandafter \def \csname MYMEMORY #1 \endcsname{#2}%
}
\newcommand{\MYload}[1]{%
  \csname MYMEMORY #1 \endcsname%
}
\newcommand{\MYnewlabel}[1]{%
  \renewcommand\MYcurrentlabel{#1}%
  \MYoldlabel{#1}%
}
\newcommand{\MYdummylabel}[1]{}
\newcommand{\torestate}[1]{%
  \let\MYoldlabel\label%
  \let\label\MYnewlabel%
  #1%
  \MYstore{\MYcurrentlabel}{#1}%
  \let\label\MYoldlabel%
}
\newcommand{\restatetheorem}[1]{%
  \let\MYoldlabel\label
  \let\label\MYdummylabel
  \begin{theorem*}[Restatement of \cref{#1}]
    \MYload{#1}
  \end{theorem*}
  \let\label\MYoldlabel
}
\newcommand{\restatelemma}[1]{%
  \let\MYoldlabel\label
  \let\label\MYdummylabel
  \begin{lemma*}[Restatement of \cref{#1}]
    \MYload{#1}
  \end{lemma*}
  \let\label\MYoldlabel
}
\newcommand{\restateprop}[1]{%
  \let\MYoldlabel\label
  \let\label\MYdummylabel
  \begin{proposition*}[Restatement of \cref{#1}]
    \MYload{#1}
  \end{proposition*}
  \let\label\MYoldlabel
}
\newcommand{\restatefact}[1]{%
  \let\MYoldlabel\label
  \let\label\MYdummylabel
  \begin{fact*}[Restatement of \prettyref{#1}]
    \MYload{#1}
  \end{fact*}
  \let\label\MYoldlabel
}
\newcommand{\restate}[1]{%
  \let\MYoldlabel\label
  \let\label\MYdummylabel
  \MYload{#1}
  \let\label\MYoldlabel
}

\newcommand{\e}{\epsilon}

\allowdisplaybreaks
\newcommand*{\tr}{\mathrm{tr}}

\DeclareMathOperator{\pE}{\tilde{\mathbb{E}}}

\newcommand{\1}{\bm{1}}

\newcommand{\spe}{\mathsf{spectral}}

\def\expec#1#2{{\bf \mathbb{E}}_{#1}[ #2 ]}
\def\expecf#1#2{{\bf \mathbb{E}}_{#1}\left[ #2 \right]}

\def\abs#1{\left|#1  \right|}

\def\norm#1{\left\| #1 \right\|}

\def \dtv{d_{\mathsf{TV}}}

\def\expec#1#2{{\bf \mathbb{E}}_{#1}[ #2 ]}
\def\expecf#1#2{{\bf \mathbb{E}}_{#1}\left[ #2 \right]}

\def\tzeta{\tilde{\zeta}}

\title{
  Outlier-Robust Clustering of Non-Spherical Mixtures
}

\author{
    Ainesh Bakshi\thanks{Carnegie Mellon University} \\ abakshi@cs.cmu.edu  \and Pravesh K. Kothari \thanksmark{1} \\praveshk@cs.cmu.edu
}
\begin{document}

\pagestyle{empty}


\maketitle
\thispagestyle{empty} 


\begin{abstract}
We give the first outlier-robust efficient algorithm for clustering a mixture of $k$ statistically separated $d$-dimensional Gaussians ($k$-GMMs). Concretely, our algorithm takes input an $\epsilon$-corrupted sample from a $k$-GMM and whp in $d^{\poly(k/\eta)}$ time,  outputs an approximate clustering that mis-classifies at most $k^{O(k)}(\epsilon+\eta)$ fraction of the points whenever every pair of mixture components are separated by $1-\exp(-\poly(k/\eta)^k)$ in total variation (TV) distance. Such a result was not previously known even for $k=2$. 

TV separation is the statistically weakest possible notion of separation and captures important special cases such as \emph{mixed linear regression} and \emph{subspace clustering}. In particular, it allows clustering of mixtures where all components have the same mean and covariances differ in a single unknown direction or are separated in Frobenius distance. 

Our main conceptual contribution is to distill simple analytic properties - (certifiable) \emph{hypercontractivity} and bounded-variance of degree $2$ polynomials and \emph{anti-concentration} of linear projections - that are necessary and sufficient for mixture models to be (efficiently) clusterable. As a consequence, our results extend to clustering mixtures of arbitrary affine transforms of the uniform distribution on the $d$-dimensional unit sphere. Even the information theoretic clusterability of separated distributions satisfying these two analytic assumptions was not known prior to our work and is likely to be of independent interest.

Our algorithms build on the recent sequence of works relying on \emph{certifiable anti-concentration} first introduced in~\cite{DBLP:journals/corr/abs-1905-05679,DBLP:journals/corr/abs-1905-04660}. Our techniques expand the sum-of-squares toolkit to show robust \emph{certifiability} of TV-separated Gaussian clusters in data. This involves giving a low-degree \emph{sum-of-squares proof} of statements that relate parameter (i.e. mean and covariances) distance to total variation distance  by relying only on degree $2$ polynomial concentration and anti-concentration. 

\end{abstract}

\clearpage


  \microtypesetup{protrusion=false}
  \setcounter{tocdepth}{1}
  \tableofcontents{}

  \microtypesetup{protrusion=true}

\clearpage

\pagestyle{plain}
\setcounter{page}{1}


\section{Introduction} \label{sec:intro}
A flurry of recent work has focused on designing outlier-robust efficient algorithms for statistical estimation for basic tasks such as estimating mean, covariance ~\cite{DBLP:conf/focs/LaiRV16,DBLP:conf/focs/DiakonikolasKK016,DBLP:conf/stoc/CharikarSV17,DBLP:journals/corr/abs-1711-11581,DBLP:journals/corr/SteinhardtCV17,DBLP:conf/soda/0002D019,DBLP:conf/icml/DiakonikolasKK017,DBLP:conf/soda/DiakonikolasKK018,DBLP:conf/colt/0002D0W19}, moment tensors~\cite{DBLP:journals/corr/abs-1711-11581} of distributions, regression~\cite{DBLP:conf/focs/DiakonikolasKS17,DBLP:conf/colt/KlivansKM18,DBLP:conf/icml/DiakonikolasKK019,DBLP:journals/corr/abs-1802-06485,DBLP:journals/corr/abs-1905-05679,DBLP:journals/corr/abs-1905-04660}, and clustering of spherical mixtures~\cite{DBLP:conf/focs/DiakonikolasKS17,KothariSteinhardt17,HopkinsLi17}. This progress (see~\cite{diakonikolas2019recent} for a recent survey) has come via fundamentally new algorithmic techniques such as agnostic filtering~\cite{DBLP:conf/focs/DiakonikolasKK016} and robust-learning frameworks based on the sum-of-squares method in both the strong contamination~\cite{KothariSteinhardt17,DBLP:journals/corr/abs-1711-11581,HopkinsLi17} and list-decodable learning models \cite{ben2002size,DBLP:journals/corr/abs-1905-05679,DBLP:journals/corr/abs-1905-04660,raghavendra2020list}.

In this paper, we extend this line of work by studying outlier-robust \emph{clustering} of mixtures of distributions that exhibit mean or covariance separation. As a corollary, we obtain a poly-time outlier-robust algorithm for clustering mixtures of $k$-Gaussians ($k$-GMMs) when each pair of components is separated in total variation (TV)\footnote{The TV distance between distributions with PDFs $p,q$ is defined as $\frac{1}{2}\int_{-\infty}^{\infty}|p(x)-q(x)|dx$.} distance. This is the information-theoretically weakest notion of separation, allows components of same mean but variances differing in an unknown direction\footnote{As an interesting example, consider the case of subspace clustering: mixture of standard Gaussians restricted to unknown distinct subspaces. The components have a TV distance of 1 regardless of how close the subspaces are and thus satisfy our assumptions. } or covariances separated in \emph{relative} Frobenius distance (see Fig~\ref{fig:separability conditions}) and includes well-studied problems such as \emph{mixed linear regression} and \emph{subspace clustering} as special cases.

\paragraph{Clustering all Hypercontractive and Anti-Concentrated Distributions.} The Gaussian Mixture Model has been the subject of a century-old line of research beginning with Pearson~\cite{pearson1894contributions}. A $k$-GMM $\sum_{r\leq k} p_r \cN(\mu(r),\Sigma(r))$ is a probability distribution sampled by choosing a component $r \sim [k]$ with probability $p_r$ and outputting a sample from the Gaussian distribution with mean $\mu(r)$ and covariance $\Sigma(r)$. In the $k$-GMM learning problem, the goal is to output an approximate \emph{clustering} of the input sample or estimate the parameters (the mean and covariances) of the components. Progress on provable algorithms for learning $k$-GMMs began with the influential work of Dasgupta~\cite{dasgupta1999learning} followed up by~\cite{sanjeev2001learning,vempala2004spectral,brubaker2008isotropic,brubaker2009robust} yielding  clustering algorithms that succeed under various separation assumptions. These assumptions, however, do not capture natural separated instances of Gaussians (e.g., see (b) or (c) in Fig~\ref{fig:separability conditions}). A more general approach~\cite{DBLP:conf/stoc/KalaiMV10,DBLP:conf/focs/MoitraV10,DBLP:journals/siamcomp/BelkinS15} circumvents clustering altogether by giving an efficient algorithm ( time $\sim d^{\poly(k)}$) for parameter estimation without any separation assumptions.

Our main result is a polynomial-time algorithm based on the sum-of-squares (SoS) method for clustering TV-separated $k$-GMMs in the presence of an $\epsilon$-fraction of fully adversarial outliers. Such a result was not known prior to our work even for $k=2$. Our algorithms actually succeed more generally for mixtures of all distributions that satisfy two well-studied analytic conditions: certifiable \emph{anti-concentration} and certifiable \emph{hypercontractivity} and thus apply, for e.g., to clustering mixtures of arbitrary affine transforms of uniform distribution on the unit sphere. We consider identifying clean analytic conditions that enable the existence of efficient clustering algorithms an important contribution of our work.

\paragraph{Techniques.}  
Our work is naturally related to the recent progress (see Chapter 4~\cite{TCS-086} for an exposition) on learning spherical mixtures\footnote{More generally, the SoS-based algorithms succeed when the means of the components are separated when compared to the maximum variance of the components in any direction.} of Gaussians~\cite{diakonikolas2018list,KothariSteinhardt17,HopkinsLi17} and more generally, all Poincaré distributions~\cite{KothariSteinhardt17}.  These results rely on subgaussian moment \emph{upper bounds} and extend to the outlier-robust setting. However, moment upper bounds are inherently insufficient to cluster non-spherical mixtures. Informally, this is because the property of having subgaussian moment upper bounds is closed under taking mixtures and thus cannot distinguish between a single Gaussian and mixture of a few.

Indeed, it was ``folklore'' that obtaining generalization of the results above to non-spherical mixtures will likely require algorithmic use of \emph{moment lower bounds}. A recent line of work begun by~\cite{DBLP:journals/corr/abs-1905-05679,DBLP:journals/corr/abs-1905-04660} and further built on in~\cite{bakshi2020list,raghavendra2020list} introduced \emph{certifiable anti-concentration} that allows algorithmically accessing moment lower-bounds to solve list-decodable variants (harsher outlier model than ours) of regression and subspace recovery. An important technical contribution of our work is to show that moment lower-bounds, inferred from anti-concentration inequalities along with certifiable hypercontractivity of degree-2 polynomials are enough to obtain the desired generalization for clustering of all TV-separated mixtures.

The key technical contribution of our work is a low-degree sum-of-squares proof of a basic statistical statement that gives a strong, dimension-independent bound relating closeness of distribution in \emph{total variation distance} (TV) to an appropriate \emph{parameter distance} between their means and covariances. Our proof of this basic result works for all distributions that satisfy (certifiable)  anti-concentration and hypercontractivity of degree-2 polynomials. To the best of our knowledge, even the information-theoretic relationship between total variation and parameter distances of such distributions was not known prior to our work. Further, in Section~\ref{sec:subgaussian-no-TV-vs-param}, we give a simple proof by exhibiting two (certifiably) hypercontractive (and, thus, also subgaussian) distributions that are $1-\eta$  close in TV distance but arbitrarily far in parameter distance showing that moment upper bounds are provably not enough for the TV vs parameter distance relationships to hold.

Along the way, we grow the general purpose SoS toolkit for algorithm design. For instance, we give low-degree sum-of-squares formulations of \emph{conditional} arguments using uniform polynomial approximators and basic matrix analytic facts (see for e.g. Lemma~\ref{lem:contraction-property}). As another application of our techniques, we give an outlier-robust algorithm for covariance estimation of all certifiable hypercontractive distributions with $\tilde{O}(\epsilon)$ relative Frobenius error guarantee. All prior works~\cite{DBLP:journals/corr/abs-1711-11581,DBLP:conf/focs/LaiRV16} either gave error guarantees in spectral norm, which only translate into dimension dependent guarantees for relative Frobenius distance, or worked only for the Gaussian distribution~\cite{DBLP:conf/focs/DiakonikolasKK016}). Combined with our outlier-robust clustering algorithm, we obtain a statistically optimal outlier-robust parameter estimation algorithms for mixtures of Gaussians.





\subsection{Our Results}\label{sec:results}
\paragraph{Outlier-Robust Clustering of $k$-GMMs.}
Our main result is an efficient algorithm for outlier-robust clustering of $k$-GMMs whenever every pair of components of the mixture are separated in total variation distance. Formally, our algorithms work  in the \emph{strong contamination} model studied in the bulk of the prior works on robust estimation where an adversary changes an arbitrary, potentially adversarially chosen $\epsilon$-fraction of the input sample before passing it on to the algorithm.

\begin{theorem}[Main Result, Outlier-Robust Clustering of $k$-GMMs]
\label{thm:clustering-robust-main}
Fix $\eta, \epsilon > 0$.
Let $\cD_r = \cN(\mu(r),\Sigma(r))$ for $r \leq k$ be $k$-Gaussians such that $d_{TV}(\cD_r, \cD_{r'}) \geq 1- \exp(-\poly(k/\eta))$ whenever $r \neq r'$.
Then, there exists an algorithm that takes input an $\epsilon$-corruption $Y$ of a sample $X = C_1 \cup C_2 \cup \ldots \cup C_k$ of size $n$, with equal sized clusters $C_i$ drawn i.i.d. from $\cD_i$ for each $r \leq k$, and with probability $\geq 0.99$, outputs an approximate clustering $Y = \hat{C}_1 \cup \hat{C}_2 \cup \ldots \cup \hat{C}_k$ satisfying $\min_{i \leq k} \frac{|\hat{C}_i\cap C_i|}{|C_i|}  \geq 1 - O(k^{2k}) (\epsilon+\eta)$. The algorithm succeeds whenever $n \geq d^{O(\poly(k/\eta))}$ and runs in time $n^{O\Paren{\poly( k/\eta))}}$.
\end{theorem}
We can use off-the-shelf robust estimators for mean and covariance of Gaussians(~\cite{DBLP:conf/focs/DiakonikolasKK016}) in order to get statistically optimal estimates of the mean and covariances of the target $k$-GMM.
\begin{corollary}[Parameter Recovery from Clustering]
In the setting of Theorem~\ref{thm:clustering-robust-main}, with the same running time, sample complexity and success probability, our algorithm can output $\{\hat{\mu}(r),\hat{\Sigma}(r)\}_{r \leq k}$ such that for some permutation $\pi:[k] \rightarrow [k]$, $d_{TV}(\cN(\mu(r),\Sigma(r)), \cN(\hat{\mu}(\pi(r)),\hat{\Sigma}(\pi(r))) \leq \tilde{O}(k^{2k}(\epsilon+\eta))$.
\end{corollary}

\paragraph{Discussion}
These are the first outlier-robust algorithms that work for clustering $k$-GMMs under information-theoretically optimal separation assumptions. Such results were not known even for $k=2$. To discuss the bottlenecks in prior works, it is helpful to use (see Prop~\ref{prop:tv-vs-param-for-gaussians} in Section~\ref{Sec:tv-vs-param-gaussian} for a proof) following consequence of two Gaussians with means $\mu(1),\mu(2)$ and covariances $\Sigma(1),\Sigma(2)$ being at a TV distance $\geq 1-\exp(-O(\Delta^2))$ in terms of the distance between their parameters.

\begin{figure}[htb!]
    \centering
    \includegraphics[scale=0.5]{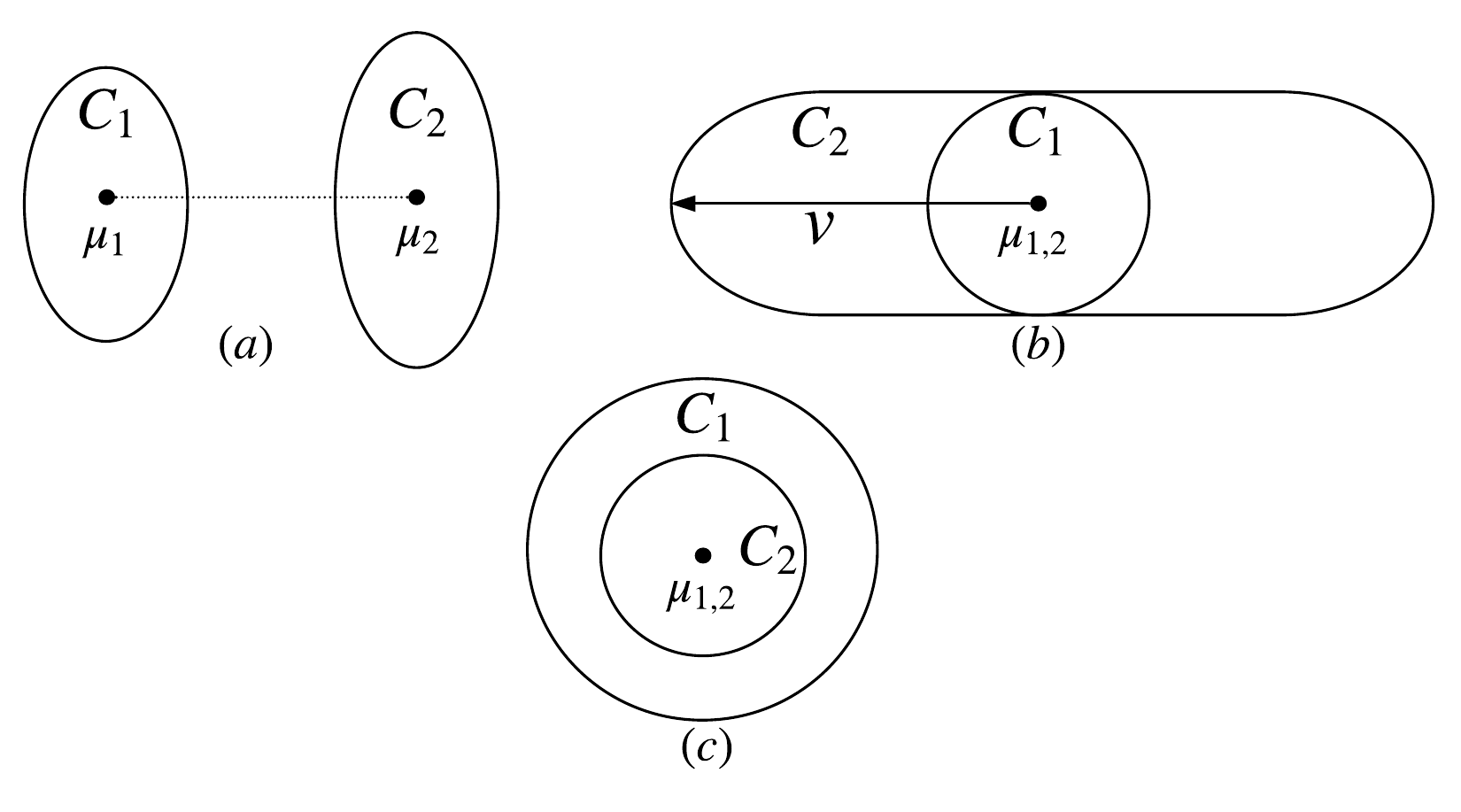}
    \caption{(a) Mean Separation (b) Spectral Separation (c) Relative Frobenius Separation}
    \label{fig:separability conditions}
\end{figure}
\vspace{-3mm}
\begin{definition}[$\Delta$-Separated Mixture Model] \label{def:separated-mixture-model}
An equi-weighted mixture $\cD_1, \cD_2,\ldots, \cD_k$ with parameters $\{\mu(i),\Sigma(i)\}_{i \leq k}$ is $\Delta$-separated if for every pair of distinct components $i,j$, one of the following three conditions hold ($\Sigma^{\dagger/2}$ is the  square root of pseudo-inverse of $\Sigma$):
\begin{enumerate}
  \item \textbf{Mean-Separation:} $\exists v \in \R^d$ such that $\langle \mu(i) - \mu(j), v \rangle^2 > \Delta^2 v^{\top} (\Sigma(i) + \Sigma(j)) v$,
  \item \textbf{Spectral-Separation:} $\exists v \in \R^d$ such that $v^{\top} \Sigma(i) v > \Delta v^{\top}\Sigma(j)v$,
  \item \textbf{Relative-Frobenius Separation:}\footnote{Unlike the other two distances, relative Frobenius distance is meaningful only for high-dimensional Gaussians. As an illustrative example, consider two $0$ mean Gaussians with covariances $\Sigma_1= I$ and $\Sigma_2 = (1+ \Theta(1/\sqrt{d}))I$. Then, for large enough $d$, the parameters are separated in relative Frobenius distance but not spectral or mean distance.} $\Sigma(i)$ and $\Sigma(j)$ have the same range space and $\Norm{ \Sigma(i)^{\dagger/2} \Sigma(j) \Sigma(i)^{\dagger/2}-I}^2_F > \Delta^2 \Norm{\Sigma(i)^{\dagger/2} \Sigma(j)^{1/2}}_{op}^4$.
\end{enumerate}
\end{definition}
The key bottleneck for known algorithms was handling separation in cases 2 and 3 above. 

\emph{Dependence on $k$.} The dependence on the number of components $k$ in our result is doubly exponential. A singly exponential lower bound in the statistical query model (for even the non-robust variant) was shown by Diakonikolas, Kane and Stewart~\cite{DBLP:conf/focs/DiakonikolasKS17}.

\emph{Dependence on $\epsilon$:} While the information-theoretically optimal bound on fraction of misclassified samples is $O(k\epsilon)$, we only obtain the weaker bound of $k^{O(k)}\epsilon$. Our algorithms in Sections~\ref{sec:clustering},~\ref{sec:outlier-robust-clustering} do obtain this the stronger $O(k\epsilon)$ guarantee at the cost of a larger running time. We believe it should be possible to match the optimal recovery guarantee without incurring this running time penalty.

\emph{Handling General Weights.} While we have not attempted to do it in this work, it seems possible to generalize our techniques to handle arbitrary mixing weights albeit with an exponential dependence on the reciprocal of the smallest mixing weight in both the running time and sample complexity on the algorithm.

\paragraph{Clustering and Parameter Recovery for all Reasonable Distributions.}
Our results apply more generally to mixture models where each component distribution $\cD$ satisfies two natural and well-studied analytic conditions: \emph{hypercontractivity} of degree 2 polynomials and \emph{anti-concentration} of all directional marginals. Our algorithmic results hold for distributions (e.g. Gaussians and affine transforms of uniform distribution on the unit sphere) that admit efficiently verifiable analogs (in the SoS proof system, see Sec~\ref{sec:preliminaries}) of these properties.
\begin{definition}[Certifiable Hypercontractivity]
\label{def:certifiable-hypercontractivity}
An isotropic distribution $\cD$ on $\R^d$ is said to be $h$-certifiably $C$-hypercontractive if there's a degree $h$ sum-of-squares proof of the following unconstrained polynomial inequality in $d \times d$ matrix-valued indeterminate $Q$:
\[ \E_{x \sim \cD} \paren{x^{\top}Qx}^{h} \leq \paren{Ch}^{h} \Paren{\E_{x \sim \cD} \paren{x^{\top}Qx}^2}^{h/2}\mper
\]

A set of points $X \subseteq \R^d$ is said to be  $C$-certifiably hypercontractive if the uniform distribution on $X$ is  $h$-certifiably $C$-hypercontractive.
\end{definition}

Hypercontractivity is an important notion in high-dimensional probability and analysis on product spaces~\cite{MR3443800-ODonnell14}. Kauers, O'Donnell, Tan and Zhou~\cite{DBLP:conf/soda/KauersOTZ14} showed certifiable hypercontractivity of Gaussians and more generally product distributions with subgaussian marginals. Certifiable hypercontractivity strictly generalizes the better known \emph{certifiable subgaussianity} property (studied first in~\cite{DBLP:journals/corr/abs-1711-11581}) that controls higher moments of linear polynomials. 

\paragraph{Certifiable anti-concentration.} In contrast to subgaussianity, anti-concentration forces \emph{lower-bounds} of the form $\Pr[\iprod{x,v}^2 \geq \delta \Norm{v}^2_2] \geq \delta'$ for all directions $v$. Certifiable anti-concentration was recently introduced in independent works of Karmalkar, Klivans and Kothari~\cite{DBLP:journals/corr/abs-1905-05679} and Raghavendra and Yau~\cite{DBLP:journals/corr/abs-1905-04660} and later used ~\cite{bakshi2020list,raghavendra2020list} for the related problems of list-decodable linear regression and subspace recovery\footnote{List-decodable versions of these problems generalize the ``mixture'' variants - mixed linear regression and subspace clustering - that are easily seen to be special cases of mixtures of $k$-Gaussians with TV separation 1.}.

Following~\cite{DBLP:journals/corr/abs-1905-05679}, we formulate  certifiable anti-concentration via a univariate, even polynomial $p_{\delta,\Sigma}$ that uniformly approximates the $0$-$1$ core-indicator $\1(\iprod{x,v}^2 \geq \delta v^{\top}\Sigma v)$ over a large enough interval around $0$. Let $q_{\delta,\Sigma}(x,v)$ be a multivariate (in $v$) polynomial defined by $q_{\delta, \Sigma}(x,v) = \Paren{v^{\top}\Sigma v}^{2s} p_{\delta,\Sigma}\Paren{\frac{\iprod{x,v}}{\sqrt{v^{\top}\Sigma v}}}$.Since $p_{\delta,\Sigma}$ is an even polynomial, $q_{\delta,\Sigma}$ is a polynomial in $v$.



\begin{definition}[Certifiable Anti-Concentration] \label{def:certifiable-anti-concentration}
A mean $0$ distribution $D$ with covariance $\Sigma$ is $2s$-certifiably $(\delta,C\delta)$-anti-concentrated if for $q_{\delta,\Sigma}(x,v)$ defined above,  there exists a degree $2s$ sum-of-squares proof of the following two unconstrained polynomial inequalities in indeterminate $v$:
\[
\Set{\iprod{x,v}^{2s} + \delta^{2s} q_{\delta,\Sigma}(x,v)^2 \geq \delta^{2s} \Paren{v^{\top} \Sigma v}^{2s}} \text{ , } \Set{\E_{x \sim D} q_{\delta,\Sigma}(x,v)^2 \leq C\delta \Paren{v^{\top}\Sigma v}^{2s}}\mper
\]
An isotropic subset $X \subseteq  \R^d$ is $2s$-certifiably $(\delta,C\delta)$-anti-concentrated if the uniform distribution on $X$ is $2s$-certifiably $(\delta,C\delta)$-anti-concentrated.
\label{def:certifiable-anti-concentration-homogenous}
\end{definition}
\begin{remark}
For natural examples, $s(\delta) \leq 1/\delta^c$ for some fixed constant $c$. For e.g., $s(\delta)= O(\frac{1}{\delta^2})$ for standard Gaussian distribution and the uniform distribution on the unit sphere (see~\cite{DBLP:journals/corr/abs-1905-05679} and ~\cite{bakshi2020list}). To simplify notation, we will assume $s(\delta)\leq \poly(1/\delta)$ in the statement of our results.
\end{remark}

Additionally, we need that the variance of degree-$2$ polynomials is bounded in terms of the Frobenius norm of the coefficients of the polynomial. Formally, 

\begin{definition}[Degree-$2$ Polynomials with Certifiably Bounded Variance]
\label{def:bounded_variance}
A mean $0$ distribution $\cD$ with covariance $\Sigma$ certifiably bounded variance degree $2$ polynomials if there is a  degree $2$ sum-of-squares proof of the following inequality in the indeterminate $Q \in \mathbb{R}^{d \times d}$
\[
\Set{ \E_{x\sim \cD} \Paren{x^{\top} Q x-\E_{x\sim \cD} x^{\top} Q x}^2 \leq C \Norm{\Sigma^{1/2}Q\Sigma^{1/2}}_F^2 } \mper
\]
\end{definition}

Our general result gives an outlier-robust clustering algorithm for separated mixtures of \emph{reasonable} distributions, i.e., one that satisfies both certifiable hypercontractivity, anti-concentration and has bounded variance of degree-$2$ polynomials(see Definition~\ref{def:nice_distribution}). Even the information-theoretic (and without outliers, i.e., $\epsilon = 0$) clusterability of such distributions was not known prior to our work.

\begin{theorem}[Outlier-Robust Clustering of Separated Mixtures, see Theorem~\ref{thm:main-robust-clustering-section} for precise bounds]
Fix $\eta > 0, \epsilon > 0$.
Let $\cD_r$ be a $\Delta$-separated mixture of \emph{reasonable} distributions (see Definition~\ref{def:nice_distribution}).
Then, there exists an algorithm that takes input an $\epsilon$-corruption $Y$ of a sample $X = C_1 \cup C_2 \cup \ldots C_k$, with true clusters $C_i$ of size $n/k$ drawn i.i.d. from $\cD_r$ for each $r \leq k$, and outputs an approximate clustering $Y = \hat{C}_1 \cup \hat{C}_2 \cup \ldots \cup \hat{C}_k$ satisfying $\min_{i \leq k} \frac{|\hat{C}_i\cap C_i|}{|C_i|}  \geq 1 - O(k^{2k}) (\epsilon+\eta)$.
The algorithm succeeds with probability  at least  $0.99$ over the draw of the original sample  $X$ whenever $n \geq d^{O(\poly(k/\eta))}$ and runs in time $n^{O\Paren{\poly( k/\eta))}}$ whenever $\Delta \geq \poly(k/\eta)^k$.
\label{thm:main-clustering-general-main}
\end{theorem}

\paragraph{Robust Covariance Estimation in Relative Frobenius Distance.} In Section~\ref{sec:param_recovery}, we give an outlier-robust algorithm for covariance estimation for all certifiably hypercontractive distributions.

\begin{theorem}[Robust Parameter Covariance Estimation for Certifiably Hypercontractive Distributions] \label{thm:param-estimation-main}
Fix an $\epsilon > 0$ small enough fixed constant so that $Ct\epsilon^{1-4/t} \ll 1$\footnote{This notation means that we needed $Ct\epsilon^{1-4/t}$ to be at most $c_0$ for some absolute constant $c_0 > 0$}. For every even $t \in \N$, there's an algorithm that takes input $Y$ be an $\epsilon$-corruption of a sample $X$ of size $n \geq n_0 = d^{O(t)}/\epsilon^2$ from a $2t$-certifiably $C$-hypercontractive and certifiably $C$-bounded variance with unknown mean $\mu_*$ and covariance $\Sigma_*$ respectively and in time $n^{O(t)}$ outputs an estimate $\hat{\mu}$ and $\hat{\Sigma}$ satisfying:
\begin{enumerate}
\item $\Norm{\Sigma^{-1/2}(\mu_*-\hat{\mu})}_2 \leq O(Ct)^{1/2} \epsilon^{1-1/t}$,
\item $(1-\eta) \Sigma_* \preceq \hat{\Sigma} \preceq (1+\eta) \Sigma_*$ for $\eta \leq O(Ck) \epsilon^{1-2/t}$, and,
\item $\Norm{\Sigma_*^{-1/2} \hat{\Sigma} \Sigma_*^{-1/2}-I}_F \leq (Ct) O(\epsilon^{1-1/t})$.
\end{enumerate}
In particular, letting $t = O(\log(1/\epsilon))$ results in the error bounds of $\tilde{O}(\epsilon)$ in all the three inequalities above.
\end{theorem}
The first two guarantees above were shown in~\cite{DBLP:journals/corr/abs-1711-11581} for all certifiably subgaussian distributions.  ~\cite{DBLP:journals/corr/abs-1711-11581} also observed (see last paragraph of page 6 for a  counter example) that it is provably impossible to obtain dimension-independent error bounds in relative Frobenius distance assuming only certifiable subgaussianity. We prove that under the stronger assumption of certifiable \emph{hypercontractivity} along with certifiably bounded variance of degree $2$ polynomials, we can indeed obtain dimension-independent, information-theoretically optimal (for e.g. for Gaussians) error guarantees in relative Frobenius error. Prior works either obtained the weaker spectral error guarantee (that incurs a loss of $\sqrt{d}$ factor when translating into relative Frobenius distance)~\cite{DBLP:conf/focs/LaiRV16,DBLP:journals/corr/abs-1711-11581} or worked only for Gaussians~\cite{DBLP:conf/focs/DiakonikolasKK016}\footnote{We note that the algorithm of~\cite{DBLP:conf/focs/DiakonikolasKK016} for Gaussian distributions works in fixed polynomial time to obtain $\tilde{O}(\epsilon)$ error-estimate of the covariance in relative Frobenius distance  whereas our algorithm works more generally for all certifiably hypercontractive distributions but runs in time $d^{O(log^2(1/\epsilon))}$.}.

Combining this theorem with our clustering results above yields:
\begin{corollary}[Parameter Recovery from Clustering, General Case]
In the setting of either Theorem~\ref{thm:main-clustering-general-main}, there's an algorithm with same bounds on running time and sample complexity, that with probability at least $0.99$, outputs $\{\hat{\mu}(r),\hat{\Sigma}(r)\}_{r \leq k}$ such that for some permutation $\pi:[k] \rightarrow [k]$, for every $i$, $\hat{\mu}(\pi(i)),\hat{\Sigma}(\pi(i))$ is $\Delta$-close to $\mu,\Sigma$ in the three distances defined in Definition~\ref{def:separated-mixture-model} for $\Delta = \tilde{O}(k^{O(k)}(\epsilon+\eta))$.
\end{corollary}

\paragraph{Applications.} Two important special cases of mixtures of separated reasonable distributions are noiseless \emph{mixed linear regression} where we are given samples generated as $y = \iprod{\ell,x}$ where $x$ is drawn from $\N(0,I_d)$ and $\ell$ is chosen uniformly from an unknown list $(\ell_1,\ell_2,\ldots,\ell_k)$ and \emph{subspace clustering} - where the input is a mixture of isotropic Gaussians restricted to a $k$ unknown subspaces. There's extensive work~\cite{MR1028403,doi:10.1162/neco.1994.6.2.181,MR2757044,2013arXiv1310.3745Y,DBLP:journals/corr/BalakrishnanWY14,DBLP:conf/colt/ChenYC14,DBLP:conf/nips/Zhong0D16,DBLP:conf/aistats/SedghiJA16,DBLP:conf/colt/LiL18,Vidal2011-ag,Parsons2004-ur} on both these problems in signal processing and machine learning with recent push~\cite{DBLP:journals/corr/abs-1912-07629,DBLP:conf/colt/LiL18} in TCS to obtain efficient algorithms with provable guarantees. Both these cases are immediately seen as mixtures with pairwise separation of $\infty$ (for Gaussians,  this is equivalent to TV separation of 1). Thus, we immediately obtain efficient outlier-robust algorithms for these problems.

\subsection{Previous Version}
In a previous version of this paper, our main result was the clustering algorithm in Algorithm~\ref{algo:polynomial-rounding-for-pseudo-distribution-robust} where degree of the polynomial running time scales linearly in $O(\log \kappa)$ where $\kappa$ is the spread of the target mixture.

\subsection{Related Independent Works}

In an independent work, Diakonikolas, Hopkins, Kane and Karmalkar~\cite{DHKK20} obtained an efficient algorithm for clustering mixtures of $k$-Gaussians with components separated in TV distance. The running time of their algorithm has a slightly worse dependence on $k$ than our algorithm (invoked for Gaussian distribution). For a constant accuracy and fraction of outliers, our algorithm needs $n= d^{k^{O(k)}}$ samples and $n^{k^{O(k)}}$ time while the one in~\cite{DHKK20} needs $d^{F(k)}$ samples and $n^{F(k)}$ time where $F(k)$ is a function that grows  as a $\poly(k)$ size tower of exponentials in $k$.

\section{Overview}
\label{sec:overview}
In this section, we given an informal overview of our approach and main ideas.
All of our conceptual ideas appear in obtaining a clustering algorithm in the non-robust (without outliers) setting.
So we will restrict ourselves to this setting for most of this section.
The reader might find it helpful to use this overview as a ``chart'' to navigate the somewhat technical structure of our proof.

Formally, our results hold for $\Delta$-separated (in the sense of Definition~\ref{def:separated-mixture-model}) mixtures of all \emph{reasonable} distributions defined below.

\begin{definition}[Reasonable Distributions]
\label{def:nice_distribution}
An isotropic (i.e. mean $0$ and $I$-covariance) distribution $\cD$ on $\R^d$ is \emph{reasonable} if it satisfies the following two properties:
\begin{enumerate}
\item \emph{Certifiable Anti-Concentration Under $4$-wise Convolutions}: The distribution of $x\pm y \pm z\pm w$ for independent copies $x,y,z,w \sim \cD$ is certifiably $(\delta,C\delta)$ anti-concentrated %
for all $\delta > 0$.
\item  \emph{Certifiable Hypercontractivity Under $4$-wise Convolutions}: The distribution of $x\pm y \pm z \pm w$ for independent $x,y,z,w\sim \cD$ has certifiably  hypercontractive degree $2$ polynomials.
\item \emph{Certifiable Bounded Variance}: The distribution of $x \sim \cD$ has degree $2$ polynomials of certifiably $C$-bounded variance (Definition \ref{def:bounded_variance}). 
\end{enumerate}
Observe that if $\cD$ has $h$-certifiably $C$-hypercontractive degree 2 polynomials then it is also $h$-certifiably $C$-subgaussian.  For any $\mu,\Sigma\succ 0$, we denote $\cD(\mu,\Sigma)$ to be the distribution of the random variable $\Sigma^{1/2} x + \mu$ where $x \sim \cD$.
\end{definition}
In Section~\ref{sec:reasonable-distributions}, we prove that Gaussian distributions and affine transforms of uniform distribution on the unit sphere are reasonable distributions.



\paragraph{Setup.} The input to our algorithm is a sample $X$ of size $n$ from an equi-weighted mixture of $\{\cD(\mu(r),\Sigma(r))\}_{r \leq k}$ for some reasonable distribution $\cD$. Let $X = C_1 \cup C_2 \cup \ldots C_k$ be the partition of $X$ into true clusters unknown to the algorithm. We follow the high-level approach of using low-degree sum-of-squares proofs of \emph{certifiability}\footnote{We find the term \emph{certifiability} more accurate than the usual ``identifiability'' in this context. Formally, certifiability refers to checking that a purported solution is ``good'' while identifiability relates to a sample containing information about a certain parameter we desire to estimate. Certifiability implies identifiability - it gives a test that we can check for all possible candidate solutions with the guarantee that only true solutions will pass the checks. } to design efficient algorithms. 

The two key parts of our proofs are 1) giving a low-degree sum-of-squares proof of certifiability of approximate clusters and 2) a recursive clustering based on rounding pseudo-distributions. We discuss the high-level ideas behind both these pieces below.

\paragraph{Certifying Purported Clusters.}
In this approach, we ignore the algorithmic issues and focus simply on the issue of how to \emph{certify} that a given subset $\hat{C} \subseteq X$ - described by an associated set of indicator variables $w_1,w_2,\ldots, w_n$ of the samples included in $\hat{C}$ -  is (close to) a true cluster $C_r$ for some $r \leq k$. Let $w(C_r) = \frac{|\hat{C} \cap C_r|}{|C_r|}$ for every $r$. 

By standard concentration arguments (see Lemma~\ref{lem:typical-samples-good}), for $n$ large enough, the uniform distribution on $C_i$ for each $i$ is itself reasonable - that is, it satisfies the conditions in Def~\ref{def:nice_distribution-overview}. Further, the parameters of each  $C_r$ are close to the true parameters $\mu(r),\Sigma(r)$. Instead of introducing new notation, we will simply assume that $\mu(r),\Sigma(r)$ are the mean and covariances of $C_r$ (instead of the distribution that generates $C_r$). This slight abuse of notation doesn't meaningfully change our results or techniques.

Finally, another simple but useful observation is that for distributions that are uniform on subsets of $A,B \subseteq X$ of size $n/k$, the total variation distance equals $1-(k/n)|A\cap B|$. In particular, large TV distance corresponds to small intersection and vice-versa.

The only properties we know of the true clusters is that they are of size $n/k$ and that uniform distributions on them are reasonable distributions. Thus, the natural checks we can perform on $\hat{C}$ is to simply verify the properties of being certifiably hypercontractive and anti-concentrated. Our polynomial constraint system $\cA$ in Section~\ref{sec:clustering} in indicator variables $w$ encodes these checks.

Since we check only the properties that a true cluster $C_i$ would satisfy, it's clear that the true clusters should pass our checks. Thus, we can focus on proving \emph{soundness} of our test: if $\hat{C}$ passes the checks we made, then it must be close to one of the true clusters $C_i$s. The key ``bad case'' for us to rule out is when $w(C_r)$ and $w(C_{r'})$ are both large for some $r \neq r'$. In that case, the set $\hat{C}$ indicated by $w$ cannot be close to any single cluster $C_i$.

Indeed, bulk of our analysis goes into showing that for  every $r \neq r'$, $w(C_r)w(C_{r'})$ must be small whenever $w$ passes our checks above (see Lemma~\ref{lem:intersection-bounds-from-separation}). This immediately implies that $w(C_r)$ and $w(C_{r'})$ cannot simultaneously be large. We call such results \emph{simultaneous intersection bounds} because they control the simultaneous intersection of $\hat{C}$ with $C_r$ and $C_{r'}$.

\subsection{Enter TV vs Parameter Distance Lemmas}
When $w(C_r)$ and $w(C_{r'})$ are simultaneously larger than, say, $\eta$, the uniform distribution on $\hat{C}$ is $1-\eta$ close in TV  distance  to both $C_r$ and $C_{r'}$. On the other hand, since $C_r$ and $C_{r'}$ have $\Delta$-separated parameters, the parameters of the uniform distribution on $\hat{C}$ must be far from that of at least one of $C_r$ and $C_{r'}$ - say, $C_r$ WLOG (follows from a triangle-like inequality that is easy to prove for the notion of parameter distance in Def~\ref{def:separated-mixture-model}). In that case, we have a reasonable distribution (uniform distribution  on $\hat{C}$) that is close to another reasonable distribution (uniform on $C_r$) in TV distance but their parameters are far from each other! We will prove that this is not possible because:

\begin{center}
\emph{Reasonable distributions close in TV distance have close parameters. }
\end{center}

It is important to observe that such a statement is false even for subgaussian distributions - indeed, moment upper bounds (such as those that follow from subgaussianity) are simply not enough to give \emph{any} bound on the parameter distance of TV-close pairs at all. See Section~\ref{sec:subgaussian-no-TV-vs-param} for a simple proof. As might be apparent from the example in Section~\ref{sec:subgaussian-no-TV-vs-param}, anti-concentration (and the consequent moment lower bound) is crucial to prove such a statement.

There's a lot of work in statistics that proves such statements for natural families of distributions such as Gaussians (see for e.g.~\cite{devroye2018total}). In fact, all works that design outlier-robust estimation algorithms in the strong contamination model implicitly prove such a statement. This connection is made explicit in the work on robust moment estimation~\cite{DBLP:journals/corr/abs-1711-11581}. Our setting, however, differs from these works because we deal with the regime where the TV distance is close to $1$ (in contrast to the setting where TV distance is close to $0$ in the above works) outlier-robust estimation. See Section~\ref{sec:subgaussian-no-TV-vs-param} for an effect of the TV distance on our simple example.

For the special case of Gaussians, proving such a statement even for the regime where TV distance happens to be $\sim 1$ turns out to be elementary (see Proposition~\ref{prop:tv-vs-param-for-gaussians}). However such a proof, because it uses the PDF of the distribution heavily is unlikely to be expressible in low-degree sum-of-squares proof system - a key necessity for our algorithmic application.

But perhaps even more importantly, the proof for the Gaussian case above is opaque and doesn't reveal what properties of the distribution come into play for such a statement to be true. We show that the statement above holds for all hyper-contractive and anti-concentrated distributions. As a result, we obtain both, an argument that applies to more general class of distributions and a proof translatable (with some effort) into low-degree sum-of-squares proof system.

\paragraph{Proof Idea: Proving TV vs Parameter Distance Bounds via Variance Mismatch}
We will prove the TV vs parameter distance relationships for reasonable distributions by giving a low-degree sum-of-squares proof of the statement in the contrapositive form. In this form, the result informally says that if $\hat{C}$ (indicated by $w$) that defines a reasonable distribution cannot simultaneously have large intersections with two well-separated, reasonable distributions $C_r$ and $C_{r'}$. That is, the product $w(C_r)w(C_{r'})$ must be small.

To prove such a statement, we deal with each of the three ways (see Def~\ref{def:separated-mixture-model}) $C_r,C_{r'}$ can be separated one by one. In each of these cases, we will find a degree 2 polynomial in $x \sim \hat{C}$ (the purported cluster) that simultaneously has high variance if $w(C_r)$ and $w(C_{r'})$ are both large (since $C_r$ and $C_{r'}$ are separated). On the other hand, we will also show that for certifiably hyper-contractive $\hat{C}$, the polynomial above cannot have too large a variance. Taken together, these two statement yield a bound on the product $w(C_r)w(C_{r'})$.

In the following, we discuss the ideas that go into proving such statements for each of the three kinds of parameter separation.
We will also briefly discuss two basic additions to sum-of-squares toolkit that allow us to translate this proof into the low-degree SoS proof system. It turns out that the ``hardest'' case to deal is that of spectral separation. 

\subsection{Simultaneous intersection bounds from spectral separation} \label{sec:spectral-separation-overview}
For the purpose of this discussion, assume that the means $\mu(r) = \mu(r') = 0$. Since $C_r$ and $C_{r'}$ are spectrally separated, there exists a unit vector $v$ such that $\Delta_{\spe} v^{\top} \Sigma(r) v \leq  v^{\top} \Sigma(r') v$. We will use the polynomial $\iprod{x,v}^2$ for this $v$ as our ``mismatch'' marker as discussed above.

The key idea of the proof is to show that if $w(C_r)$ and $w(C_{r'})$ are simultaneously large, then, because of the stark difference in the behavior of $C_r$ and $C_{r'}$ in direction $v$, the degree 2 polynomial $\iprod{x,v}^2$ for $x \sim\hat{C}$ must have a large variance. We will prove this statement by using anti-concentration of $C_r$ and $C_{r'}$. On the other hand, we will show that since $\hat{C}$ is also anti-concentrated, $\Iprod{x,v}^2$ for $x \sim \hat{C}$ cannot have \emph{too} large a variance. Stringing together these bounds should, in principle, give us upper bound on $w(C_r)w(C_{r'})$.

While we manage to prove both the statements above via low-degree SoS proofs, putting them together turns out to be involved. It's easy to do this via a ``real-world'' argument. However, such a proof relies on case analysis that doesn't appear easy to SoSize. This is where we incur a dependence on the spread parameter $\kappa$. We explain these steps in more detail next.

\paragraph{Lower-Bound on the variance (Lemma~\ref{lem:spectral-lower-bound-intersection-clustering}).}
We start by considering (the reason will become clear in a moment) the random variable $z-z'$ where $z,z' \sim \hat{C}$ are independent uniform draws. Then, it's easy to compute that $z-z'$ has mean $0$ and covariance $2 \Sigma(w)$. Thus, in order to lower bound $v^{\top}\Sigma(w)v$, we can consider the polynomial $\E_{z,z' \sim \hat{C}} \iprod{z-z',v}^2$.

Here's the simple but important observation (and our reason for looking at $z-z'$).
With probability $w(C_r)$, $z \in C_r$ and with probability $w(C_{r'})$, $z' \in C_{r'}$.
Thus, $w(C_r)w(C_{r'})$ fraction of samples $z-z'$ from $\hat{C}$ are differences of independent samples from $C_r$ and $C_{r'}$.

Let's now understand the distribution of differences of independent samples from $C_r$ and $C_{r'}$.
The covariance of this distribution is $\Sigma(r) + \Sigma(r')$.
Further, since each of $C_r$ and $C_{r'}$ are anti-concentrated, so is the convolution obtained by taking differences of independent samples from $C_r$ and $C_{r'}$. Thus, $z-z'$ takes a value $\leq \delta \sqrt{v^{\top} \Paren{\Sigma(r) + \Sigma(r')}v}$ with probability at most $\sim \delta$. Thus, the contribution of $z-z'$ to $v^{\top} \Sigma(w)v$, when it's larger than the above bound, should be at least $\geq \Paren{w(C_r)w(C_{r'}) - \delta} \delta^2 v^{\top} \Paren{\Sigma(r) + \Sigma(r')}v \geq \delta^2 v^{\top} \Sigma(r')v$.


\paragraph{Upper bound on variance (Lemma~\ref{lem:anti-conc-upper-bound-w}).}
The main idea is to again rely on anti-concentration - but this time of $\hat{C}$ which is enforced by our constraint system $\cA$.
Now, we know that with $w(C_r)$ probability, $\hat{C}$ outputs a point from $C_r$.
Since these points are in $C_r$, their contribution to the variance of $\hat{C}$ cannot be larger than $v^{\top}\Sigma(r)v$. On the other hand, since $\hat{C}$ is anti-concentrated, the contribution to the variance of $\hat{C}$ from points shared with $C_r$ must be comparable to that of $\hat{C}$ if $w(C_r)$ is large. Stringing together these observations allows us to conclude that when $w(C_r)$ is large, $v^{\top}\Sigma(w)v$ must be comparable to $v^{\top}\Sigma(r)v$.

\paragraph{Combining Upper and Lower Bounds: Real Life vs SoS, dependence on $\kappa$.}
Observe that the first claim above showed a lower bound on $v^{\top}\Sigma(w)v$ in terms of $v^{\top}\Sigma(r')v$ when $w(C_r)w(C_{r'})$ is large. The second claim shows an upper bound on $v^{\top}\Sigma(w)v$ (when  $w(C_r)$ is large) in  terms of $v^{\top}\Sigma(r)v$. Combining this with the spectral separation condition $\Delta_{\spe} v^{\top}\Sigma(r)v \leq v^{\top}\Sigma(r')v$ should immediately yield a bound on $w(C_r)w(C_{r'})$.

This argument indeed can be done easily in ``real-world''("high"-degree SoS, see Lemma~\ref{lem:real-world-argument}) and complete the proof of Lemma~\ref{lem:intersection-bounds-from-spectral-separation}. However, the proof involves a case-analysis based on when $w(C_r) > \delta$ vs $w(C_r) \leq \delta$ separately. Such a case analysis appears hard to perform with a low-degree SoS proof. 

A natural strategy to do this in SoS requires, in addition, a ``rough'' bound on $v^{\top} \Sigma(w)v$.
We obtain this bound (Lemma~\ref{lem:rough-bound-spectral-upper}), again, by relying on anti-concentration of $\hat{C}$.
This rough bound essentially allows us to bound $v^{\top}\Sigma(w)v$ by (some multiple of) the maximum of $v^{\top} \Sigma(r)v$ as $r$ ranges over all the $k$ clusters.

\paragraph{The case of $k=2$ vs $k>2$.}
For the case of $k=2$, the rough bound above depends only on the clusters we are dealing with (since there are only two of them) and leads to a proof without any dependence on $\kappa$. For the case of $k>2$, however, the rough bound depends on $v^{\top} \Sigma(i)v$ for clusters $C_i$ for $i \not \in \{r,r'\}$ - the set we are currently dealing with and, in principle, could be arbitrarily large. We use our assumption on the \emph{spread} of the mixture to control $v^{\top} \Sigma(i)v$ for all such $i \not \in \{r,r'\}$.

\paragraph{Using uniform approximators for thresholds over $[0,1]$.}
A naive argument implementing the above reasoning loses a polynomial factor in $\kappa$ in the exponent. We lessen the blow by a technical trick using uniform polynomial approximators for thresholds (Lemma~\ref{lem:poly-approximate-threshold}) over the unit interval. We construct such polynomial by relying on standard tools from approximation theory in Section~\ref{sec:poly-approx-threshold} of the Appendix. These polynomials allow us to capture the conditional reasoning in the real-world proof above with a low-loss leading to a logarithmic dependence on the SoS degree on $\kappa$.

\subsection{Intersection Bounds from Relative Frobenius Separation}
Obtaining intersection bounds from mean separation turns out to be relatively stress free and uses ideas similar to the ones discussed in the spectral separation case above. So we move on to the case of Relative Frobenius separation here. For the sake of exposition here, we assume $\mu(r),\mu(r') = 0$ as before and set $\Sigma(r')=I$. Then, relative Frobenius  separation guarantees us that $\Norm{\Sigma(r)-I}_F^2 \geq \Delta_{cov}^2$.

Let's understand what happens to $\E_{\hat{C}}Q(x)$ - the expectation of this polynomial over the purported cluster $\hat{C}$ if it has a large intersection with both $C_r$ and $C_{r'}$.

\paragraph{Lower Bound on the  Variance of Q (Lemma~\ref{lem:Large-Intersection-Implies-High-Variance-frobenius}).}
Consider the polynomial $Q(x) = x^{\top}Qx$ for $Q = \Sigma(r) -I$.
By direct computation, the expectation of this polynomial on $C_r$ equals $\Norm{\Sigma(r)-I}_F^2 + \tr(\Sigma(r)-I)$.
While the expectation on $C_{r'}$ equals $\tr(\Sigma(r)-I)$.

Using \emph{hypercontractivity} of degree 2 polynomials over $C_r$ and $C_{r'}$, we show that the variance of the polynomial $Q(x)$ on $C_r$ and $C_{r'}$ is $\ll \Delta_{cov}^2$. Thus, on $\hat{C}$, for a $w(C_r)$ fraction of points $Q(x)$ would be $\approx \Norm{\Sigma(r)-I}_F^2 + \tr(\Sigma(r)-I)$ while for a $w(C_{r'})$ fraction of points, $Q(x)$ would be $\approx \tr(\Sigma(r)-I$. The difference in these values is  $|\E_{x\sim C_r}Q(x)-\E_{x \sim C_{r'}} Q(x)| = \Norm{\Sigma(r)-I}_F^2\geq\Delta_{cov}^2$. 
Thus, if $w(C_r)w(C_{r'})$ is large, $Q(x)$ must have a variance comparable to $w(C_r)w(C_{r'})\Delta_{cov}^2$ on $\hat{C}$. Thus, we expect that if $\hat{C}$ picks a significant mass from both $C_r$ and $C_{r'}$, then, $Q(x)$ must have a large variance on $\hat{C}$.

\paragraph{Upper Bound on the Variance of Q via SoSizing Contraction Lemma (Lemma~\ref{lem:variance-upper-bound-hypercontractivity}).}
In contrast to the the case of mean separation where we relied on anti-concentration of $\hat{C}$, we prove an upper bound on the variance of  $Q$ by relying on hypercontractivity of degree 2 polynomials of $\hat{C}$. A key step in this proof relies on \emph{SoSizing} a basic matrix inequality: For all $d \times d$ matrices $A,B$, $\Norm{AB}_F^2 \leq \Norm{A}_{op}^2 \Norm{B}_F^2$.



\subsection{Outlier-Robust Variant}
Making the algorithm in the discussion above outlier-robust is relatively straightforward. Observe that in this case, we do not get access to the original sample $X$ as above. Instead, we get an $\epsilon$-corruption of $X$, say $Y$ as input. Our goal is to give a clustering of $Y$ that corresponds to the clustering $X$ with at most $O(k\epsilon)$ points misclassified in any given cluster. Observe that this is the information-theoretically the best possible result we can expect since all the $\epsilon n$ outliers could end up corrupting a single chosen  true cluster.

Our key idea here is to introduce a new collection of variables $X'$ that ``guess'' the original sample  that generated $Y$. We add the constraint that $X$ and $Y$ intersect in $(1-\epsilon)$-fraction of the points to capture the only property of $X$ that we know.

We then use a version of the system of constraints $\cA$ with $X$ replaced by $X'$. Let $C'_1, C'_2,\ldots,C'_k$ be the clusters induced by taking the points with the same indices as in $C_i$ from $X'$. Note that in this case, $X'$ and $C'_i$s are indeterminates in our constraint system. Our proof from the previous section  generalizes with only a few changes to yield simultaneous intersection bounds on $w'(C_r')w'(C_{r'}')$. The intersection bounds with $Y$ then  follow by noting a (degree 2 SoS proof of) $|C'_i \cap C_i| \geq (1-2k\epsilon)|C_i|$.

\subsection{Recursive Clustering Algorithm}
\paragraph{Simple rounding with larger running time.}
Given our certifiability proofs that prove upper bounds on simultaneous intersection of $\hat{C}$ with true clusters, one can immediately obtain an algorithm for clustering mixtures of reasonable distributions that runs in time $n^{O(s(\poly(\eta/k)) \log (\kappa))}$ and obtain Theorems~\ref{thm:main-clustering-section} and~\ref{thm:main-robust-clustering-section}. These algorithms work by computing a pseudo-distribution $\tzeta$ on $w$ (the indicator of samples in $\hat{C}$) and rounding it. For the purpose of this overview, it is helpful to think of pseudo-distributions as giving us access to low-degree moments of a distribution on $w$ that satisfies the checks that we made (certifiable anti-concentration and hypercontractivity) in our certifiability proofs above. A pseudo-distribution of degree $t$ in $n$ variables can be computed in time $n^{O(t)}$ via semidefinite programming and satisfies all inequalities that can be derived from our checks (constraint system) via low-degree SoS proofs.

Our rounding algorithm is simple and is the same as the one described in Section 4.3 of the monograph~\cite{TCS-086} that gives a simpler proof of the recently obtained algorithm for clustering spherical mixtures~\cite{KothariSteinhardt17,HopkinsLi17}. We use the simultaneous intersection bounds to derive that the second moment matrix $\pE_{\tzeta}[ww^{\top}]$ of $w$ (indicating $\hat{C}$) is approximately block diagonal, with approximate clusters as blocks. This allows us to iteratively \emph{peel off} approximate clusters greedily - see the proof of Theorem~\ref{thm:main-clustering-section}. To establish this block diagonal structure our proof requires the pseudo-distribution to have a degree that scales with $\log \kappa$ where $\kappa$ is the spread of the mixture.

\paragraph{Spread-independent recursive rounding.}
In Section~\ref{sec:better-algo}, we give a more sophisticated rounding with a running time that does not depend on the spread $\kappa$. The conceptual idea behind this rounding is based on two curious facts that we establish in Section~\ref{sec:better-algo}.

\begin{enumerate}
	\item \emph{Simple rounding has non-trivial information at constant degrees}.
The first fact (Lemma~\ref{lem:partial-cluster-recovery}) shows that when we run the simple rounding with a pseudo-distribution $\tzeta$ of degree that \emph{does not} grow with $\log \kappa$, we can still prove that $\pE_{\tzeta}[ww^{\top}]$ has a \emph{partial block diagonal} structure. This structure allows us to prove that for the clustering $\hat{C}_1, \hat{C}_2, \ldots, \hat{C}_k$ output by our simple rounding above, there exists a (non-trivial) partition $S \cup L = [k]$ such that both $\cup_{i\in S} \hat{C}_i$ and $\cup_{j \in T} \hat{C}_j$ are  essentially unions of the true clusters.

The proof relies on two facts: 1) if no pair of components of the input mixture are spectrally separated, then, the spread $\kappa$ is small so our simple rounding already works. 2) Even when there's a pair of components that are spectrally separated, the SoS degree required in our simultaneous intersection bounds can be much smaller than $\kappa$. Concretely, our analysis in Lemma~\ref{lem:intersection-bounds-from-spectral-separation} and ~\ref{lem:intersection-bounds-from-spectral-separation-robust} yields a degree that scales with $\frac{v^{\top}\Sigma(i)v}{v^{\top}\Sigma(r')v}$ that we loosely upper bound by $\kappa  = \max_{i,j} \frac{v^{\top}\Sigma(i)v}{v^{\top}\Sigma(j)v}$. If $v^{\top}\Sigma(r')v$ is comparable to $\max_{i\leq k} v^{\top}\Sigma(i)v$, then, the SoS degree of the proof is much smaller than $\kappa$. We use this observation to show that there's a $S \subseteq [k]$ and a $O(1)$ degree SoS proof that bounds the simultaneous intersection of $\hat{C}$ with true clusters $C_i$ and $C_j$ whenever $i\in S$ and $j \not \in  S$. This is enough to obtain a \emph{partial cluster recovery} guarantee.

Thus, $\cup_{i\in S} \hat{C}_i$ can be treated as a mixture of ( $< k$) components along with a small fraction of outliers and we can recurse. Of  course, we do not know $S$, so our algorithm tries all the $2^k$ possible choices and recursively tries to cluster them.

\item \emph{Verifying clusters requires only constant-degree pseudo-distributions} (Lemma~\ref{lem:soundness-verification-main}). In order to run the recursive clustering algorithm suggested above, we need a subroutine that can efficiently verify that a given purported cluster is close to a true one.  While we cannot show that degree $O(s(\poly(\eta/k)))$ pseudo-distributions are enough to \emph{find} a clustering, we will prove that they are enough to \emph{verify} a purported clustering. Concretely, given a purported cluster $\hat{C}$, we show that there's a pseudo-distribution of constant degree (independent of $\kappa$) consistent with verification constraints (see Section~\ref{sec:verification}) iff $\hat{C}$ is close to a true cluster.

The ``completeness'' of the verification algorithm is easy to prove. The meat of the analysis is proving soundness - i.e. if a purported cluster $\hat{C}$ has an appreciable  intersection  with two different true clusters, then the  verification algorithm must output reject. 

A priori, such a result can appear a bit confusing - after all, we just spent most of this overview arguing SoS proofs of degree that grow with $\kappa$ for verifying purported clusters. The key technical difference (quite curious from a proof complexity perspective) is that in the setting of verification, we are trying to derive a contradiction from the assumption that the intersection bounds are simultaneously large for two distinct  true clusters. While in the simultaneous intersection bounds, the goal is similar statement but stated in terms of the  \emph{contrapositive}.
\end{enumerate}

\subsection{Covariance Estimation in Relative Frobenius Error}
Tools in this paper allow us to get an additional application - an outlier-robust algorithm to compute the covariance of a distribution with optimal \emph{relative Frobenius error}. Prior works~\cite{DBLP:conf/focs/LaiRV16,DBLP:journals/corr/abs-1711-11581} gave guarantees for covariance estimation in spectral distance (which implies only dimension dependent bounds on the relative Frobenius error) or worked only for Gaussian distributions~\cite{DBLP:conf/focs/DiakonikolasKK016}. We show an optimal $\tilde{O}(\epsilon)$ (independent of the dimension) error guarantee on relative Frobenius error in the presence of an $\epsilon$-fraction adversarial outliers whenever the target distribution is certifiably hypercontractive. Our algorithm is same as the one used in~\cite{DBLP:journals/corr/abs-1711-11581} but our analysis relies on certifiable hypercontractivity along with the SoS contraction lemma discussed above.

As a corollary of this result, we can take an accurate clustering output by our clustering algorithms for reasonable distributions and use our covariance estimation algorithm here to get statistically optimal estimates of mean and covariance in the distances presented in Definition~\ref{def:separated-mixture-model} thus obtaining outlier-robust parameter estimation algorithms from our outlier-robust clustering algorithm. 


\section{Preliminaries}
\label{sec:preliminaries}



Throughout this paper, for a vector $v$, we use $\norm{v}_2$ to denote the Euclidean norm of $v$. For a $n \times m$ matrix $M$, we use $\norm{M}_2 = \max_{\norm{x}_2=1} \norm{Mx}_2$ to denote the spectral norm of $M$ and $\norm{M}_F = \sqrt{\sum_{i,j} M_{i,j}^2}$ to denote the Frobenius norm of $M$. For symmetric matrices we use $\succeq$ to denote the PSD/Löwner ordering over eigenvalues of $M$. For a $n \times n$, rank-$r$ symmetric matrix $M$, we use $U\Lambda U^{\top}$ to denote the Eigenvalue Decomposition, where $U$ is a $n \times r$ matrix with orthonormal columns and $\Lambda$ is a $r \times r$ diagonal matrix denoting the eigenvalues. We use $M^{\dagger} = U \Lambda^{\dagger} U^{\top} $ to denote the Moore-Penrose pseudoinverse, where $\Lambda^{\dagger}$ inverts the non-zero eigenvalues of $M$. If $M \succeq 0$, we use $M^{\dagger/2} = U \Lambda^{\dagger/2} U^{\top}$ to denote taking the square-root of the non-zero eigenvalues. We use $\Pi = U U^{\top}$ to denote the Projection matrix corresponding to the column/row span of $M$. Since $\Pi = \Pi^2$, the pseudo-inverse of $\Pi$ is itself, i.e. $\Pi^{\dagger} = \Pi$.

\begin{definition}[$\sigma$-Sub-gaussian Distribution]
A random variable $x$ is drawn from a $\sigma$-Sub-gaussian distribution if for all $t\geq0$, 
$\Pr\left[|x| \geq t \right] \leq 2\exp(-t^2/\sigma^2)$.
\end{definition}

We work with $1$-Sub-gaussian distributions unless otherwise specified and drop the $1$ when clear from context. 

\paragraph{Probability Preliminaries.} We begin with standard convergence results for mean and covariance.

\begin{fact}[Empirical Mean for Sub-gaussians]
\label{fact:empirical_mean}
Let $\cD$ be a Sub-gaussian distribution on $\R^d$ with mean $\mu$ and covariance $\Sigma$ and let $x_1, x_2, \ldots x_n \sim \cD$. Then, with probability $1-\delta$,
\[
\Norm{ \frac{1}{n}\sum^n_{i=1} x_i - \mu }_2 \leq \sqrt{\frac{\Tr(\Sigma)}{n}} + \sqrt{\frac{\|\Sigma\|_2 \log(1/\delta)}{n}}
\]
\end{fact}

\begin{fact}[Empirical Covariance for Sub-gaussians, Proposition 2.1 \cite{MR2956207-Vershynin12}]
\label{fact:empirical_cov}
Let $\cD$ be a Sub-gaussian distribution on $\mathbb{R}^d$ with mean $\mu$ and covariance $\Sigma$ and let $x_1, x_2, \ldots x_n \sim \cD$. Then, with probability $1-\delta$,
\[
\Norm{ \frac{1}{n}\sum^n_{i=1} x_i x_i^{\top} - \Sigma }_2 \leq c \left(\sqrt{\frac{d}{n}} + \sqrt{ \frac{\log(1/\delta)}{n}}\right)
\]
\end{fact}

\begin{definition}[Hellinger Distance]
For probability distribution $p,q$ on  $\R^d$, let 
\[
h(P,Q) = \frac{1}{\sqrt{2}} \sqrt{ \int_{\R^d} \left(\sqrt{p(x)} - \sqrt{q(x)}\right)^2 dx}
\] 
be the Hellinger distance between them. 
\end{definition}

\begin{remark}
Hellinger distance between $p,q$ satisfies: $h(p,q)^2 \leq \dtv(p,q) \leq h(p,q)\sqrt{2 -h(p,q)^2}$.  
\end{remark}

\begin{fact}[Hellinger Distance between Gaussians] \label{fact:hellinger-gaussians}
\begin{equation*}
h(\cN(\mu,\Sigma), \cN(\mu',\Sigma'))^2 =1 - \frac{\det(\Sigma)^{1/4} \det(\Sigma')^{1/4}}{\det\left( \frac{\Sigma+\Sigma'}{2}\right)^{\frac{1}{2}}} \exp\left(-\frac{1}{8} (\mu - \mu)^{\top}\Paren{\frac{\Sigma + \Sigma'}{2}}^{-1} (\mu - \mu') \right)
\end{equation*}
\end{fact}


Next, we define pseudo-distributions and sum-of-squares proofs. Detailed exposition of the sum-of-squares method and its usage in average-case algorithm design can be found in ~\cite{TCS-086} and the lecture notes~\cite{BarakS16}.

Let $x = (x_1, x_2, \ldots, x_n)$ be a tuple of $n$ indeterminates and let $\R[x]$ be the set of polynomials with real coefficients and indeterminates $x_1,\ldots,x_n$. We say that a polynomial $p\in \R[x]$ is a \emph{sum-of-squares (sos)} if there exist polynomials $q_1,\ldots,q_r$ such that $p=q_1^2 + \cdots + q_r^2$.

\subsection{Pseudo-distributions}

Pseudo-distributions are generalizations of probability distributions.
We can represent a discrete (i.e., finitely supported) probability distribution over $\R^n$ by its probability mass function $D\from \R^n \to \R$ such that $D \geq 0$ and $\sum_{x \in \mathrm{supp}(D)} D(x) = 1$.
Similarly, we can describe a pseudo-distribution by its mass function by relaxing the constraint $D\ge 0$ to passing certain low-degree non-negativity tests.

Concretely, a \emph{level-$\ell$ pseudo-distribution} is a finitely-supported function $D:\R^n \rightarrow \R$ such that $\sum_{x} D(x) = 1$ and $\sum_{x} D(x) f(x)^2 \geq 0$ for every polynomial $f$ of degree at most $\ell/2$.
(Here, the summations are over the support of $D$.)
A straightforward polynomial-interpolation argument shows that every level-$\infty$-pseudo distribution satisfies $D\ge 0$ and is thus an actual probability distribution.
We define the \emph{pseudo-expectation} of a function $f$ on $\R^d$ with respect to a pseudo-distribution $D$, denoted $\pE_{D(x)} f(x)$, as
\begin{equation}
  \pE_{D(x)} f(x) = \sum_{x} D(x) f(x) \,\mper
\end{equation}
The degree-$\ell$ moment tensor of a pseudo-distribution $D$ is the tensor $\E_{D(x)} (1,x_1, x_2,\ldots, x_n)^{\otimes \ell}$.
In particular, the moment tensor has an entry corresponding to the pseudo-expectation of all monomials of degree at most $\ell$ in $x$.
The set of all degree-$\ell$ moment tensors of probability distribution is a convex set.
Similarly, the set of all degree-$\ell$ moment tensors of degree $d$ pseudo-distributions is also convex.
Unlike moments of distributions, there's an efficient separation oracle for moment tensors of pseudo-distributions.

\begin{fact}[\cite{MR939596-Shor87,parrilo2000structured,MR1748764-Nesterov00,MR1846160-Lasserre01}]
  \label[fact]{fact:sos-separation-efficient}
  For any $n,\ell \in \N$, the following set has a $n^{O(\ell)}$-time weak separation oracle (in the sense of \cite{MR625550-Grotschel81}):
  \begin{equation}
    \Set{ \pE_{D(x)} (1,x_1, x_2, \ldots, x_n)^{\otimes d} \mid \text{ degree-d pseudo-distribution $D$ over $\R^n$}}\,\mper
  \end{equation}
\end{fact}
This fact, together with the equivalence of weak separation and optimization \cite{MR625550-Grotschel81} allows us to efficiently optimize over pseudo-distributions (approximately)---this algorithm is referred to as the sum-of-squares algorithm. The \emph{level-$\ell$ sum-of-squares algorithm} optimizes over the space of all level-$\ell$ pseudo-distributions that satisfy a given set of polynomial constraints (defined below).

\begin{definition}[Constrained pseudo-distributions]
  Let $D$ be a level-$\ell$ pseudo-distribution over $\R^n$.
  Let $\cA = \{f_1\ge 0, f_2\ge 0, \ldots, f_m\ge 0\}$ be a system of $m$ polynomial inequality constraints.
  We say that \emph{$D$ satisfies the system of constraints $\cA$ at degree $r$}, denoted $D \sdtstile{r}{} \cA$, if for every $S\subseteq[m]$ and every sum-of-squares polynomial $h$ with $\deg h + \sum_{i\in S} \max\set{\deg f_i,r}$, $\pE_{D} h \cdot \prod _{i\in S}f_i  \ge 0$.

  We write $D \sdtstile{}{} \cA$ (without specifying the degree) if $D \sdtstile{0}{} \cA$ holds.
  Furthermore, we say that $D\sdtstile{r}{}\cA$ holds \emph{approximately} if the above inequalities are satisfied up to an error of $2^{-n^\ell}\cdot \norm{h}\cdot\prod_{i\in S}\norm{f_i}$, where $\norm{\cdot}$ denotes the Euclidean norm\footnote{The choice of norm is not important here because the factor $2^{-n^\ell}$ swamps the effects of choosing another norm.} of the coefficients of a polynomial in the monomial basis.
\end{definition}

We remark that if $D$ is an actual (discrete) probability distribution, then we have  $D\sdtstile{}{}\cA$ if and only if $D$ is supported on solutions to the constraints $\cA$. We say that a system $\cA$ of polynomial constraints is \emph{explicitly bounded} if it contains a constraint of the form $\{ \|x\|^2 \leq M\}$.
The following fact is a consequence of \cref{fact:sos-separation-efficient} and \cite{MR625550-Grotschel81},

\begin{fact}[Efficient Optimization over Pseudo-distributions]
There exists an $(n+ m)^{O(\ell)} $-time algorithm that, given any explicitly bounded and satisfiable system\footnote{Here, we assume that the bit complexity of the constraints in $\cA$ is $(n+m)^{O(1)}$.} $\cA$ of $m$ polynomial constraints in $n$ variables, outputs a level-$\ell$ pseudo-distribution that satisfies $\cA$ approximately. \label{fact:eff-pseudo-distribution}
\end{fact}
\paragraph{Basic Facts about Pseudo-Distributions.}

We will use the following Cauchy-Schwarz inequality for pseudo-distributions:
\begin{fact}[Cauchy-Schwarz for Pseudo-distributions]
Let $f,g$ be polynomials of degree at most $d$ in indeterminate $x \in \R^d$. Then, for any degree d pseudo-distribution $\tzeta$,
$\pE_{\tzeta}[fg] \leq \sqrt{\pE_{\tzeta}[f^2]} \sqrt{\pE_{\tzeta}[g^2]}$.
 \label{fact:pseudo-expectation-cauchy-schwarz}
\end{fact}

\begin{fact}[Hölder's Inequality for Pseudo-Distributions] \label{fact:pseudo-expectation-holder}
Let $f,g$ be polynomials of degree at most $d$ in indeterminate $x \in \R^d$. 
Fix $t \in \N$. Then, for any degree $dt$ pseudo-distribution $\tzeta$,
$\pE_{\tzeta}[f^{t-1}g] \leq \paren{\pE_{\tzeta}[f^t]}^{\frac{t-1}{t}} \paren{\pE_{\tzeta}[g^t]}^{1/t}$.
\end{fact}

\begin{corollary}[Comparison of Norms]
\label{fact:comparison-of-pseudo-expectation-norms}
Let $\tzeta$ be a degree $t^2$ pseudo-distribution over a scalar indeterminate $x$. Then, $\pE[x^{t}]^{1/t} \geq \pE[x^{t'}]^{1/t'}$ for every $t' \leq t$. 
\end{corollary}

\paragraph{Reweighting Pseudo-Distributions}

The following fact is easy to verify and has been used in several works (see~\cite{DBLP:conf/stoc/BarakKS17} for example).  
\begin{fact}[Reweightings] \label{fact:reweightings}
Let $\tmu$ be a pseudo-distribution of degree $k$ satisfying a set of polynomial constraints $\cA$ in variable $x$. 
Let $p$ be a sum-of-squares polynomial of degree $t$ such that $\pE[p(x)] \neq 0$.
Let $\tmu'$ be the pseudo-distribution defined so that for any polynomial $f$, $\pE_{\tmu'}[f(x)] = \pE_{\tmu}[ f(x)p(x)]/\pE_{\tmu}[p(x)]$. Then, $\tmu'$ is a pseudo-distribution of degree $k-t$ satisfying $\cA$. 
\end{fact}

\subsection{Sum-of-squares proofs}

Let $f_1, f_2, \ldots, f_r$ and $g$ be multivariate polynomials in $x$.
A \emph{sum-of-squares proof} that the constraints $\{f_1 \geq 0, \ldots, f_m \geq 0\}$ imply the constraint $\{g \geq 0\}$ consists of  polynomials $(p_S)_{S \subseteq [m]}$ such that
\begin{equation}
g = \sum_{S \subseteq [m]} p_S \cdot \Pi_{i \in S} f_i
\mper
\end{equation}
We say that this proof has \emph{degree $\ell$} if for every set $S \subseteq [m]$, the polynomial $p_S \Pi_{i \in S} f_i$ has degree at most $\ell$.
If there is a degree $\ell$ SoS proof that $\{f_i \geq 0 \mid i \leq r\}$ implies $\{g \geq 0\}$, we write:
\begin{equation}
  \{f_i \geq 0 \mid i \leq r\} \sststile{\ell}{}\{g \geq 0\}
  \mper
\end{equation}
For all polynomials $f,g\colon\R^n \to \R$ and for all functions $F\colon \R^n \to \R^m$, $G\colon \R^n \to \R^k$, $H\colon \R^{p} \to \R^n$ such that each of the coordinates of the outputs are polynomials of the inputs, we have the following inference rules.

The first one derives new inequalities by addition/multiplication:
\begin{equation} \label{eq:sos-addition-multiplication-rule}
\frac{\cA \sststile{\ell}{} \{f \geq 0, g \geq 0 \} } {\cA \sststile{\ell}{} \{f + g \geq 0\}}, \frac{\cA \sststile{\ell}{} \{f \geq 0\}, \cA \sststile{\ell'}{} \{g \geq 0\}} {\cA \sststile{\ell+\ell'}{} \{f \cdot g \geq 0\}}\mper
\end{equation}
The next one derives new inequalities by transitivity: 
\begin{equation} \label{eq:sos-transitivity}
\frac{\cA \sststile{\ell}{} \cB, \cB \sststile{\ell'}{} C}{\cA \sststile{\ell \cdot \ell'}{} C}\mcom 
\end{equation}
Finally, the last rule derives new inequalities via substitution:
\begin{equation} \label{eq:sos-substitution}
\frac{\{F \geq 0\} \sststile{\ell}{} \{G \geq 0\}}{\{F(H) \geq 0\} \sststile{\ell \cdot \deg(H)} {} \{G(H) \geq 0\}}\tag{substitution}\mper
\end{equation}

Low-degree sum-of-squares proofs are sound and complete if we take low-level pseudo-distributions as models.
Concretely, sum-of-squares proofs allow us to deduce properties of pseudo-distributions that satisfy some constraints.
\begin{fact}[Soundness]
  \label{fact:sos-soundness}
  If $D \sdtstile{r}{} \cA$ for a level-$\ell$ pseudo-distribution $D$ and there exists a sum-of-squares proof $\cA \sststile{r'}{} \cB$, then $D \sdtstile{r\cdot r'+r'}{} \cB$.
\end{fact}
If the pseudo-distribution $D$ satisfies $\cA$ only approximately, soundness continues to hold if we require an upper bound on the bit-complexity of the sum-of-squares $\cA \sststile{r'}{} B$  (number of bits required to write down the proof). In our applications, the bit complexity of all sum of squares proofs will be $n^{O(\ell)}$ (assuming that all numbers in the input have bit complexity $n^{O(1)}$). This bound suffices in order to argue about pseudo-distributions that satisfy polynomial constraints approximately.

The following fact shows that every property of low-level pseudo-distributions can be derived by low-degree sum-of-squares proofs.
\begin{fact}[Completeness]
  \label{fact:sos-completeness}
  Suppose $d \geq r' \geq r$ and $\cA$ is a collection of polynomial constraints with degree at most $r$, and $\cA \vdash \{ \sum_{i = 1}^n x_i^2 \leq B\}$ for some finite $B$.

  Let $\{g \geq 0 \}$ be a polynomial constraint.
  If every degree-$d$ pseudo-distribution that satisfies $D \sdtstile{r}{} \cA$ also satisfies $D \sdtstile{r'}{} \{g \geq 0 \}$, then for every $\epsilon > 0$, there is a sum-of-squares proof $\cA \sststile{d}{} \{g \geq - \epsilon \}$.
\end{fact}

\paragraph{Basic Sum-of-Squares Proofs}

\begin{fact}[Operator norm Bound]
\label{fact:operator_norm}
Let $A$ be a symmetric $d\times d$ matrix and $v$ be a vector in $\mathbb{R}^d$. Then,
\[
\sststile{2}{v} \Set{ v^{\top} A v \leq \|A\|_2\|v\|^2_2 }
\]
\end{fact}

\begin{fact}[SoS Hölder's Inequality] \label{fact:sos-holder}
Let $f_i,g_i$ for $1 \leq i \leq s$ be indeterminates. 
Let $p$ be an even positive integer. 
Then, 
\[
\sststile{p^2}{f,g} \Set{  \Paren{\frac{1}{s} \sum_{i = 1}^s f_i g_i^{p-1}}^{p} \leq \Paren{\frac{1}{s} \sum_{i = 1}^s f_i^p}^q \Paren{\frac{1}{s} \sum_{i = 1}^s g_i^p}^{p-1}}\mper
\]
\end{fact}
Observe that using $p = 2$ yields the SoS Cauchy-Schwarz inequality. 

\begin{fact}[SoS Almost Triangle Inequality] \label{fact:sos-almost-triangle}
Let $f_1, f_2, \ldots, f_r$ be indeterminates. Then,
\[
\sststile{2t}{f_1, f_2,\ldots,f_r} \Set{ \Paren{\sum_{i\leq r} f_i}^{2t} \leq r^{2t-1} \Paren{\sum_{i =1}^r f_i^{2t}}}\mper
\]
\end{fact}

\begin{fact}[SoS AM-GM Inequality, see Appendix A of~\cite{MR3388192-Barak15}] \label{fact:sos-am-gm}
Let $f_1, f_2,\ldots, f_m$ be indeterminates. Then, 
\[
\sststile{m}{f_1, f_2,\ldots, f_m} \Set{ \Paren{\frac{1}{m} \sum_{i =1}^n f_i }^m \geq \Pi_{i \leq m} f_i} \mper
\]
\end{fact}

The following fact is a simple corollary of the fundamental theorem of algebra:

\begin{fact}
For any univariate degree $d$ polynomial $p(x) \geq 0$ for all $x \in \R$,
$\sststile{d}{x} \Set{p(x) \geq 0}$.
 \label{fact:univariate}
\end{fact}

This can be extended to univariate polynomial inequalities over intervals of $\R$.
2\begin{fact}[Fekete and Markov-Lukacs, see \cite{laurent2009sums}]
For any univariate degree $d$ polynomial $p(x) \geq 0$ for $x \in [a, b]$,  $\Set{x\geq a, x \leq b} \sststile{d}{x} \Set{p(x) \geq 0}$.  \label{fact:univariate-interval}
\end{fact}

\section{Clustering Mixtures of Reasonable Distributions} \label{sec:clustering}
In this section, we provide algorithm for clustering mixtures of \emph{reasonable} distributions (see Definition~\ref{def:nice_distribution}). 
The main results of this section are \emph{simultaneous intersection bounds} (Lemmas~\ref{lem:intersection-bounds-from-spectral-separation},~\ref{lem:intersection-bounds-from-mean-separation}, and ~\ref{sec:intersection-bounds-frobenius}) that we'll rely on in the next two sections. We then use these bounds to immediately derive an algorithm (via the rounding used in Chapter 4.3 of~\cite{TCS-086}) for clustering that runs in time $d^{\poly(k) \log(\kappa)}$ where $\kappa$ is the \emph{spread} of the mixture defined as  the maximum of $\frac{v^{\top} \Sigma(j) v}{v^{\top} \Sigma(i)v}$ over all $i,j \leq k$. In Section~\ref{sec:better-algo}, we will show how to improve the running time of this algorithm to have no dependence on the spread and prove our main result (Theorem~\ref{thm:main-clustering-general-main}). 



\begin{theorem}[Clustering Mixtures of Separated Reasonable Distributions]
There exists an algorithm that takes input a sample of size $n$ from $\Delta$-separated equi-weighted mixture of reasonable distributions $\cD(\mu(r),\Sigma(r))$ for $r \leq k$ with true clusters $C_1,C_2,\ldots,C_k$ and outputs $\hat{C}_1, \hat{C}_2, \ldots \hat{C}_k$ such that there exists a permutation $\pi:[k] \rightarrow [k]$ satisfying
\[
\min_{i \leq k} \frac{|C_i \cap \hat{C}_{\pi(i)}|}{|C_i|} \geq 1- O(\eta)\mper
\]
The algorithm succeeds with probability at least $1-1/k$  whenever $\Delta =\Omega(s(\poly(\eta/k))/\poly(\eta))$, needs $d^{O\Paren{s(\poly(\eta/k))\poly(k)}}$ samples and runs in time $n^{O\Paren{s(\poly(\eta/k)) \poly(k) \log \kappa}}$ where $\kappa$ is spread of the mixture.

\label{thm:main-clustering-section}
\end{theorem}

\subsection{Algorithm}
Our constraint system $\cA$ uses polynomial inequalities to describe a subset $\hat{C}$ of size $\alpha n$ of the input sample $X$. We impose constraints on $\hat{C}$ so that the uniform distribution on $\hat{C}$ satisfies certifiable anti-concentration and hypercontractivity of degree-$2$ polynomials. We intend the true clusters $C_1, C_2, \ldots, C_r$ to be the only solutions for $\hat{C}$. Proving that this statement holds and that it has a low-degree SoS proof is the bulk of our technical work in this section.

We describe the specific formulation next.
Throughout this section, we use the notation $Q(x)$ to denote $x^{\top}Qx$ for $d \times d$ matrix valued indeterminate $Q$.
For ease of exposition, we break our constraint system $\cA$ into natural categories $\cA_1 \cup \cdots \cup \cA_5$.
Our constraint system relies on parameter $\tau,\delta$ that we will set in proof of Theorem~\ref{thm:main-clustering-section} below.

For our argument, we will need access to the square root of the indeterminate $\Sigma$. 
So we introduce the constraint system $\cA_1$ with an extra matrix valued indeterminate $\Pi$  (with auxiliary matrix-valued indeterminate $U$) that satisfies the polynomial equality constraints corresponding to $\Pi$ being the square root of $\Sigma$. Note that the first constraint is equivalent to $\Pi \succeq 0$ in ``ordinary math''.

\begin{equation}
\text{Square-Root Constraints: $\cA_1$} = 
  \left \{
    \begin{aligned}
      &
      &\Pi
      &= UU^{\top}\\
      &
      &\Pi^2
      &=\Sigma\mper\\
    \end{aligned}
  \right \}
\end{equation}
Next, we formulate intersection constraints that identify the subset $\hat{C}$ of size $\alpha n$.
\begin{equation}
\text{Subset Constraints: $\cA_2$} = 
  \left \{
    \begin{aligned}
      &\forall i\in [n]
      & w_i^2
      & = w_i\\
      &&
      \textstyle\sum_{i\in[n]} w_i
      &= \frac{n}{k}\mper \\
    \end{aligned}
  \right \}
\end{equation}
Next, we enforce that $\hat{C}$ must have mean $\mu$ and covariance $\Sigma$, where both $\mu$ and $\Sigma$ are indeterminates.
\begin{equation}
\text{Parameter Constraints: $\cA_3$} = 
  \left \{
    \begin{aligned}
      &
      &\frac{1}{n}\sum_{i = 1}^n w_i x_i
      &= \mu\\
      &
      &\frac{1}{n}\sum_{i = 1}^n w_i \paren{x_i-\mu}\paren{x_i-\mu}^{\top}
      &= \Sigma\mper\\
    \end{aligned}
  \right \}
\end{equation}
Finally, we enforce certifiable anti-concentration at two slightly different parameter regimes (characterized by $\tau \leq \delta$) along with the hypercontractivity of $\hat{C}$ .
\begin{equation}
\text{Certifiable Anti-Concentration : $\cA_4$} =
  \left \{
    \begin{aligned}
      &
      &\frac{k^2}{n^2}\sum_{i,j=  1}^n w_i w_j q_{\delta,2\Sigma}^2\left(\Paren{x_i-x_j},v\right)
      &\leq 2^{s(\delta)} C\delta \Paren{v^{\top}\Sigma v}^{s(\delta)}\\
      &
      &\frac{k^2}{n^2}\sum_{i,j=  1}^n w_i w_j q_{\tau,2\Sigma}^2\left(\Paren{x_i-x_j},v\right)
      &\leq 2^{s(\tau)} C\tau \Paren{v^{\top}\Sigma v}^{s(\tau)}\mper\\
     \end{aligned}
    \right\}
 \end{equation}

\begin{equation}
\text{Certifiable Hypercontractivity : $\cA_5$} = 
  \left \{
    \begin{aligned}
     &\forall j \leq 2s,
     &\frac{k^2}{n^2} \sum_{i,j \leq n} w_i w_j Q(x_i-x_j)^{2j}
     &\leq (Cj)^{2j}\Norm{\Pi Q \Pi}_F^{2j}\mper
    \end{aligned}
  \right \}
\end{equation}
\text{Certifiable Bounded Variance: $\cA_6$} = 
\begin{equation}
  \left \{
    \begin{aligned}
     &\forall j \leq 2s,
     &\frac{k^2}{n^2} \sum_{i,\ell \leq n} w_i w_\ell \Paren{Q(x_i-x_\ell)-\frac{k^2}{n^2} \sum_{i,\ell \leq n} w_i w_\ell Q(x_i-x_\ell)}^{2}
     &\leq C \Norm{\Pi Q \Pi}_F^{2}\mper
    \end{aligned}
  \right \}
\end{equation}

\paragraph{Algorithm.}
We are now ready to describe our algorithm.
Our algorithm follows the same outline as the simplified proof for clustering spherical mixtures presented in~\cite{TCS-086} (Chapter 4.3). The idea is to find a pseudo-distribution $\tzeta$ that minimizes the objective $\Norm{\pE[w]}_2$ and is consistent with the constraint system $\cA$. 

It is simple to round the resulting solution to true clusters: our analysis yields that the matrix $\pE[ww^{\top}]$ is approximately block diagonal with the blocks approximately corresponding to the true clusters $C_1, C_2,\ldots,C_k$. We can then recover a cluster by a repeatedly greedily selecting $n/k$ largest entries in a random row, removing those columns off and repeating. We describe this algorithm below.

\begin{mdframed}
  \begin{algorithm}[Clustering General Mixtures]
    \label{algo:rounding-for-pseudo-distribution}\mbox{}
    \begin{description}
    \item[Given:]
        A sample $X$ of size $n$ with true clusters $C_1, C_2 , \ldots, C_k$ of size $n/k$ each.
    \item[Output:]
      A partition of $X$ into an approximately correct clusters $\hat{C}_1, \hat{C}_2, \ldots, \hat{C}_k$.
    \item[Operation:]\mbox{}
    \begin{enumerate}
    \item Find a pseudo-distribution $\tzeta$ satisfying $\cA$ minimizing $\Norm{\pE[w]}_2^2$.
      \item For $M = \pE_{w \sim \tzeta} [ww^{\top}]$, repeat for $1 \leq \ell \leq k$:
        \begin{enumerate} \item Choose a uniformly random row $i$ of $M$.
        \item Let $\hat{C}_{\ell}$ be the set of points indexed by the largest $\frac{n}{k}$ entries in the $i$th row of $M$.
        \item Remove the rows and columns with indices in $\hat{C}_{\ell}$.
       \end{enumerate}
      \end{enumerate}
    \end{description}
  \end{algorithm}
\end{mdframed}

\paragraph{Analysis of the Algorithm.}
We first show that the sample $X$ inherits the relevant properties of the distributions.
Towards this, we make the following definition.

\begin{definition}["Good" Sample]
A sample $X \subseteq \R^d$ of size $n$ is said to be a good sample from a $\Delta$-separated mixture of $\cD(\mu(r),\Sigma(r))$ for $r \leq k$ if there exists a partition $X = C_1 \cup C_2 \cup \cdots C_k \subseteq \R^d$ with empirical mean and covariance $\hat{\mu}(1), \hat{\Sigma}(1), \ldots, \hat{\mu}(k), \hat{\Sigma}(k) $ such that for all $ r \in [k]$ and $s = s(\poly(\eta/k))$, 
\begin{enumerate}
  \item Empirical mean: $\Iprod{\hat{\mu}(r)  - \mu(r),v}^2 \leq 0.1 v^{\top} \Sigma(r) v$ 
  \item Empirical covariance: $ \left(1 - \frac{1}{2^{2s}}\right)\Sigma(r) \preceq \hat{\Sigma}(r) \preceq \left(1 + \frac{1}{2^{2s}}\right) \Sigma(r)$.
  \item Certifiable Anti-concentration: For all $\tau \geq \poly(\eta/k)$,
\[
\sststile{2s}{v} \Set{
\frac{k^2}{n^2}\sum_{\substack{ i\neq j \in C_r}}  q^2_{\tau, \hat{\Sigma}(r)} \Paren{x_i-x_j,v } \leq 10C\tau \Paren{v^{\top}\hat{\Sigma}(r) v}^{2s}_2}\mper
\]

\[
\sststile{2s}{v} \Set{
\frac{k}{n}\sum_{\substack{ i_1, i_2 \in C_r, j_1, j_2 \in C_{r'}}}  q^2_{\tau, \hat{\Sigma}(r)} \Paren{x_{i_1} -x_{i_2} -x_{j_1} + x_{j_2},v } \leq 10C\tau \Paren{v^{\top}(\hat{\Sigma}(r) + \hat{\Sigma}(r') v}^{2s}_2}\mper
\]

\item Certifiable Hypercontractivity: For every $j \leq s$,
\[
\sststile{2s}{Q} \Set{
\frac{k^2}{n^2} \sum_{\substack{ i\neq \ell \in C_r}}  Q(x_i-x_\ell)^{2j}
     \leq (Cj)^{2j} 2^{2j} \Norm{\hat{\Sigma}(r)^{\frac{1}{2}}Q\hat{\Sigma}(r)^{\frac{1}{2}}}_F^{2j}}\mper
\]

\item Certifiable Bounded-Variance: 
\[
\sststile{2}{Q} \Set{
\frac{k^2}{n^2} \sum_{\substack{ i\neq \ell \in C_r}}  \Paren{Q(x_i-x_\ell)-\frac{k^2}{n^2} \sum_{\substack{ i\neq \ell \in C_r}}  Q(x_i-x_\ell)}^{2}
     \leq C\Norm{\Sigma(r)^{1/2}Q\Sigma(r)^{1/2}}_F^2}\mper 
\]
\end{enumerate}
\end{definition}

Via standard concentration arguments, it is straightforward (See Section~\ref{sec:feasibility} of Appendix) to verify that a large enough sample $X$ from a $\Delta$-separated mixture of reasonable distributions is a good.

\begin{lemma}[Typical samples are good] \label{lem:typical-samples-good}
Let $X$ be a  sample of size $n$ from a equi-weighted $\Delta$-separated mixture $\cD(\mu(r),\Sigma(r))$ for $r \leq k$. 
Then, for $n_0 = \Omega\Paren{ (s(\poly(\eta/k))d)^{8s(\poly(\eta/k))} k \log k}$ and any $n \geq n_0$, $X$ is good with probability at least $1-1/d$. Further, the the uniform distribution on  $C_1,C_2,\ldots, C_k$ are pairwise $\Delta/2$-separated.
\end{lemma}

As in the spherical case~\cite{TCS-086}, the heart of the analysis involves showing that $\pE_{\tzeta}[ww^{\top}]$ is indeed approximately block diagonal whenever $\tzeta$ satisfies $\cA$. This follows immediately from the following lemma that shows that that there's a low-degree SoS proof that shows that the subset indicated by $w$ cannot simultaneously have large intersections  with two distinct clusters $C_r,C_{r'}$.

\begin{lemma}[Simultaneous Intersection Bounds from Separation]
Let $X$ be a good sample of size $n$ from a $\Delta$-separated, equi-weighted mixture of affine transforms of a reasonable distribution $\cD$ with true clusters $C_1, C_2,\ldots,C_k$. For all $r\in [k]$,  let $w(C_r)$ denote the linear polynomial $\frac{k}{n} \sum_{i \in C_r} w_i$. Then, for every $r \neq r'$ and $\delta > 0$, 
\[
\cA \sststile{O(s(\delta)^2 \log \kappa)}{w} \Set{ w(C_r) w(C_{r'}) \leq O(\delta^{1/3})}\mper
\] \label{lem:intersection-bounds-from-separation}
\end{lemma}

For the  special case of $k=2$, we obtain the following improved version with no dependence on $\kappa$ in the degree.

\begin{lemma}[Simultaneous Intersection Bounds from Separation, Two Components]
Let $X = C_1 \cup C_2$ be a good sample with true clusters $C_1, C_2$ of size $n/2$ from a $\Delta$-separated, equi-weighted mixture of affine transforms of a reasonable distribution $\cD$.
Let $w(C_r)$ denote the linear polynomial $\frac{k}{n} \sum_{i \in C_r} w_i$ for every $r \leq 2$.
Then,
\[
\cA \sststile{O(s(\delta)^2)}{w} \Set{ w(C_1) w(C_{2}) \leq O(\delta^{1/3})}\mper
\]
\label{lem:intersection-bounds-from-separation-two-components}
\end{lemma}

It is easy to finish the analysis of the algorithm given Lemma~\ref{lem:intersection-bounds-from-separation}.

\begin{proof}[Proof of Theorem~\ref{thm:main-clustering-section}]

\textbf{Enforcing Constraints.} First, we argue that the number of constraints in the SDP we need to solve to find $\tzeta$ in Step 1 above is $d^{O( \log \kappa )s(\delta)^2}$. For this, it is enough to show that the number of polynomial inequalities needed to enforce $\cA$ is appropriately bounded. $\cA_1,\cA_2,\cA_3$ encode $O(d^2)$ inequalities by direct inspection. $\cA_4,\cA_5$ superficially encode an infinitely many constraints - by applying the quantifier alternation technique that uses SoS certifiability (first used in ~\cite{DBLP:journals/corr/abs-1711-11581,HopkinsLi17}, see Page 131 of~\cite{TCS-086} for an exposition) to compress such constraints by leveraging low-degree SoS proofs allows us to encode them into $d^{O(s(\delta)^2)}$ constraints.

\textbf{Minimizing Norm.} Observe that $\Norm{\pE[w]}_2$ is a convex function in $\pE[w]$ and thus, a pseudo-distribution minimizing $\Norm{\pE[w]}_2$ consistent with $\cA$ can be found in time $n^{O(\log \kappa)s(\delta)^2)}$ if it exists using the ellipsoid method (using the separation oracle from Fact~\ref{fact:sos-separation-efficient}). The rounding itself is easily seen to take at most $O(n^2)$ time. This completes the analysis of the running time.

\textbf{Feasibility of the SDP.} In the remaining part of the analysis, we condition on the event that the input $X$ is a good sample. We show that the SDP for computing the pseudo-distribution in Step 1 of the algorithm is feasible. We exhibit a feasible solution by describing a natural setting of the indeterminates in our constraint program. Let $\zeta$ be the uniform distribution (thus, also a pseudo-distribution of degree $\infty$) on $\1(C_r)$, for all $r\in k$. That is, $\zeta$ is uniformly distributed on the true clusters. Lemma~\ref{lem:typical-samples-good} implies that setting $w = \1(C_r)$ satisfies all the constraints in $\cA$. Thus, $\tzeta$ is indeed a feasible for the SDP. Observe further that for every $i$, $\pE_{\tzeta}[w_i] = 1/k$.

\textbf{Analysis of the SDP Solution.} Now, let $\tzeta$ be the pseudo-distribution computed in Step 1 of the algorithm.  
First, observe that by Cauchy-Schwarz inequality, $\norm{\pE_{\tzeta}[w]}_2^2= \sum_{i\leq n} \pE_{\tzeta}[w_i]^2 \geq \frac{1}{n} \Paren{\sum_{i \leq n} \pE_{\tzeta}[w_i]}^2=\frac{n}{k^2}$ where we used that $\cA \sststile{}{} \Set{\frac{k}{n} \sum_{i = 1}^n w_i = 1}$. On the other hand,  we exhibited a feasible pseudo-distribution $\zeta$ above with $\Norm{\pE_{\zeta}[w]}_2^2 =\frac{n}{k^2}$. Together, we obtain that the output $\tzeta$ obtained by solving the  SDP relaxation must satisfy $\norm{\pE_{\tzeta}[w]}_2^2 = \frac{n}{k^2}$. Observe that this is equivalent to $\pE_{\tzeta}[w_i] = 1/k$ for every $i \leq n$. Thus, we can assume in the following that $\pE_{\tzeta}[w_i] = 1/k$ for all $i$. Our analysis is similar to the proofs of Lemmas 4.21 and Lemma 4.23 in~\cite{TCS-086}.

Let $M = \pE[ww^{\top}]$. Let's understand the entries of $M$ more carefully. First, since $\pE[w_i w_j] = \pE[w_i^2 w_j^2] \geq 0$, $M(i,j)$ is non-negative. The diagonals of $M$ are $\pE[w_i^2] = \pE[w_i] = 1/k$. By the Cauchy-Schwarz inequality for pseudo-distributions (Fact~\ref{fact:pseudo-expectation-cauchy-schwarz}), $M(i,j) = \pE[w_i w_j] \leq \sqrt{\pE[w_i^2]} \sqrt{\pE[w_j^2]} \leq 1/k$. Thus, the entries of $M$ are between $0$ and $1/k$. Next, observe that since $\cA \sststile{}{} \Set{w_i\frac{k}{n}\sum_{j \leq n} w_j= w_i}$, taking pseudo-expectations and rearranging yields that for every $i$, $\E_{j \sim [n]} M(i,j)= \frac{1}{k^2}$. 

For $\eta' = \eta^2/k^3$, choose $\delta = {\eta'}^3/k^{3}$. Then, applying Lemma~\ref{lem:intersection-bounds-from-separation} and using Fact~\ref{fact:sos-completeness}, we have that for every $r$, $\E_{i \in C_r} \E_{j \not \in C_{r'}} M(i,j) =\sum_{r'\neq r} \E_{i \in C_r} \E_{j \in C_{r'}} \pE[w_i w_j] = \pE[w(C_r)w(C_{r'})] \leq O(\eta')$. 

Fix any cluster $C_r$. Call an entry of $M$ large if it exceeds $\eta/k^2$. Using the above estimates, we obtain that, the fraction of entries in the $i$th row that exceed $\eta/k^2$ is at least $(1-\eta)/k$. 

On the other hand, by Markov's inequality applied to the calculation above, we obtain that with probability $1-1/k^2$ over the uniformly random choice of $i \in C_r$, $\E_{ j \not \in C_{r}} M(i,j) \leq O(\eta') = O(\eta^2/k^3)$. Call an $i \in C_r$ for which this condition holds ``good''. 

By Markov's inequality, for each good row, the fraction of $j \not \in C_r$ such that $M(i,j) \geq \eta/k^2$ is at most $\eta/k$. Thus, for any good row in $C_r$, if we take the indices $j$ corresponding to the largest $n/k$ entries $(i,j)$ in $M$, then, at most $\eta$ fraction of such $j$ are not in $C_r$. Thus, picking uniformly random row in $C_r$ and taking the largest $n/k$ entries in that row gives a subset that intersects with $C_r$ in $(1-\eta)$ fraction of the points. 

Thus, each iteration of our rounding algorithm succeeds with probability at least $1-1/k^2$. By union bound, all iterations succeed with probability at least $1-1/k$.

\end{proof}

\paragraph{Proving Lemma~\ref{lem:intersection-bounds-from-separation}}
In what follows, we focus attention on proving Lemma~\ref{lem:intersection-bounds-from-separation}. Before describing the analysis, we set some notation/shorthand and simplifying assumptions that we will use throughout this section.
\begin{enumerate}
\item First, Lemma~\ref{lem:typical-samples-good} guarantees us that $C_r$ has mean and Covariance close to the true $\mu(r),\Sigma(r)$. We abuse the notation a little bit and use $\mu(r),\Sigma(r)$ to denote the mean and covariance of $C_r$ too. This allows us the luxury of dropping an extra piece of notation and doesn't change the guarantees we obtain.

\item In the following, we will use $\cD_r = \cD(\mu(r),\Sigma(r))$ to denote the uniform distribution on $C_r$. We will use $\cD_w$ to informally (in the context of non low-degree SoS reasoning) refer to the uniform distribution on the subset indicated by $w$.
\end{enumerate}

Depending on whether $C_r,C_{r'}$ are mean separated, spectrally separated or separated in relative Frobenius distance, our proof of Lemma~\ref{lem:intersection-bounds-from-separation} breaks into three natural cases. The key part of the analysis is dealing with the case of spectral separation which then plugs into the other two cases. So we begin with it.






\subsection{Intersection Bounds from Spectral Separation}
In this subsection, we give a sum-of-squares proof of an upper bound on $w(C_r)w(C_{r'})$ whenever $\cD_r, \cD_{r'}$ are samples chosen from \emph{spectrally} separated distributions. Note that we do not have any control of the means of $\cD_r$, $\cD_{r'}$ in this subsection and our arguments must work regardless of the means (or their separation, whether large or small) of  $\cD_r,\cD_{r'}$. 

Formally, we will prove the following upper bound on $w(C_r)w(C_{r'})$ where the degree of the sum-of-squares proof grows logarithmically in the spread $\kappa$ of the mixture.

\begin{lemma}[Intersection Bounds from Spectral Separation] \label{lem:intersection-bounds-from-spectral-separation}
Let $X = C_1 \cup C_2 \cup \ldots C_r$ be a good sample of size $n$. Suppose there exists a vector $v$ such that $\Delta_{\spe} v^{\top} \Sigma(r) v \leq v^{\top} \Sigma(r') v$ for $\Delta_{\spe} \gg Cs/\delta^2$. Then, $\cA \sststile{O(s\log (2B)}{w} \Set{w(C_r) w(C_{r'}) \leq O(\sqrt{\delta})}$ where $B = \max_{i \leq k} \frac{v^{\top}\Sigma(i)v}{v^{\top} \Sigma(r') v} \leq \kappa$.
\end{lemma}
Observe that for $k = 2$, $B=1$ and thus, the lemma above results in a bound of $O(s/\delta^2)$ on the degree of the SoS proof.  As we discussed in Section~\ref{sec:overview}, the proofs of both the statements above follow by using anti-concentration of $\cD_r$ and $\cD_{r'}$ to first show a lower-bound on the variance of $\Sigma(w)$ in terms of the $v^{\top} \Sigma(r) v$ and $v^{\top} \Sigma(r')v$ and then combine it with an upper bound on $v^{\top} \Sigma(w) v$ using anti-concentration of $\cD_w$.

\Pnote{expand on the difficulty etc. here - this discussion should later be used for the of the billion overviews we are writing.}

\begin{lemma}[Large Intersection Implies High Variance, Spectral Separation]
\label{lem:spectral-lower-bound-intersection-clustering}
\begin{equation}
\cA \sststile{4s}{ } \Biggl\{w(C_{r'}) w(C_r) \Paren{v^{\top} \Paren{\Sigma(r) + \Sigma(r')} v}^{s} \leq \Paren{\frac{2}{\delta^{2}}}^s \Paren{v^{\top} \Sigma(w) v}^s + C \delta \Paren{v^{\top} \Paren{\Sigma(r) + \Sigma(r')} v}^s \Biggr\}
\end{equation}
\end{lemma}

\begin{proof}


We know from Lemma~\ref{lem:typical-samples-good} that two-sample-centered points from both  $C_r$ and $C_{r'}$ are $2s$-certifiably $(\delta,C\delta)$-anti-concentrated. Using Definition~\ref{def:certifiable-anti-concentration}, thus yields:

\begin{multline}
\cA \sststile{4s}{} \Biggl\{ \frac{k^4}{n^4} \sum_{i_1,i_2 \in C_r, j_1,j_2 \in C_{r'}} w_{i_1} w_{i_2} w_{j_1} w_{j_2} \Iprod{x_{i_1}-x_{i_2} -x_{j_1}+x_{j_2},v}^{2s}  \\ \geq \delta^{2s} w(C_r)^2w(C_{r'})^2 \Paren{v^{\top} 2(\Sigma(r) + \Sigma(r')) v^{\top}}^s\\ -  \delta^{2s}\frac{k^4}{n^4} \sum_{i_1, i_2 \in C_r, j_1,j_2 \in C_{r'
}} w_{i_1} w_{i_2} w_{j_1} w_{j_2} q_{\delta,2(\Sigma(r)+\Sigma(r'))}^2 (x_{i_1}-x_{i_2} -x_{j_1}+x_{j_2}, v) \Biggr\} \label{eq:anti-conc-app-clusters}
\end{multline}

Using that $\cA \sststile{}{} \Set{w_{i_1} w_{i_2} w_{j_1} w_{j_2} \leq 1}$ for every $i_1,i_2,j_1,j_2$ and using $2s$-certifiable $(\delta,C\delta)$-anti-concentration of $x_{i_1}-x_{i_2} -x_{j_1}+x_{j_2}$ and invoking Definition~\ref{def:certifiable-anti-concentration}, we have:

\begin{multline}
\cA \sststile{4s}{} \Biggl\{\frac{k^4}{n^4} \sum_{i_1, i_2 \in C_r, j_1,j_2 \in C_{r'
}}  w_{i_1} w_{i_2} w_{j_1} w_{j_2} q_{\delta,2(\Sigma(r)+\Sigma(r'))}^2 (x_{i_1}-x_{i_2} -x_{j_1}+x_{j_2}, v)\\ \leq \frac{k^4}{n^4} \sum_{i_1, i_2 \in C_r, j_1, j_2 \in C_{r'
}} q_{\delta,2(\Sigma(r)+\Sigma(r'))}^2 (x_{i_1}-x_{i_2} -x_{j_1}+x_{j_2}, v) \leq  C\delta \Paren{v^{\top} 2(\Sigma(r) + \Sigma(r')) v}^s \Biggr\}
\end{multline}
Plugging in the above bound in \eqref{eq:anti-conc-app-clusters} gives:
\begin{multline}
\cA \sststile{4s}{} \Biggl\{\frac{k^4}{n^4} \sum_{i_1, i_2 \in C_r, j_1, j_2 \in C_{r'
}} w_{i_1} w_{i_2} w_{j_1} w_{j_2} \Iprod{x_{i_1}-x_{i_2} -x_{j_1}+x_{j_2}, v}^{2s} \\ \geq \delta^{2s} \Paren{w(C_r)^2w(C_{r'})^2 - C\delta} \Paren{v^{\top} 2(\Sigma(r)+ \Sigma(r')) v^{\top}}^s \Biggr\}
\end{multline}
Rearranging thus yields:
\begin{multline}
\cA \sststile{4s}{} \Biggl\{\frac{1}{\delta^{2s}} \frac{k^4}{n^4} \sum_{i_1, i_2 \in C_r, j_1, j_2 \in C_{r'
}} w_{i_1} w_{i_2} w_{j_1} w_{j_2} \Iprod{x_{i_1}-x_{i_2} -x_{j_1}+x_{j_2}, v}^{2s}  + C\delta \Paren{v^{\top} 2(\Sigma(r) + \Sigma(r')) v^{\top}}^s\\ \geq w(C_r)^2w(C_{r'})^2 \Paren{v^{\top} 2(\Sigma(r)+\Sigma(r')) v^{\top}}^s \Biggr\} \label{eq:spectral-lower-bound-1}
\end{multline}
To finish the proof, we note that:
\begin{multline}
\cA \sststile{4s}{} \Biggl\{ \Paren{\frac{4cs}{\delta^{2}}}^s \Paren{v^{\top} \Sigma(w) v}^s \geq \frac{1}{\delta^{2s}} \frac{k^4}{n^4} \sum_{i_1, i_2, j_1, j_2 \in [n]} w_{i_1} w_{i_2} w_{j_1} w_{j_2} \Iprod{x_{i_1}-x_{i_2} -x_{j_1}+x_{j_2}, v}^{2s}\\  \geq \frac{1}{\delta^{2s}} \frac{k^4}{n^4} \sum_{i_1, i_2 \in C_r, j_1, j_2 \in C_{r'
}} w_{i_1} w_{i_2} w_{j_1} w_{j_2} \Iprod{x_{i_1}-x_{i_2} -x_{j_1}+x_{j_2}, v}^{2s}   \Biggr\}
\end{multline}

Plugging in the upper bound above in \eqref{eq:spectral-lower-bound-1} and canceling out a copy of $2^s$ from both sides gives the lemma.

\end{proof}

Moving forward with our proof plan, we can clearly complete the proof by giving an \emph{upper} bound on $\Paren{v^{\top} \Sigma(w) v}$ that scales as the variance of the \emph{smaller} variance component (i.e. $r$ above). We make this happen by invoking certifiable anti-concentration again - this time, however, applying it to the $w$-samples instead of $C_r$ and $C_{r'}$.

\begin{lemma}[Spectral Upper Bound via Anti-Concentration] \label{lem:anti-conc-upper-bound-w}
\begin{equation}
\cA \sststile{4s}{} \Biggl\{ \Paren{w(C_r)^2 - C\delta} \Paren{v^{\top} \Sigma(w) v^{\top}}^s \leq \Paren{\frac{Cs}{\delta^{2}}}^s \Paren{v^{\top} \Sigma(r) v}^s \Biggr\}
\end{equation}

\end{lemma}

\begin{proof}

Our constraint system $\cA$ allows us to derive that two-sample-centered points indicated by $w$ are $2s$-certifiably $(\delta,C\delta)$-anti-concentrated with witnessing polynomial $p_{\cD}$. Using Definition~\ref{def:certifiable-anti-concentration}, thus yields:

\begin{multline}
\cA \sststile{4s}{} \Biggl\{   \delta^{2s} w(C_r)^2 \Paren{v^{\top} \Sigma(w) v^{\top}}^s \\\leq \frac{k^2}{n^2} \sum_{i,j \in C_r} w_i w_j \Iprod{ \frac{1}{\sqrt{2}} \Paren{x_i - x_j}, v}^{2s} +\delta^{2s}\frac{k^2}{n^2} \sum_{i \neq j \in C_r} w_i w_j q_{\delta,\Sigma(w)}^2 \Paren{ \frac{1}{\sqrt{2}} \Paren{x_i - x_j}, v}   \Biggr\} \label{eq:main-bound-spectral-upper}
\end{multline}

Using that $\cA \sststile{4s}{\Sigma, w} \Set{w_i w_j \leq 1}$ for every $i,j$, using that $\cA$ derives $2s$-certifiable $(\delta,C\delta)$-anti-concentration of $w$-samples and invoking Definition~\ref{def:certifiable-anti-concentration}, we have:

\begin{equation}
\begin{split}
\cA \sststile{4s}{} \Biggl\{\frac{k^2}{n^2} \sum_{i \neq j \in C_r} w_i w_j q_{\delta,\Sigma(w)}^2 (\frac{1}{\sqrt{2}} \Paren{x_i - x_j}, v) & \leq \frac{k^2}{n^2} \sum_{i \neq j \in [n]} w_i w_j q_{\delta,\Sigma(w)}^2 \Paren{ \frac{1}{\sqrt{2}} \Paren{x_i - x_j}, v}\\
& \leq  C\delta \Paren{v^{\top} \Sigma(w) v}^s \Biggr\}
\end{split}
\end{equation}


Further, using that $\cA \sststile{4s}{\Sigma, w} \Set{ w_i w_j \leq 1 }$ for all $i,j$ and relying on the certifiable Sub-gaussianity of $C_r$, we have:

\begin{equation}
\cA \sststile{4s}{} \Biggl\{\frac{k^2}{n^2} \sum_{i,j \in C_r} w_i w_j \Iprod{ \frac{1}{\sqrt{2}} \Paren{x_i - x_j}, v}^{2s}  \leq \frac{k^2}{n^2} \sum_{i,j \in C_r} \Iprod{ \frac{1}{\sqrt{2}} \Paren{x_i - x_j}, v}^{2s} = \Paren{Cs}^s \Paren{v^{\top} \Sigma(r) v}^s\Biggr\}
\end{equation}

Combining the last two bounds with \eqref{eq:main-bound-spectral-upper} thus yields:

\begin{equation}
\cA \sststile{4s}{} \Biggl\{ w(C_r)^2 \Paren{v^{\top} \Sigma(w) v^{\top}}^s \leq \frac{1}{\delta^{2s}} \Paren{Cs}^s \Paren{v^{\top} \Sigma(r) v}^s + C \delta  \Paren{v^{\top} \Sigma(w) v^{\top}}^s \Biggr\}
\end{equation}

\end{proof}

\paragraph{Digression: ``Real-World'' Proof}
We'd now like to combine the upper and lower bounds on $v^{\top} \Sigma(w) v$ obtained in the two previous lemmas in order to conclude a bound on the intersection size $w^2(C_r) w^2(C_{r'})$. To aid the intuition, observe that this is easy to do in ``usual math'' (in contrast to low-degree sum-of-squares proof system). If the reader prefers to skip this digression, they can skip to the paragraph titled \emph{Upper Bounds via SoSizing Conditional Argument}.

\begin{lemma}[Low Intersection Size from Spectral Separation (\emph{not} a low-degree SoS Proof)] 
\label{lem:real-world-argument}
Let $v \in \R^d$ be a unit vector such that $ \Delta v^{\top} \Sigma(r) v \leq v^{\top} \Sigma(r')v$ for some $\Delta \gg 2Cs/\delta^3$. Then, $w^3(C_r) w^3(C_{r'}) \leq \delta$.
\end{lemma}

\begin{proof}
We split into two cases: 1) $w^2(C_r) \leq \delta$ and 2) $w(C_r)^2 >\delta$.
In the first, case $w^3(C_r)w^3(C_{r'})$ is clearly at most $\delta$.
So we are done!

In the second case, we invoke Lemma~\ref{lem:spectral-lower-bound-intersection-clustering} to write:
\[
w(C_{r'}) w(C_r) \Paren{v^{\top} \Paren{\Sigma(r) + \Sigma(r')} v^{\top}}^{s} \leq \frac{2^{s}}{\delta^{2s}} \Paren{v^{\top} \Sigma(w) v}^s + C \delta \Paren{v^{\top} \Paren{\Sigma(r) + \Sigma(r')} v^{\top}}^s\mper
\]
Since $(w^2(C_r)-\delta)\geq 0$, we can multiply both sides of above by $(w^2(C_r)-\delta)$ without changing the inequality.
By Lemma~\ref{lem:anti-conc-upper-bound-w}:
\[
 \Paren{w(C_r)^2 - C\delta} \Paren{v^{\top} \Sigma(w) v^{\top}}^s \leq \frac{1}{\delta^{2s}} \Paren{Cs}^s \Paren{v^{\top} \Sigma(r) v}^s \mper
\]
Using the above bound, using that $w(C_r)w(C_{r'})\leq 1$ and rearranging, we have:
\[
w(C_r)^2 w(C_{r'}) w(C_r) \Paren{v^{\top} \Paren{\Sigma(r) + \Sigma(r')} v^{\top}}^{s} \leq (C+1) \delta \Paren{v^{\top} \Paren{\Sigma(r) + \Sigma(r')} v^{\top}}^s + \Paren{\frac{2}{\delta}}^s \frac{1}{\delta^{2s}} \Paren{Cs}^s \Paren{v^{\top} \Sigma(r) v}^s
\mper\]
Using the above bound with the spectrally separating direction $v$, we know that $v^{\top} \Paren{\Sigma(r) + \Sigma(r')} v^{\top} \geq  \Delta v^{\top} \Sigma(r)v$. Thus rearranging the above inequality gives:
\[
w(C_r)^3 w(C_{r'})^3 \leq w^3(C_r) w(C_{r'}) \leq (C+1) \delta  + \Paren{\frac{2}{\delta^3}}^s (Cs)^s \Delta^{-s}\mcom
\]
which is at most $2C\delta$ whenever $\Delta \gg Cs/\delta^3$ as desired.
\end{proof}

Crucial to the above ``real world'' argument is the second step where we use the non-negativity of $w(C_r)^2-\delta$ so as to multiply the starting inequality on both sides with it while preserving the direction of the inequality. This step relies on an ``if-then'' case analysis which, unfortunately, cannot, in general, be implemented \emph{as is} in low-degree sum-of-squares proof system.

\paragraph{Upper Bounds via SoSizing Conditional Argument}
In order to implement an argument similar to the one above, within the low-degree SoS system, we will introduce a polynomial $\cJ$ which approximates the thresholding operation withing SoS. We prove the existence of such a polynomial in Appendix~\ref{sec:poly-approx-threshold}. This will, however, lose us a $\log (\kappa)$ factor in the SoS degree required (and thus cause an exponential dependence on $\log (\kappa)$ in the running time of our clustering algorithm).
\newcommand{\chb}{\cC}

\begin{lemma}[Polynomial Approximator for Thresholds, See Section~\ref{sec:poly-approx-threshold} for a proof] \label{lem:poly-approximate-threshold}
Let $1/2 \geq \rho \geq 0$ and $c \in [0,1]$. There exists a square polynomial $\cJ$ satisfying:
\begin{enumerate}
\item $\cJ(x) \in [1,2]$ for all $x \in [2c,1]$.
\item $\cJ(x) \leq \rho$ for all $x \in [0,c]$.
\item $\deg(\cJ) \leq O(\log(1/\rho)/c)$.
\end{enumerate}
\end{lemma}

\begin{lemma} \label{lem:consequences-of-threshold-polynomials}
For any  $0 < \rho < 1$,
\[\Set{0 \leq w(C_r) \leq 1} \sststile{O(\log(1/\rho)/\delta^2)}{w} \Set{\cJ(w(C_r)) (w(C_r)-\delta) \geq -\delta \rho}\mcom\]
and,
\[\Set{0 \leq w(C_r) \leq 1} \sststile{O(\log(1/\rho)/\delta^2)}{w} \Set{\cJ(w(C_r)) w(C_r)\geq (w(C_r)-2\delta) }\mper\]

\end{lemma}
\begin{proof}
Observe that the conclusion is a polynomial inequality in single variable $w(C_r)$. Thus, it is enough to give \emph{any} proof of $\cJ(w(C_r)) (w(C_r)-\delta) \geq -\delta \rho$ and apply Lemma~\ref{fact:univariate-interval}.

To see why the inequality holds, observe that if $w(C_r) \geq \delta$, $\cJ(w(C_r)) (w(C_r)-\delta) \geq 0 > -\delta \rho$. On the other hand, if $w(C_r) \leq \delta$, then, $\cJ(w(C_r)) \leq \rho$ while $|w(C_r)-\delta| \leq \delta$. On the other hand, observe that $\cJ(w(C_r)(w(C_r)-\delta) \leq \cJ(w(C_r) w(C_r) \leq 2w(C_r)$. This completes the proof of the first inequality.

For the second claim, notice that if $w(C_r) < 2\delta$, the inequality trivially holds since $\cJ(w(C_r)) \geq 0$. If on the other hand, $w(C_r) > 2\delta$, then, $\cJ(w(C_r)) \geq 1 \geq w(C_r) \geq w(C_r)-\delta$.
\end{proof}

We can now implement the above real-world ``conditional'' argument within SoS using the polynomial $\cJ$ above. To do this, we will need a rough upper bound on $v^{\top}\Sigma(w)v$ in terms of $v^{\top}\Sigma(r)v$ for $r \leq k$. We will prove this via another application of certifiable anti-concentration of $\cD_w$ - this time, invoked with the slightly different parameter $\tau$.  

\begin{lemma}[Rough Spectral Upper bound on $\Sigma(w)$] \label{lem:rough-bound-spectral-upper}
\begin{equation}
\cA \sststile{}{} \Biggl\{ \Paren{v^{\top} \Sigma(w) v^{\top}}^{s} \leq (2Ck)^{s+1} \Paren{Cs}^{s} \sum_{r \leq k} \Paren{v^{\top} \Sigma(r) v}^{s} \Biggr\}
\end{equation}
\end{lemma}
 \begin{proof}

Our proof is similar to the proof of Lemma~\ref{lem:anti-conc-upper-bound-w} with a key additional step. As in the proof of Lemma~\ref{lem:anti-conc-upper-bound-w}, we start by invoking our constraints to conclude (note that  we sum over all samples  this time instead of  those just in $C_r$ as in the previous lemma:

\begin{multline}
\cA \sststile{}{} \Biggl\{   \tau^{2s} \sum_{r \leq k} w'(C_r)^2 \Paren{v^{\top} \Sigma(w) v^{\top}}^s \\ \leq \frac{k^2}{n^2} \sum_{r\leq k} \sum_{i,j \in C_r} w_i w_j \Iprod{ \frac{1}{\sqrt{2}} \Paren{x_i - x_j}, v}^{2s} +\tau^{2s}\frac{k^2}{n^2} \sum_{r \leq k} \sum_{i \neq j \in C_r} w_i w_j q_{\tau,\Sigma(w)}^2 \Paren{\frac{1}{\sqrt{2}} \Paren{x_i - x_j}, v}   \Biggr\} \label{eq:main-bound-spectral-upper}
\end{multline}

The second term on the RHS can be upper bounded just as in the proof of Lemma~\ref{lem:anti-conc-upper-bound-w} to yield:

\begin{equation}
\begin{split}
\cA \sststile{}{} \Biggl\{\frac{k^2}{n^2} \sum_{r \leq k} \sum_{i \neq j \in C_r} w_i w_j q_{\tau,\Sigma(w)}^2(\Iprod{ \frac{1}{\sqrt{2}} \Paren{x_i - x_j}, v}) & \leq \frac{k^2}{n^2} \sum_{i \neq j \in [n]} w_i w_j q_{\tau,\Sigma(w)}^2 \Paren{\Iprod{ \frac{1}{\sqrt{2}} \Paren{x_i - x_j}, v}} \\
&  \leq  C\tau \Paren{v^{\top} \Sigma(w) v}^s \Biggr\}
\end{split}
\end{equation}

The first term can be also be upper bounded - this time in terms of the Covariances of all the $k$ components.

\begin{equation}
\begin{split}
\cA \sststile{}{} \Biggl\{\frac{k^2}{n^2} \sum_{r \leq k} \sum_{i,j \in C_r} w_i w_j \Iprod{ \frac{1}{\sqrt{2}} \Paren{x_i - x_j}, v}^{2s}  & \leq \sum_{r \leq k} \frac{k^2}{n^2} \sum_{i,j \in C_r} \Iprod{ \frac{1}{\sqrt{2}} \Paren{x_i - x_j}, v}^{2s} \\
& = \Paren{Cs}^s \sum_{r \leq k} \Paren{v^{\top} \Sigma(r) v}^s\Biggr\}
\end{split}
\end{equation}
We can now combine the two estimates above to yield:
\begin{equation}
\cA \sststile{}{} \Biggl\{ \Paren{\sum_{r \leq k} w(C_r)^2 - C\tau} \Paren{v^{\top} \Sigma(w) v^{\top}}^s  \leq \frac{1}{\tau^{2s}} \Paren{Cs}^s \sum_{r \leq k} \Paren{v^{\top} \Sigma(r) v}^s \Biggr\} \label{eq:main-bound-spectral-upper}
\end{equation}
So far the argument closely follows the proof of Lemma~\ref{lem:anti-conc-upper-bound-w}. The key departure we make is with the following simple observation:
\[
\cA \sststile{}{} \Set{ \sum_{r \leq k} w(C_r)^2 \geq \frac{1}{k} \Paren{\sum_{r \leq k} w(C_r)}^2 = \frac{1}{k}}\mper
\]
Thus, as long as $\tau < \frac{1}{2Ck}$, we can derive:
\begin{equation}
\cA \sststile{}{} \Biggl\{ \Paren{v^{\top} \Sigma(w) v}^{s} \leq k^{s+1} \Paren{Cs}^{s} \sum_{r \leq k} \Paren{v^{\top} \Sigma(r) v}^{s} \Biggr\} \label{eq:main-bound-spectral-upper}
\end{equation}
This is the ``rough'' upper bound on $\Sigma(w)$ we were after.
\end{proof}

We can use the above lemma to get an ``upgraded'' version of Lemma~\ref{lem:anti-conc-upper-bound-w}.

\begin{lemma}[Upper Bound on Variance of $\cD_w$] \label{lem:upper-bound-variance-w-samples-sharper}
Let $\lambda_{\max}(v) \Norm{v}_2^2$ be the maximum of $v^{\top} \Sigma(r) v$ over all $r \leq k$. Then,
\begin{multline}
\cA \sststile{}{} \Biggl\{ \Paren{\cJ(w(C_r)) (w(C_r)-\delta) + \delta \rho} \Paren{v^{\top} \Sigma(w)v}^s \leq 2\frac{1}{\delta^{2s}} \Paren{Cs}^s \Paren{v^{\top} \Sigma(r) v}^s \\+ \delta \rho s^{2s} \Paren{Cs}^{s} k \lambda_{\max}(v)^{s} \norm{v}_2^{2s} \Biggr\}\mper
\end{multline}
\end{lemma}

\begin{proof}
From Lemma~\ref{lem:rough-bound-spectral-upper}, we have:
\begin{equation}
\cA \sststile{}{} \Biggl\{ \Paren{v^{\top} \Sigma(w) v}^{s} \leq (s)^{s+1} \Paren{Cs}^{s} \sum_{r \leq k} \Paren{v^{\top} \Sigma(r) v}^{s} \Biggr\}
\end{equation}

Then, the above bound implies:
\begin{equation}
\cA \sststile{}{}\Biggl\{ \Paren{v^{\top} \Sigma(w) v}^{s} \leq (s^{s+1} \Paren{Cs}^{s} k \lambda_{max}(v)^{s} \Biggr\} \mper \label{eq:better-bound-based-on-eig}
\end{equation}

From Lemma~\ref{lem:consequences-of-threshold-polynomials}, we have: $\cA \sststile{}{} \Set{\cJ(w(C_r)) \leq 2}$. Thus, using Lemma~\ref{lem:anti-conc-upper-bound-w} and applying \eqref{eq:better-bound-based-on-eig} on the RHS, we can conclude:

\begin{align*}
\cA &\sststile{}{} \Biggl\{ \Paren{\cJ(w(C_r)) (w(C_r)-\delta) + \delta \rho}\Paren{v^{\top} \Sigma(w) v}^s \leq \delta \rho \Paren{v^{\top} \Sigma(w) v}^s +  2\Paren{\frac{Cs}{\delta^{2}}}^s \Paren{v^{\top} \Sigma(r) v}^s\\
&\leq \delta \rho s^{2s} \Paren{Cs}^{s} k \lambda_{max}(v)^{s} \norm{v}_2^{2s}  +  2\Paren{\frac{Cs}{\delta^{2}}}^s \Paren{v^{\top} \Sigma(r) v}^s \Biggr\} \mper
\end{align*}

\end{proof}

We are now ready to complete the proof of Lemma ~\ref{lem:intersection-bounds-from-spectral-separation}.

\begin{proof}[Proof of Lemma~\ref{lem:intersection-bounds-from-spectral-separation}]
Observe that $\cA \sststile{}{} \Set{0 \leq w(C_r) \leq 1}$.
Thus,
\begin{equation} \label{eq:non-negative-version-diff}
\cA \sststile{}{}\Set{\cJ(w(C_r))(w(C_r) -\delta)  + \delta \rho \geq 0}\mper
\end{equation}

From Lemma~\ref{lem:spectral-lower-bound-intersection-clustering}, we have:

\[
\cA \sststile{}{} \Biggl\{ w(C_{r'}) w(C_r) \Paren{v^{\top} \Paren{\Sigma(r) + \Sigma(r')} v}^{s} \leq \frac{2^{s}}{\delta^{2s}} \Paren{v^{\top} \Sigma(w) v}^s + C \delta \Paren{v^{\top} \Paren{\Sigma(r) + \Sigma(r')} v}^s \Biggr\} \mper
\]

Using \eqref{eq:non-negative-version-diff} along with \eqref{eq:sos-addition-multiplication-rule} with $\cJ(w(C_r))(w(C_r) -\delta)  + \delta \rho$ gives:

\begin{multline}
\cA \sststile{}{}\Biggl\{ \Paren{ \cJ(w(C_r))(w(C_r) -\delta)}w(C_{r'}) w(C_r) \Paren{v^{\top} \Paren{\Sigma(r) + \Sigma(r')} v}^{s} \\\leq \delta \rho \Paren{v^{\top} \Paren{\Sigma(r) + \Sigma(r')} v}^{s} +  \Paren{ \cJ(w(C_r))(w(C_r) -\delta) + \delta \rho }\frac{2^{s}}{\delta^{2s}} \Paren{v^{\top} \Sigma(w) v}^s\\ + \Paren{ \cJ(w(C_r))(w(C_r) -\delta) + \delta \rho } \frac{1}{\delta^{2s}} \Paren{Cs}^s \Paren{v^{\top} \Sigma(r) v}^s + \Paren{ \cJ(w(C_r))(w(C_r) -\delta) + \delta \rho } 2C \delta \Paren{v^{\top} \Sigma(r') v}^s \Biggr\} \mper
\end{multline}

Rearranging yields:
\begin{multline}
\cA \sststile{}{} \Biggl\{ \cJ(w(C_r))(w(C_r)w(C_{r'}) w(C_r) \Paren{v^{\top} \Paren{\Sigma(r) + \Sigma(r')} v}^{s} \\\leq   2\delta \rho \Paren{v^{\top} \Paren{\Sigma(r) + \Sigma(r')} v}^{s} +  \Paren{ \cJ(w(C_r))(w(C_r) -\delta) + \delta \rho }\frac{2^{s}}{\delta^{2s}} \Paren{v^{\top} \Sigma(w) v}^s\\ + \Paren{ \cJ(w(C_r))(w(C_r) -\delta) + \delta \rho } \frac{1}{\delta^{2s}} \Paren{Cs}^s \Paren{v^{\top} \Sigma(r) v}^s + \Paren{ \cJ(w(C_r))(w(C_r) -\delta) + \delta \rho } 2C \delta \Paren{v^{\top} \Sigma(r') v}^s \Biggr\} \mper
\end{multline}

Using Lemma~\ref{lem:consequences-of-threshold-polynomials}, we have that $\cJ(w(C_r)) w(C_r) \geq (w(C_r)-\delta)$.
Multiplying the above inequality (using \eqref{eq:sos-addition-multiplication-rule}) by the SoS (and thus  non-negative) polynomial $w(C_r)w(C_{r'})\Paren{v^{\top} \Paren{\Sigma(r) + \Sigma(r')} v^{\top}}^{s}$ yields:
\[
\cA \sststile{}{} \Biggl\{ \cJ(w(C_r))w(C_{r'}) w^2(C_r) \Paren{v^{\top} \Paren{\Sigma(r) + \Sigma(r')} v}^{s} \geq (w(C_r)-\delta) w(C_{r'}) w (C_r)\Paren{v^{\top} \Paren{\Sigma(r) + \Sigma(r')} v}^{s} \Biggr\}\mper
\]

Thus, the LHS above is lower bounded by $(w(C_r)-\delta) w(C_{r'}) w (C_r)\Paren{v^{\top} \Paren{\Sigma(r) + \Sigma(r')} v}^{s}$.

Let's analyze the terms in the RHS one by one. The first term can be upper bounded directly by applying Lemma~\ref{lem:upper-bound-variance-w-samples-sharper}.

The remaining two  terms in  the RHS can be upper bounded by relying on: $\cA \sststile{}{} \Set{\cJ(w(C_r))(w(C_r) -\delta) + \delta \rho \leq 2}$.

Thus, using the above bounds we have:

\begin{multline}
\cA \sststile{}{}\Biggl\{ w(C_r)^2 w(C_{r'}) \Paren{v^{\top} \Paren{\Sigma(r) + \Sigma(r')} v}^{s} \leq 3\delta \Paren{v^{\top} \Paren{\Sigma(r) + \Sigma(r')} v}^{s}\\ + 2 \frac{1}{\delta^{2s}} \Paren{Cs}^s \Paren{v^{\top} \Sigma(r) v}^s + \delta \rho s^{2s} \Paren{Cs}^{s} k \lambda_{\max}(v)^{s} \norm{v}_2^{2s} \\+ 2\frac{1}{\delta^{2s}} \Paren{Cs}^s \Paren{v^{\top} \Sigma(r) v}^s + 4C \delta \Paren{v^{\top} \Sigma(r') v}^s \Biggr\} \label{eq:using-condition-number-spectral-sep}
\end{multline}

Next, observe that since $C_r,C_{r'}$ are spectrally separated and  $0 \leq v^{\top} \Sigma(r) v < v^{\top} \Sigma(r') v$. Thus, $v^{\top} \Sigma(r')v=^{def} \lambda_{r'}(v) \norm{v}_2^{2} > 0$.

We now set $\eta \leq s^{-2s} \Paren{Cs}^{-s} k^{-1} \lambda_{max}(v)^{-s} \lambda_{r'}(v)^s \geq s^{-O(s)} k^{-1} B^{-s}$ and use that $\Delta^s \geq Cs/\delta^2$ to conclude:

\begin{equation}
\cA \sststile{(s \log(2B)/\delta^2)}{} \Biggl\{ w(C_r)^2 w^2(C_{r'}) \leq  w(C_r)^2 w(C_{r'}) \leq O(\delta) \Biggr\}
\end{equation}
Applying Lemma~\ref{lem:cancellation-SoS-constant-RHS} completes the proof.
\end{proof}

\paragraph{Simpler Proof for $k =2$}
For the special case of $k = 2$, we can bypass the use of the threshold approximator above to get a simpler proof. 

\begin{proof}[Special case of $k=2$]

We proceed exactly as in the proof of Lemma~\ref{lem:intersection-bounds-from-spectral-separation} until equation ~\eqref{eq:using-condition-number-spectral-sep} where we invoke the uniform eigenvalue upper bound.
Instead of using the uniform eigenvalue upper bound on $\Sigma(w)$, we use Lemma~\ref{lem:rough-bound-spectral-upper}, setting $t = s(1/2Ck) \leq 1/k^{\Theta(1)} = O(1)$ for  $k = 2$ to derive:

\begin{equation}
\cA \sststile{4t}{} \Biggl\{ \Paren{v^{\top} \Sigma(w) v^{\top}}^{t} \leq 2^{O(t)} \Paren{\Paren{v^{\top} \Sigma(1) v}^{t} +\Paren{v^{\top} \Sigma(2) v}^{t} } \Biggr\}
\end{equation}

With this sharper upper bound, we can complete the proof as in Lemma~\ref{lem:intersection-bounds-from-spectral-separation} by setting $\tau = 2^{-\Theta(s) } k^{-1} \delta$ instead of $1/\poly(\kappa)$. Since  $\log(1/\tau) = \Theta(s)/\delta = \poly(1/\delta)$, the degree of the SoS proof does not grow with $\kappa$ anymore.

\end{proof}

\begin{remark}[Difficulty in extending the simpler argument to $k>2$] \label{remark:difficulty-rough-upper-bound}
For mixtures with larger number of components, the upper bound from Lemma~\ref{lem:rough-bound-spectral-upper} is not enough. This is because the upper bound in the Lemma~\ref{lem:rough-bound-spectral-upper} scales with the largest variance of any of the $k$ component distributions which could be a lot larger than the variance of $\cD_r$ and $\cD_{r'}$ in the direction $v$. 
\end{remark}

\subsection{Intersection Bounds from Mean Separation}
In this section, we give a low-degree sum-of-squares proof that if $C_r,C_{r'}$ are mean separated then $w(C_r)w(C_{r'})$ must be small. Formally, we will show:

\begin{lemma}[Intersection Bounds from Mean Separation]\label{lem:intersection-bounds-from-mean-separation}
Let $X = C_1 \cup C_2 \cup \ldots C_r$ be a good sample of size $n$. Suppose there exists a vector $v \in \R^d$ such that $\Iprod{\mu_r -\mu_{r'},v}_2^2 \geq \Delta^2_m v^{\top} \Paren{\Sigma(r) +\Sigma(r')} v$.

Then, whenever $\Delta_m \gg Cs/\delta$,
\[
\cA \sststile{O(s(\delta)/\delta^2 \log (\kappa))}{w} \Set{w(C_r)w(C_{r'}) \leq  O(\sqrt{\delta}) }\mper
\]
\end{lemma}

As in the previous subsection, we can get a sum-of-squares proof of absolute constant degree for the special case of $k=2$ components.
 \begin{lemma}[Intersection Bounds from Mean Separation]\label{lem:intersection-bounds-from-mean-separation-2-components}
Let $X = C_1 \cup C_2$ be a good sample of size $n$. Suppose there exists a vector $v \in \R^d$ such that $\Iprod{\mu(1) -\mu(2),v}_2^2 \geq \Delta^2_m v^{\top} \Paren{\Sigma(1) +\Sigma(2)} v$.

Then, whenever $\Delta_m \gg \Theta(1)$,
\[
\cA \sststile{O(s(\delta)/\delta^2}{w} \Set{w(C_1)w(C_{2}) \leq  O(\sqrt{\delta}) }\mper
\]
\end{lemma}

We will need the following technical fact in our proof.\begin{lemma}[Lower Bounding Sums] \label{lem:lower-bounding-sums-SoS}
Let $A,B,C,D$ be scalar-valued indeterminates. Then, for any $\tau >0$,
\[
\Set{0 \leq A, B \leq A+B \leq 1} \cup \Set{0\leq C,D} \cup \Set{C+D\geq \tau} \sststile{2}{A,B,C} \Set{ AC+ BD \geq \tau AB}\mper
\]

\end{lemma}
\begin{proof}
We have:
\begin{multline}
\Set{0 \leq A, B \leq A+B \leq 1} \cup \Set{0\leq C,D} \cup \Set{C+D\geq F} \sststile{}{} \Biggl\{ AC + BD \geq (A+B) (AC+BD) \\\geq A^2C + AB (C+D) + B^2 D \geq AB(C+D) \geq \tau AB\Biggr\}
\end{multline}
\end{proof}

\begin{proof}[Proof of Lemma~\ref{lem:intersection-bounds-from-mean-separation}]
Let $v$ be the direction in which the means of $C_r$ and $C_{r'}$ are separated. Then, we have:
\begin{equation} \label{eq:mean-separation-bound}
\Iprod{\mu_r -\mu_{r'},v}_2^{2s} \geq \Delta_m^{2s} \Paren{v^{\top} \Paren{\Sigma(r) +\Sigma(r')} v}^s\mper
\end{equation}

Assume, WLOG, that $v^{\top}\Sigma(r) v \leq v^{\top} \Sigma(r')v$.

Applying Lemma~\ref{lem:lower-bounding-sums-SoS} with $A = w(C_r)$, $B  = w(C_{r'})$, $C = \Iprod{\mu_r -\mu(w),v}^{2s}$ and $D =  \Iprod{\mu_{r'} -\mu(w),v}^{2s}$ along with the SoS Almost Triangle Inequality (Fact~\ref{fact:sos-almost-triangle}) and certifiable Sub-gaussianity constraints ($\cA_5$) yields:
\begin{align*}
\cA \sststile{4s}{\mu, w} & \Biggl\{ \Paren{Cs}^s \Paren{v^{\top} \Sigma(w) v}^s \geq \frac{1}{n} \sum_{i \leq n} w_i \Iprod{x_i - \mu(w),v}^{2s} \geq \frac{1}{n} \sum_{i \in C_r \cup C_{r'}} w_i \Iprod{x_i -\mu(w),v}^{2s}\\
&\geq \frac{1}{2^s} \Paren{w(C_{r}) \Iprod{\mu_r -\mu(w),v}^{2s}-\frac{1}{n} \sum_{i \in C_r} w_i \Iprod{x_i -\mu_r,v}^{2s}} \\
&+ \frac{1}{2^s} \Paren{w(C_{r'}) w_i \Iprod{\mu_{r'} -\mu(w),v}^{2s} - \frac{1}{n} \sum_{i \in C_{r'}} w_i \Iprod{x_i -\mu_{r'},v}^{2s}}\\
&\geq \frac{1}{2^s} \Paren{w(C_{r}) \Iprod{\mu_r -\mu(w),v}^{2s} + w(C_{r'}) \Iprod{\mu_{r'} -\mu(w),v}^{2s}} - \frac{1}{2^s} \Paren{v^{\top} \Sigma(r) v}^s- \frac{1}{2^s} \Paren{v^{\top} \Sigma(r') v}^s\\
&\geq \frac{1}{2^{s+1}} \Paren{w(C_{r})w(C_{r'}) \Paren{\Iprod{\mu_r -\mu(w),v}^{2s} + \Iprod{\mu_{r'} -\mu(w),v}^{2s}}} - \frac{1}{2^s} \Paren{v^{\top} \Sigma(r) v}^s- \frac{1}{2^s} \Paren{v^{\top} \Sigma(r') v}^s\\
&\geq \Paren{\frac{\Delta_m}{4}}^{2s} \Paren{w(C_{r})w(C_{r'}) \Paren{\Paren{v^{\top} \Sigma(r) v}^{s}+ \Paren{v^{\top} \Sigma(r') v}^{s}}} - \frac{1}{2^s} \Paren{v^{\top} \Sigma(r) v}^s- \frac{1}{2^s} \Paren{v^{\top} \Sigma(r') v}^s
\Biggr\}\mcom
\end{align*}

Rearranging the chain of reasoning above thus yields:
\begin{equation}
\cA \sststile{4s}{}  \Biggl\{ 2^s \Paren{\Paren{Cs}^s \Paren{v^{\top} \Sigma(w) v}^s + \Paren{v^{\top} \Sigma(r) v}^s + \Paren{v^{\top} \Sigma(r') v}^s} \geq \Delta_m^{2s}w(C_{r})w(C_{r'}) \Paren{\Paren{v^{\top} \Sigma(r) v}^{s}+ \Paren{v^{\top} \Sigma(r') v}^{s}} \Biggr\} \mper \label{eq:basic-bound-mean-separation}
\end{equation}

Lemma~\ref{lem:consequences-of-threshold-polynomials} shows a low-degree SoS proof of non-negativity of $\cJ(w(C_r))(w(C_r) - \delta) + \delta \rho $ in variables $w$:
\[
\cA \sststile{4s}{w}  \Set{\cJ(w(C_r))(w(C_r) - \delta) + \delta \rho \geq 0}\mper
\]

Thus, we can multiply \eqref{eq:basic-bound-mean-separation} by $\Paren{\cJ(w(C_r))(w(C_r) - \delta) + \delta \rho}$ throughout to obtain:
\begin{multline}
\cA \sststile{4s}{\mu, w}   \Biggl\{ \Paren{\cJ(w(C_r))(w(C_r) - \delta) + \delta \rho} \Paren{\Paren{2Cs}^s \Paren{v^{\top} \Sigma(w) v}^s + 2^s \Paren{v^{\top} \Sigma(r) v}^s +  2^s \Paren{v^{\top} \Sigma(r') v}^s} \\\geq  \Delta_m^{2s} \Paren{\cJ(w(C_r))(w(C_r) - \delta) + \delta \rho} \Paren{w(C_{r})w(C_{r'})} \Paren{\Paren{v^{\top} \Sigma(r) v}^{s}+ \Paren{v^{\top} \Sigma(r') v}^{s}} \Biggr\} \mper \label{eq:basic-bound-mean-separation}
\end{multline}

Applying Lemma~\ref{lem:upper-bound-variance-w-samples-sharper} for the first term on the LHS and using that $\cA \sststile{}{} \Set{\Paren{\cJ(w(C_r))(w(C_r) - \delta) + \delta \rho} \leq 2}$ and rearranging the above inequality gives:
\begin{multline}
\cA \sststile{4s}{\mu, w}   \Biggl\{ (2Cs)^s \Paren{\delta \rho s^{2s} \Paren{Cs}^{s} k \lambda_{max}(v)^{s} + 2\frac{1}{\delta^{2s}} \Paren{Cs}^s \Paren{v^{\top} \Sigma(r) v}^s} + 2^s \Paren{v^{\top} \Sigma(r) v}^s +  2^s \Paren{v^{\top} \Sigma(r') v}^s\\ + 2\Delta_m^{2s}\delta \Paren{\Paren{v^{\top} \Sigma(r) v}^{s}+ \Paren{v^{\top} \Sigma(r') v}^{s}} \\\geq  \Delta_m^{2s} \cJ(w(C_r)) \Paren{w^2(C_{r})w(C_{r'})} \Paren{\Paren{v^{\top} \Sigma(r) v}^{s}+ \Paren{v^{\top} \Sigma(r') v}^{s}} \Biggr\} \mper \label{eq:bound-mean-separation-2clustering}
\end{multline}

Using Lemma~\ref{lem:consequences-of-threshold-polynomials}, we also have:
\[
\cA \sststile{4s}{w}  \Set{ \cJ(w(C_r))w(C_r) \geq \Paren{w(C_r)-\delta }}\mper
\]

Using this bound on the RHS of \eqref{eq:bound-mean-separation-2clustering} and rearranging yields:
\begin{multline}
\cA \sststile{4s}{\mu, w}   \Biggl\{ (2Cs)^s \Paren{\delta \rho \lambda_{max}^s + 2\frac{1}{\delta^{2s}} \Paren{Cs}^s \Paren{v^{\top} \Sigma(r) v}^s} + 2^s \Paren{v^{\top} \Sigma(r) v}^s +  2^s \Paren{v^{\top} \Sigma(r') v}^s\\ + 2\Delta_m^{2s}\delta \Paren{\Paren{v^{\top} \Sigma(r) v}^{s}+ \Paren{v^{\top} \Sigma(r') v}^{s}} \\\geq \Delta_m^{2s} \Paren{w^2(C_{r})w(C_{r'})} \Paren{\Paren{v^{\top} \Sigma(r) v}^{s}+ \Paren{v^{\top} \Sigma(r') v}^{s}} \Biggr\} \mper
\end{multline}

Dividing throughout by $\Delta_m^{2s} \Paren{\Paren{v^{\top} \Sigma(r) v}^{s}+ \Paren{v^{\top} \Sigma(r') v}^{s}}$ and recalling that $ v^{\top} \Sigma(r) v \leq v^{\top} \Sigma(r') v$ yields:

\begin{equation}
\cA \sststile{4s}{\mu, w}   \Biggl\{  \Paren{w^2(C_{r})w(C_{r'})} \leq \Delta_m^{-2s} (2Cs)^s \Paren{\delta \rho \kappa^s} + 2 \Paren{\frac{C\sqrt{s}}{\Delta_m\delta}}^{2s} + 2\delta \Biggr\} \mper
\end{equation}

Thus, choosing $\rho = \kappa^{-s}$ and using that $\Delta_m \gg Cs/\delta$ ensures that we obtain:

\begin{equation}
\cA \sststile{4s}{}   \Biggl\{  \Paren{w^2(C_{r})w^2(C_{r'})} \leq \Paren{w^2(C_{r})w(C_{r'})} \leq  O(\delta) \Biggr\} \mper
\end{equation}

\end{proof}

\paragraph{Improved SoS Degree Bounds for $k =2$}
\begin{proof}[Proof of Lemma~\ref{lem:intersection-bounds-from-mean-separation-2-components}]
We proceed exactly as in the above proof of Lemma~\ref{lem:intersection-bounds-from-mean-separation} up until \eqref{eq:bound-mean-separation-2clustering} where we invoke a rough eigenvalue upper bound on $\Sigma(w)$. We replace this bound by the sharper bound for the $k=2$ case given by Lemma~\ref{lem:rough-bound-spectral-upper} analogous to the case of spectral separation and  get to choose $\tau = O(1/\delta)$. We can  then finish the argument as in the proof of Lemma~\ref{lem:intersection-bounds-from-mean-separation} above.
\end{proof}

\subsection{Intersection Bounds from Relative Frobenius Separation of Covariances} \label{sec:intersection-bounds-frobenius}
In this section, we show that if $C_r$ and $C_{r'}$ are generated by Gaussians with covariances that are separated in relative Frobenius distance, then $w(C_r)w(C_{r'}) = O(\delta)$. 

Recall that in this case, $\Sigma(r)$ and $\Sigma(r')$ have the same range (as linear operators). 
Thus, WLOG, we can assume them to be full rank.

\begin{lemma}[Intersection Bounds from Relative Frobenius Separation]
Suppose $\Norm{\Sigma(r')^{-1/2}\Sigma(r)\Sigma(r')^{-1/2} - I}_F^2 \geq \Delta_{cov}^2 \Paren{ \Norm{\Sigma(r')^{-1/2}\Sigma(r)^{1/2}}_{op}^4}$ for $\Delta_{cov} \gg Cs(\delta)/\delta^2$. Then,
\[
\cA \sststile{O(s(\delta)\log \kappa}{w} \Set{w(C_r) w(C_{r'}) \leq O(\delta^{1/3})}\mper
\] \label{lem:intersection-bounds-from-rel-frob-sep}
\end{lemma}

As in the previous two subsections, we can get a constant degree sum-of-squares proof for the special case of $k=2$ components.
\begin{lemma}[Intersection Bounds from Relative Frobenius Separation,  Two Components]
Suppose $\Norm{\Sigma(2)^{-1/2}\Sigma(1)\Sigma(2)^{-1/2} - I}_F^2 \geq \Delta_{cov}^2 \Paren{ \Norm{\Sigma(2)^{-1/2}\Sigma(1)^{1/2}}_{op}^4}$. Then,
\[
\cA \sststile{O(1/\delta^2)}{} \Set{w(C_1) w(C_{2}) \leq O(\delta^{1/3})}\mper
\] \label{lem:intersection-bounds-from-rel-frob-sep-two-components}
\end{lemma}

Let $Q$ be a $d \times d$ matrix-valued indeterminate. In the following, we write $Q(z)$ for $z^{\top}Qz$ (the quadratic form associated with $Q$). We also use the notation $\E_w Q = \frac{k}{n} \sum_{i,j} w_i w_j Q(x_i-x_j)$ - the polynomial computing the mean of $Q$ with respect to the subsample indicated by $w$. We also write $\E_{C_r} Q = \frac{k}{n} \sum_{i,j \in  C_r} Q(x_i-x_j) $ and $\E_{C_{r'}} Q = \frac{k}{n} \sum_{i,j \in  C_r} Q(x_i-x_j)$. We note that for any distribution $\cD$ with covariance $\Sigma$, $\E_{x,y \sim \cD} (x-y)^{\top}Q(x-y) = 2\tr(\Sigma Q)$.

\paragraph{Proof of Lemma~\ref{lem:intersection-bounds-from-rel-frob-sep}}
We can now proceed with the proof of Lemma~\ref{lem:intersection-bounds-from-rel-frob-sep}. As in the previous two subsections, the idea is to show a lower bound on the variance of some polynomial in terms of the intersection size $w(C_r)w(C_{r'})$ and couple it with an upper bound on the variance that follows from certifiable hypercontractivity to obtain an upper bound on $w(C_r)w(C_{r'})$.

Observe that the relative Frobenius separation condition is invariant  under linear transformations.
Thus, we can assume that $\Sigma(r') = I$ WLOG. This simplifies notation quite a bit in this argument. With this simplification, we now have: $\Norm{\Sigma(r)-I}_{F}^2 \geq \Delta_{cov}^2$. Further, the covariance of $C_r$ is now $\Sigma(r')^{-1/2} \Sigma(r) \Sigma(r')^{-1/2}$ and that of $C_{r'}$ is now $I$ after this linear transformation. It's also easy to verify that $\frac{k^2}{n^2} \sum_{i,j} w_i w_j \Sigma(r')^{-1/2}\Paren{x_i -x_j}\Paren{x_i -x_j}^{\top}\Sigma(r')^{-1/2} = 2 \Sigma(r')^{-1/2}\Sigma(w)\Sigma(r')^{-1/2}$. 

In order to simplify notation, we will simply treat $\Sigma(r') = I$ and $\Sigma(r)\rightarrow \Sigma(r')^{-1/2}\Sigma(r)\Sigma(r')^{-1/2}$ in the analysis below. 


We start with the lower-bound first.

\begin{lemma}[Large Intersection Implies High Variance] \label{lem:Large-Intersection-Implies-High-Variance-frobenius}
Let $Q = \Sigma(r')^{-1/2}\Sigma(r)\Sigma(r')^{-1/2}-I$.

\[
\cA \sststile{4}{w} \Set{4\E_w (Q-\E_w Q)^2 + 2  \E_{C_r} (Q-\E_{C_r}Q)^2 + 2 \E_{C_{r'}} (Q-\E_{C_{r'}}Q)^2\geq w(C_{r})^2w^2(C_{r'})\Norm{\Sigma(r')^{-1/2}\Sigma(r)\Sigma(r')^{-1/2}-I}_F^4}
\]
\end{lemma}

\begin{proof}
Observe that $\E_{C_r}Q = \tr(\Sigma(r) ( \Sigma(r)-I)) = \Norm{\Sigma(r)-I}_F^2 + \tr(\Sigma(r)-I)$ while, $\E_{C{r'}}Q  =  \tr(\Sigma(r)-I)$. In particular, $\E_{C_r}Q - \E_{C_{r'}}Q = \Norm{\Sigma(r)-I}_F^2 \geq \Delta_{cov}^2$. Thus, the mean of the polynomial $Q(x)$ is starkly different on the two components. By observing that the standard deviation of $Q$ on each of $C_r$ and $C_{r'}$ is much smaller than the mean, we will be able to derive a lower-bound on variance of $Q$ under $w$-samples.

Thus, applying Lemma~\ref{lem:lower-bounding-sums-SoS}, we have:
\begin{equation}
\cA \sststile{4}{w} \Set{ w(C_r)^2 \Paren{\E_{C_r}Q - \E_w Q}^2  + w(C_{r'})^2 \Paren{\E_{C_{r'}}Q - \E_w Q}^2 \geq \frac{1}{4} w(C_r)^2 w(C_{r'})^2 \Norm{\Sigma(r)-I}_F^4} \label{eq:lower-bound-frob-sep}
\end{equation}

Let's now lower bound $\E_{w}(Q-\E_{w}Q)^2$.
We have:

\begin{align*}
\cA \sststile{4}{w} \Biggl\{ \E_w (Q-\E_w Q)^2 &= \frac{k^2}{n^2} \sum_{i, j \leq n } w_i w_j \Paren{Q(x_i - x_j)- \E_w Q}^2 \geq \frac{k^2}{n^2} \sum_{i, j \leq C_r \text{ or } i,j \in C_{r'}} w_i w_j \Paren{Q(x_i - x_j)- \E_w Q}^2\\
& \geq \frac{k^2}{n^2} \sum_{i, j \leq C_r } w_i w_j \Paren{\E_{C_r}Q -\E_w Q}^2 - \frac{1}{2} \frac{k^2}{n^2} \sum_{i, j \leq C_r } w_i w_j \Paren{Q(x_i - x_j)- \E_{C_r}Q}^2\\
&+ \frac{k^2}{n^2} \sum_{i, j \leq C_r } w_i w_j \Paren{\E_{C_{r'}}Q- \E_w Q}^2-\frac{1}{2} \frac{k^2}{n^2} \sum_{i, j \leq C_{r'} } w_i w_j \Paren{Q(x_i - x_j)- \E_{C_{r'}}Q}^2\\
& \geq \frac{1}{2} w(C_r)^2 \Paren{\E_{C_r}Q -\E_w Q}^2 - \frac{1}{2} \frac{k^2}{n^2} \sum_{i, j \leq C_r } \Paren{Q(x_i - x_j)- \E_{C_r}Q}^2\\
&+ \frac{1}{2} w(C_{r'})^2 \Paren{\E_{C_{r'}}Q- \E_w Q}^2-\frac{1}{2} \frac{k^2}{n^2} \sum_{i, j \leq C_{r'} } \Paren{Q(x_i - x_j)- \E_{C_{r'}}Q}^2\\
&\geq \frac{1}{4} w(C_{r})^2w^2(C_{r'})\Norm{\Sigma(r)-I}_F^4 - \frac{1}{2}  \E_{C_r} (Q-\E_{C_r}Q)^2 - \frac{1}{2} \E_{C_{r'}} (Q-\E_{C_{r'}}Q)^2
\Biggr\}\mcom
\end{align*}
where, in the final inequality, we applied \eqref{eq:lower-bound-frob-sep}. Rearranging completes the proof.

\end{proof}

Onwards to the upper bound now.
Observe that the first two terms on the LHS of Lemma~\ref{lem:Large-Intersection-Implies-High-Variance-frobenius} can be upper bounded easily using Lemma~\ref{lem:typical-samples-good}:
$\E_{C_r} (Q-\E_{C_r}Q)^2 \leq (C-1) \Norm{\Sigma(r)^{1/2}Q\Sigma(r)^{1/2}}_F^2 \leq \Norm{\Sigma(r)^{1/2}}^2_{op} \Norm{Q}_F^2$. Similarly, $\E_{C_{r'}} (Q-\E_{C_r}Q)^2 \leq \Norm{Q}_F^2$.

Thus, to finish the proof of Lemma~\ref{lem:intersection-bounds-from-rel-frob-sep}, we need an upper bound on $\E_w (Q-\E_w Q)^2$ which we accomplish by relying on the certifiable hypercontractivity constraints.

In the following, we will use the following observation:
From our bounded-variance constraints in $\cA$, we have:
\begin{equation}
\cA \sststile{4}{\Pi,w} \Set{ \E_{w} (Q-\E_w Q)^2  \leq C\Norm{\Pi(w)Q\Pi(w)}_F^2}\mper \label{eq:variance-polynomials}
\end{equation}

From Lemma~\ref{lem:upper-bound-variance-w-samples-sharper}, we have:
\[
\cA \sststile{O(s\log \kappa)/\delta^2}{\Pi,w}\Biggl\{ \Paren{\cJ(w(C_r)) (w(C_r)-\delta) + \delta \rho} \Paren{v^{\top} \Sigma(w)v}^s \leq 2\frac{1}{\delta^{2s}} \Paren{Cs}^s \Paren{v^{\top} \Sigma(r) v}^s + \delta \rho \lambda_{max}^s \norm{v}_2^{2s} \Biggr\}\mper
\]

To implement the linear transformation $x_i \rightarrow \Sigma(r')^{-1/2} x_i$, we substitute $v = \Sigma(r')^{-1/2}v$ and use that $\Sigma(r')^{-1} \succeq 1/\lambda_{max} I$:
\begin{multline}
\cA \sststile{O(s\log \kappa)/\delta^2}{\Pi,w} \Biggl\{ \Paren{\cJ(w(C_r)) (w(C_r)-\delta) + \delta \rho} \Norm{\Pi(w)\Sigma(r')^{\dagger/2}v}_2^{2s} \\ \leq 2\frac{1}{\delta^{2s}} \Paren{Cs}^s \Norm{v}_2^{2s} + \delta \rho \lambda_{max}^s \norm{\Sigma(r')^{\dagger/2}v}_2^{2s} \leq \Paren{2\frac{1}{\delta^{2s}} \Paren{Cs}^s + \delta \rho \kappa^s} \norm{v}_2^{2s} \Biggr\}\mper \label{eq:spectral-upper-bound-repurposed-frobenius}
\end{multline}
We are now ready for the upper bound proof. 

\begin{lemma}[Certifiable Hypercontractivity Implies Low Variance] \label{lem:variance-upper-bound-hypercontractivity}
Let $Q =\Sigma(r)-I$.
\begin{multline}
\cA \sststile{O(s\log \kappa)/\delta^2}{w} \Biggl\{ \Paren{\cJ(w(C_r)) (w(C_r)-\delta) + \delta \eta}^{2s} \Paren{\E_{w} (Q-\E_w Q)^2}^s \\\leq  \Paren{4\frac{1}{\delta^{2s}} \Paren{Cs}^s \Norm{\Sigma(r)^{1/2}}_{op}^{2s}}^2s^{2s} \Norm{\Sigma(r)-I}_F^2 \Biggr\}
\end{multline}
\end{lemma}
\begin{proof}
Lemma~\ref{lem:consequences-of-threshold-polynomials} implies that $\cA \sststile{}{} \Set{\Paren{\cJ(w(C_r)) (w(C_r)-\delta) + \delta \rho} \geq 0}$. Thus, we can use the multiplication rule (Fact~\ref{eq:sos-addition-multiplication-rule}) and multiply both sides of \eqref{eq:variance-polynomials} with $\Paren{\cJ(w(C_r)) (w(C_r)-\delta) + \delta \rho}$ repeatedly while preserving the inequality.

Thus, we have using the bounded-variance constraints in $\cA$:
\begin{align*}
\cA &\sststile{O(s\log \kappa)/\delta^2}{\Pi,w} \Biggl\{ \Paren{\cJ(w(C_r)) (w(C_r)-\delta) + \delta \rho}^{s} \Paren{\E_{w} (Q-\E_w Q)^2}^s\\
&\leq \Paren{\cJ(w(C_r)) (w(C_r)-\delta) + \delta \rho}^{s} (C-1)^s \Norm{\Pi(w)Q\Pi(w)}_F^{2s}\\
&\leq 2^{s} \Paren{\cJ(w(C_r)) (w(C_r)-\delta) + \delta \rho}^2 (C-1)^s \Norm{\Pi(w)Q'\Pi(w)}_F^{2s}\\
&\leq  2^{s}\Paren{ \Paren{\frac{1}{\delta^{2}}}^s \Paren{Cs}^s \Norm{\Sigma(r)^{1/2}}_{op}^{2s}  + \delta \rho \kappa^s}s^s \Paren{\cJ(w(C_r)) (w(C_r)-\delta) + 2^s \delta \rho} \Norm{Q\Pi(w)}_F^{2s} \\
&\leq  2^s \Paren{\Paren{\frac{1}{\delta^{2}}}^s\Paren{Cs}^s \Norm{\Sigma(r)^{1/2}}_{op}^{2s}  + \delta \rho \kappa^s}^2 s^{2s} \Norm{Q}_F^{2s}\\
&= \Paren{\Paren{\frac{2}{\delta^{2}}}^s \Paren{Cs}^s \Norm{\Sigma(r)^{1/2}}_{op}^{2s}  + \delta \rho \kappa^s}^2 s^{2s} \Norm{\Sigma(r)-I}_F^{2s} \Biggr\}\mcom
\end{align*}
where, in the last two inequalities, we twice invoked the contraction bound from Lemma~\ref{lem:contraction-property} along with the bound on $\Norm{\Pi(w) \Sigma(r')^{-1/2} v}_2^s$ from~\eqref{eq:spectral-upper-bound-repurposed-frobenius}. Setting $\rho = \kappa^{-s}$ completes the proof.
\end{proof}

As in the previous subsection, we can improve the sum-of-squares degree of the proof above to be a fixed constant (independent of $\kappa$) in the case when $k=2$ by using the sharper bound on $\Sigma(w)$ in \eqref{eq:spectral-upper-bound-repurposed-frobenius}.

\begin{lemma}[Certifiable Hypercontractivity Implies Low Variance, Two Components] \label{lem:variance-upper-bound-hypercontractivity-Two-Components}
Let $Q = \Sigma(2)^{-1/2}\Sigma(1)\Sigma(2)^{-1/2}-I$.
\begin{multline}
\cA \sststile{O(s\log \kappa)/\delta^2}{Q,\Sigma, w} \Biggl\{ \Paren{\cJ(w(C(1))) (w(C_1)-\delta) + \delta \rho}^{2s} \Paren{\E_{w} (Q-\E_w Q)^2}^s \\\leq  \Paren{4\frac{1}{\delta^{2s}} \Paren{Cs}^s \Norm{\Sigma(r)^{1/2}\Sigma(2)^{-1/2}}_{op}^{2s}}^2s^{2s} \Norm{\Sigma(2)^{-1/2}\Sigma(1)\Sigma(2)^{-1/2}-I}_F^2 \Biggr\}
\end{multline}
\end{lemma}

\begin{proof}
We proceed similarly as in the proof above up until \eqref{eq:spectral-upper-bound-repurposed-frobenius} where, instead of using the uniform eigenvalue bound, we instead use the sharper bound from Lemma~\ref{lem:rough-bound-spectral-upper}. As in the previous two subsections, following through the rest of the proof in Lemma~\ref{lem:variance-upper-bound-hypercontractivity} as is, allows us to eventually set $\rho = O(1)$ yielding a $O(1)$-degree SoS proof as desired.
\end{proof}

\begin{proof}[Proof of Lemma~\ref{lem:intersection-bounds-from-rel-frob-sep}]
As in the previous two lemmas, we argue after performing the linear transformation $\Sigma(r')^{-1/2}$ on the samples in order to simplify notation.

From Lemma~\ref{lem:Large-Intersection-Implies-High-Variance-frobenius}, we have:
\[
\cA \sststile{4}{w} \Set{4\E_w (Q-\E_w Q)^2 + 2  \E_{C_r} (Q-\E_{C_r}Q)^2 + 2 \E_{C_{r'}} (Q-\E_{C_{r'}}Q)^2\geq w(C_{r})^2w^2(C_{r'})\Norm{\Sigma(r)-I}_F^4}
\]
Multiplying both sides of the and apply the SoS Almost Triangle Inequality (Fact~\ref{fact:sos-almost-triangle}) and obtain:
\[
\cA \sststile{4s}{} \Set{2^{3s} \Paren{\E_w (Q-\E_w Q)^{2s} + \E_{C_r} (Q-\E_{C_r}Q)^{2s} + \E_{C_{r'}} (Q-\E_{C_{r'}}Q)^{2s}} \geq w(C_{r})^{2s}w^{2s}(C_{r'})\Norm{\Sigma(r)-I}_F^{4s}}
\]
Multiplying by $\Paren{\cJ(w(C_r)) (w(C_r)-\delta) + \delta \rho}^{s}$ on both sides, we get:

\begin{multline}
\cA \sststile{O(s\log \kappa)}{Q,\Sigma, w} \Biggl\{\Paren{\cJ(w(C_r)) (w(C_r)-\delta) + \delta \rho}^s w(C_{r})^{2s}w^{2s}(C_{r'})\Norm{\Sigma(r)-I}_F^{4s} \\\leq \Paren{\cJ(w(C_r)) (w(C_r)-\delta) + \delta \rho}^s 2^{3s} \Paren{\E_w (Q-\E_w Q)^{2s} + \E_{C_r} (Q-\E_{C_r}Q)^{2s} + \E_{C_{r'}} (Q-\E_{C_{r'}}Q)^{2s}}\Biggr\}\mper
\end{multline}

Using the upper bounds proved above (Lemma~\ref{lem:variance-upper-bound-hypercontractivity} and the preceding discussion) on each of the three terms on the RHS, we get:

\begin{multline}
\cA \sststile{O(s\log \kappa)}{Q,\Sigma, w} \Biggl\{\Paren{\cJ(w(C_r)) (w(C_r)-\delta) + \delta \rho}^{s} w(C_{r})^{2s}w^{2s}(C_{r'})\Norm{\Sigma(r)-I}_F^{4s} \\\leq 2^{O(s)} \Paren{4\frac{1}{\delta^{2s}} \Paren{Cs}^s \Norm{\Sigma(r)^{1/2}\Sigma(r')^{-1/2}}_{op}^{2s} + 1} \Norm{\Sigma(r)-I}_F^{2s} \Biggr\}\mper
\end{multline}

Applying the SoS Cancellation lemma (Lemma~\ref{lem:cancellation-SoS-constant-RHS}), we have:

\begin{multline}
\cA \sststile{O(s\log \kappa)}{Q,\Sigma, w} \Biggl\{\Paren{\cJ(w(C_r)) (w(C_r)-\delta) + \delta \rho} w(C_{r})^{2}w^{2}(C_{r'})\Norm{\Sigma(r)-I}_F^{4} \\\leq 2^{O(s)} \Paren{4\frac{1}{\delta^{2}} \Paren{Cs} \Norm{\Sigma(r)^{1/2}\Sigma(r')^{-1/2}}_{op}^{2}} \Norm{\Sigma(r)-I}_F^{2} \Biggr\}\mper
\end{multline}

Applying Lemma~\ref{lem:consequences-of-threshold-polynomials} to observe
\[
\cA \sststile{O(s\log \kappa)}{Q,\Sigma, w} \Set{ \Paren{\cJ(w(C_r)) (w(C_r)-\delta) + \delta \rho} \geq (w(C_r)-2\delta)}\mper
\]

Thus, using $\cA \sststile{}{} \Set{w(C_r)^{2} w(C_{r'})^{2}) \leq 1}$, we get:

\begin{multline}
\cA \sststile{O(s\log \kappa)}{Q,\Sigma, w} \Biggl\{w(C_{r})^{3}w^{2}(C_{r'})\Norm{\Sigma(r)-I}_F^{4} \\\leq  2 \delta \Norm{\Sigma(r)-I}_F^{4} + 2^{O(s)} \Paren{4\frac{1}{\delta^{2}} \Paren{Cs} \Norm{\Sigma(r)^{1/2}\Sigma(r')^{-1/2}}_{op}^{2}} \Norm{\Sigma(r)-I}_F^{2} \Biggr\}\mper
\end{multline}

Dividing throughout by $\Norm{\Sigma(r)-I}_F^{4}$, and using that  and that $\Norm{\Sigma(r)-I}_F^{2} \geq \Delta_{cov}^2 \Norm{\Sigma(r)^{1/2}\Sigma(r')^{-1/2}}_{op}^2$ yields:

\begin{equation}
\cA \sststile{O(s\log \kappa)}{} \Biggl\{w(C_r)^3 w(C_{r'})^{3} \leq 2\delta  +  \Paren{4\frac{1}{\delta^{2}} \Paren{Cs} \Delta_{cov}^{-2s}} \Norm{\Sigma(r)-I}_F^{2s} \Biggr\}\mper
\end{equation}

Using that $\Delta_{cov} \gg Cs/\delta^2$ yields:
\begin{equation}
\cA \sststile{O(s\log \kappa)}{} \Biggl\{w(C_r)^3 w(C_{r'})^{3} \leq O(\delta) \Biggr\}\mper
\end{equation}

Using SoS cancellation (Lemma~\ref{lem:cancellation-SoS-constant-RHS}) again yields:
\begin{equation}
\cA \sststile{O(s\log \kappa)}{Q,\Sigma, w} \Biggl\{w(C_r) w(C_{r'}) \leq O(\delta^{1/3}) \Biggr\}\mper
\end{equation}

\end{proof}

\paragraph{Improved SoS Degree Bounds for $k =2$} By using Lemma~\ref{lem:variance-upper-bound-hypercontractivity-Two-Components} instead of Lemma~\ref{lem:variance-upper-bound-hypercontractivity} in the above argument immediately yields Lemma~\ref{lem:intersection-bounds-from-rel-frob-sep-two-components}.


\Pnote{In this section, $Y$ is the input sample, $X$ is the original uncorrupted sample and $X'$ is what our algorithm tries to find.}
\section{Outlier-Robust Clustering of Reasonable Distributions} \label{sec:outlier-robust-clustering}
In this section, we augment the algorithm from the previous section to tolerate an $\epsilon\leq O(1/k)$ fraction of fully adversarial outliers. Recall that in this setting, the input sample $Y$ is obtained by first generating a sample $X$ from the underlying mixture model and adversarially corrupting an $\epsilon$-fraction of $X$. 

The following is the main result of this section:

\begin{theorem}[Outlier-Robust Clustering of Mixture of Reasonable Distributions]
Fix $\epsilon > 0$. Let $\cD$ be a nice distribution that is $s(\delta)$-certifiably $(\delta,C\delta)$-anti-concentrated for all $\delta > 0$ and has $h$-certifiably $C$-hypercontractive degree $2$ polynomials for every $h$. There exists an algorithm that takes input an $\epsilon$ corruption $Y$ of $X$ of size $n$ generated according equi-weighted $\Delta$-separated mixture of $\cD(\mu(r),\Sigma(r))$ for $r \leq k$ with true clusters $C_1,C_2,\ldots,C_k$ and outputs $\hat{C}_1, \hat{C}_2, \ldots \hat{C}_k$ such that there exists a permutation $\pi:[k] \rightarrow [k]$ satisfying
\[
\min_{i \leq k} \frac{|C_i \cap \hat{C}_{\pi(i)}|}{|C_i|} \geq 1- \eta - O(k \epsilon)\mper
\]
The algorithm succeeds with probability at least $1-1/k$ whenever $\Delta \geq \Delta_{rob} =\Omega(s(\poly(\eta/k))/\poly(\eta))$, need $n\geq d^{O\Paren{s(\poly(\eta/k)) \poly(k/\eta)}}$ samples and runs in time $n^{O\Paren{\log \kappa s(\poly(\eta/k)) \poly(k/\eta)}}$ where $\kappa$ is spread of the mixture.

For the special case of $k=2$, the algorithm runs in time $n^{O(s(\poly(\eta/k))}$ and uses $d^{O(s(\poly(\eta/k)))}$ samples (with no dependence on the spread  $\kappa$.)
\label{thm:main-robust-clustering-section}
\end{theorem}

Recall that the spread $\kappa = \sup_{v \in \R^d} \max_{i,j\leq k} \frac{v^{\top} \Sigma(i) v}{v^{\top}\Sigma(j)v}$. 
In Section~\ref{sec:better-algo}, we will use the algorithm above as a subroutine to get a fully-polynomial algorithm with no dependence on the spread $\kappa$ of the mixture in the running time. 


\subsection{Algorithm}
\paragraph{Constraint System} 
Our constraint system $\cA_{rob}$ is similar to the one from the previous section with one key difference introduced in order to handle the adversarial outliers. In the uncorrupted setting, we are given the original uncorrupted sample $X = C_1 \cup C_2 \cup \ldots C_k$ and our program encodes constraints on a subset $\hat{C}$ of samples with the intended solutions to be the true clusters $C_i$s. 

In the outlier-robust setting, we only get to observe the $\epsilon$-corruption $Y$ of $X$. Thus, the points in the indices corresponding to $C_i$ need not satisfy the constraints from the previous section. 

We handle this by introducing an extra set of $d$-dimensional vector-valued indeterminates $X' = \{x'_1, x'_2, \ldots, x'_n\}$ that are intended to be the original uncorrupted sample $X$ that generated $Y$. Since $X'$ is (supposed to be) a uncorrupted sample, we can now encode finding a subset $\hat{C}$ of $X'$ (instead of $X$) with the intended solutions to be the true clusters $C_i$s of the original $X$. In order to force $X'$ to be close to $X$, we force constraints  intersection constraints (via the new matching variables $m_i$s) that ask $X'$ to intersect $Y$ in $(1-\epsilon)$-fraction of points (just like the true $X$ does). This implies that $X'$ intersects $X$ in $\geq (1-2\epsilon)$-fraction of the points and as we will soon see, this is enough for us to execute the arguments from the previous section with relatively little change. 

Covariance constraints introduce a matrix valued indeterminate intended to be the square root of $\Sigma$.
\begin{equation}
\text{Covariance Constraints: $\cA_1$} = 
  \left \{
    \begin{aligned}
      &
      &\Pi
      &=UU^{\top}\\
      &
      &\Pi^2
      &=\Sigma \mper\\
    \end{aligned}
  \right \}
\end{equation}
The intersection constraints force that $X'$ be close to $X$. 
\begin{equation}
\text{Intersection Constraints: $\cA_2$} = 
  \left \{
    \begin{aligned}
      &\forall i\in [n],
      & m_i^2
      & = m_i\\
      &&
      \textstyle\sum_{i\in[n]} m_i 
      &= (1-\epsilon) n\\
      &\forall i \in [n],
      &m_i (y_i-x'_i)
      &= 0\mper
    \end{aligned}
  \right \}
\end{equation}
The subset constraints introduce $w$, which indicates the subset $\hat{C}$ intended to be the true clusters of $X'$.
\begin{equation}
\text{Subset Constraints: $\cA_3$} = 
  \left \{
    \begin{aligned}
      &\forall i\in [n].
      & w_i^2
      & = w_i\\
      &&
      \textstyle\sum_{i\in[n]} w_i 
      &= \frac{n}{k} \mper\\
    \end{aligned}
  \right \}
\end{equation}

Parameter constraints create indeterminates to stand for the covariance $\Sigma$ and mean $\mu$ of $\hat{C}$ (indicated by $w$).
\begin{equation}
\text{Parameter Constraints: $\cA_4$} = 
  \left \{
    \begin{aligned}
      &
      &\frac{1}{n}\sum_{i = 1}^n w_i \Paren{x'_i-\mu}\Paren{x'_i-\mu}^{\top}
      &= \Sigma\\
      &
      &\frac{1}{n}\sum_{i = 1}^n w_i x'_i
      &= \mu \mper\\
    \end{aligned}
  \right \}
\end{equation}


Finally, we enforce certifiable anti-concentration and hypercontractivity of $\hat{C}$. 
\begin{equation}
\text{Certifiable Anti-Concentration : $\cA_4$} =
  \left \{
    \begin{aligned}
      &
      &\frac{k^2}{n^2}\sum_{i,j=  1}^n w_i w_j q_{\delta,\Sigma}^2\left(\Paren{x'_i-x'_j},v\right)
      &\leq 2^{s(\delta)} C\delta \Paren{v^{\top}\Sigma v}^{s(\delta)}\\
      &
      &\frac{k^2}{n^2}\sum_{i,j=  1}^n w_i w_j q_{\tau,\Sigma}^2\left(\Paren{x'_i-x'_j},v\right)
      &\leq 2^{s(\tau)} C\tau \Paren{v^{\top}\Sigma v}^{s(\eta)} \mper\\
     \end{aligned}
    \right\}
 \end{equation}

\begin{equation}
\text{Certifiable Hypercontractivity : $\cA_5$} = 
  \left \{
    \begin{aligned}
     &\forall j \leq 2s
     &\frac{k^2}{n^2} \sum_{i,j \leq n} w_i w_j Q(x'_i-x'_j)^{2j}
     &\leq (Cj)^{2j}2^{2j}\Norm{\Pi Q \Pi}_F^{2j}\mper
    \end{aligned}
  \right \}
\end{equation}

\text{Certifiable Bounded Variance: $\cA_6$} = 
\begin{equation}
  \left \{
    \begin{aligned}
     &\forall j \leq 2s,
     &\frac{k^2}{n^2} \sum_{i,\ell \leq n} w_i w_\ell \Paren{Q(x'_i-x'_\ell)-\frac{k^2}{n^2} \sum_{i,\ell \leq n} w_i w_\ell Q(x'_i-x'_\ell)}^{2}
     &\leq C \Norm{\Pi Q \Pi}_F^{2}\mper
    \end{aligned}
  \right \}
\end{equation}

Our rounding algorithm is exactly the same as in the previous section giving us:
\begin{mdframed}
  \begin{algorithm}[Outlier-Robust Clustering General Mixtures]
    \label{algo:rounding-for-pseudo-distribution-robust}\mbox{}
    \begin{description}
    \item[Given:]
        An $\epsilon$-corruption $Y$ of original uncorrupted sample $X = C_1 \cup C_2 \cup \ldots C_k$ with true clusters $C_1, C_2 , \ldots, C_k$.
    \item[Output:]
      A partition of $Y$ into an approximately correct clustering $\hat{C}_1, \hat{C}_2, \ldots, \hat{C}_k$.
    \item[Operation:]\mbox{}
    \begin{enumerate}
    \item Find a pseudo-distribution $\tzeta$ satisfying $\cA_{rob}$ minimizing $\Norm{\pE[w]}_2^2$.
      \item For $M = \pE_{w \sim \tzeta} [ww^{\top}]$, repeat for $1 \leq \ell \leq k$:
        \begin{enumerate} \item Choose a uniformly random row $i$ of $M$.
        \item Let $\hat{C}_{\ell}$ be the largest $\frac{n}{k}$ entries in the $i$th row of $M$.
        \item Remove the rows and columns with indices in $\hat{C}_{\ell}$.
       \end{enumerate}
      \end{enumerate}
    \end{description}
  \end{algorithm}
\end{mdframed}

\paragraph{Analysis of Algorithm}
An analog of Lemma~\ref{lem:typical-samples-good} extends to this setting without any change.
\begin{lemma}[Typical samples are good] \label{lem:typical-samples-good-robust} 
Let $X$ be an original uncorrupted sample of size $n$ from a equi-weighted $\Delta$-separated mixture $\cD(\mu(r),\Sigma(r))$ for $r \leq k$. 

Then, for $n_0 = \Omega\Paren{sd)^{8s} k \log k}$ and for all $n \geq n_0$, the original uncorrupted sample $X$ of size $n$ is good with probability at least $1-1/d$.
\end{lemma}

As in the previous section, the heart of the analysis is proving the following lemma that bounds the pairwise products $w(C_r)w(C_{r'})$ for all $r \neq r'$.

\begin{lemma}[Intersection Bounds from Separation]
Let $Y$ be an $\epsilon$-corruption of a good sample $X$ from a $\Delta\geq \Delta_{rob}$-separated mixture of reasonable distribution $\cD$ with true clusters $C_1, C_2,\ldots,C_k$ of size $n/k$.
Let $w(C_r)$ denote the linear polynomial $\frac{k}{n} \sum_{i \in C_r} w_i$ for every $r \leq k$. 
Then, for every $r \neq r'$,
\[
\cA \sststile{O(s(\delta)^2 \log \kappa /\delta^2)}{w} \Set{\sum_{r \neq r'} w(C_r) w(C_{r'}) \leq O(k\epsilon) + O(k^2\delta^{1/3})}\mper
\] \label{lem:intersection-bounds-from-separation-robust}
\end{lemma}

For the special case when the number of components in the mixture is $k= 2$, we can improve on the lemma above and give a sum-of-squares proof of degree $O(s(\delta)^2)$ with no dependence on $\kappa$.

\begin{lemma}[Intersection Bounds from Separation, Two Components]
Let $Y$ be an $\epsilon$-corruption of a good sample $X$ from a $\Delta\geq \Delta_{rob}$-separated mixture of reasonable distribution $\cD$ with true clusters $C_1,C_2$ of size $n/2$ each.
Let $w(C_r)$ denote the linear polynomial $\frac{k}{n} \sum_{i \in C_r} w_i$ for every $r \leq 2$. 
Then, 
\[
\cA_{rob} \sststile{O(s(\delta)^2 /\delta^2)}{} \Set{ w(C_1) w(C_{2}) \leq O(\epsilon + \delta^{1/3})}\mper
\] \label{lem:intersection-bounds-from-separation-two-components-robust}
\end{lemma}

Given Lemma~\ref{lem:intersection-bounds-from-separation}, the proof of Theorem~\ref{thm:main-robust-clustering-section} follows by the same argument as for Theorem~\ref{thm:main-clustering-section}. 

\subsection{Proof of Lemmas~\ref{lem:intersection-bounds-from-separation-robust} and ~\ref{lem:intersection-bounds-from-separation-two-components-robust}} 
As we show in this section, the proof of Lemma~\ref{lem:intersection-bounds-from-separation-robust} follows from essentially the same argument as in the previous section with two additional observations.  

The key idea in bringing the machinery from the previous section into play is to consider the following variables that satisfy constraints of being the indicator of the intersection between $X'$ (indeterminates in our program) and $X$ (original uncorrupted sample we do not have access to) - let $m_i' = m_i \cdot \1 (y_i = x_i)$ for every $i$. We now make the following key definition/notation.

\begin{definition}[Proxy Variables and Cluster Sizes] 
Let $w_i' = w_i m_i' = w_i m_i \1(y_i =x_i)$ and define $w'(C_r) = \frac{k}{n} \sum_{i \in C_r} w'_i$ for every $r$.
\end{definition}

We refer to $w_i'$ variables as proxy variables (they allow us to talk about subsets of $X$ by ``proxy''). Observe that we do not have access to the $w_i'$ variables through our program. They only appear in our analysis of the algorithm. They allow us to ``go between'' $x_i$s (the originals sample that we do not have access to) and $x_i'$ (the indeterminates that our constraints are defined over). 

The result that formally allows us to do this is:

\begin{lemma}[Matching with Original Uncorrupted Samples] \label{lem:replace-X-by-X'}
Let $m_i' = m_i \cdot \1 (y_i = x_i)$ for every $i$. Let $w_i' = w_i m_i' = w_i m_i \1 (y_i = x_i)$. 
Then, 
\[
\cA_{rob} \sststile{2}{w'} \Set{ {w_i'}^2 = w_i' \text{ } \forall i} \cup \Set{w_i'(x'_i - x_i) = 0}\mper
\]
\end{lemma}
\begin{proof}

For the first conclusion,
\[
\cA_{rob} \sststile{2}{w'} \Set{ {w_i'}^2 = w_i^2 m_i^2 \cdot \1 (y_i = x_i)^2 = w_i m_i \1 (y_i = x_i) = w_i'}\mper
\]

For the second conclusion,
\[
\cA_{rob} \sststile{2}{w'} \Set{ w_i'(x_i'-x_i) = w_i'(x_i'-y_i) + w_i'(y_i-x_i) = \1 (y_i = x_i)  w_i m_i (x_i'-y_i) + m_i w_i \1 (y_i = x_i) (x_i-y_i) = 0}\mper
\]
\end{proof}

Using this simple lemma, as we will soon discuss in some more detail, we get to apply our previous arguments to the original sample $X$ by simply shifting to the ``proxy'' $w_i'$ variables. As a result, we will be able to prove the following intersection bounds for the proxy cluster sizes.

\begin{lemma}[Proxy Intersection Bounds from Separation]
Let $Y$ be an $\epsilon$-corruption of a good sample $X$.
Let $w'(C_r)$ denote the linear polynomial $\frac{k}{n} \sum_{i \in C_r} w'_i$ for every $r \leq k$. 
Then, for every $r \neq r'$,
\[
\cA_{rob} \sststile{O(s(\delta)^2 \log \kappa)}{} \Set{ w'(C_r) w'(C_{r'}) \leq O(\delta^{1/3})}\mper
\] \label{lem:intersection-bounds-from-separation-robust-proxy}
\end{lemma}

For the special case when the number of components in the mixture is $k= 2$, we can improve on the lemma above and give a sum-of-squares proof of degree $O(s(\delta)^2)$ with no dependence on $\kappa$.

\begin{lemma}[Proxy Intersection Bounds from Separation, Two Components]
Let $Y$ be an $\epsilon$-corruption of a good sample $X$.
Let $w'(C_r)$ denote the linear polynomial $\frac{k}{n} \sum_{i \in C_r} w'_i$ for every $r \leq 2$. 
Then, 
\[
\cA_{rob} \sststile{O(s(\delta)^2)}{} \Set{ w'(C_1) w'(C_{2}) \leq O(\delta^{1/3})}\mper
\] \label{lem:intersection-bounds-from-separation-two-components-robust-proxy}
\end{lemma}

It is easy  to complete the proof of Lemmas~\ref{lem:intersection-bounds-from-separation-robust} and ~\ref{lem:intersection-bounds-from-separation-two-components-robust-proxy} using the above two lemmas. We show the proof for Lemma~\ref{lem:intersection-bounds-from-separation-robust}. The proof for Lemma~\ref{lem:intersection-bounds-from-separation-two-components-robust-proxy} is analogous.

We will use the following bound that (in low-degree SoS) shows that $X$ and $X'$ intersect in $(1-2\epsilon)n$  points.
\begin{lemma}[Matching with Original Uncorrupted Samples] \label{lem:intersection-X-X'}
Let $m_i' = m_i \cdot \1 (y_i = x_i)$ for every $i$. Then, 
\[
\cA_{rob} \sststile{2}{} \Set{\sum_{i \leq n} m_i' \geq (1-2\epsilon)n}\mper
\]
\end{lemma}
\begin{proof}
Observe that using $\{m_i^2 = m_i\} \sststile{2}{m} \Set{m_i \leq 1}$, we have:
\[
\cA_{rob} \sststile{2}{} \Set{ \sum_{i \leq n} m_i \cdot \1 (y_i \neq x_i) \leq \sum_{i \leq n} \1 (y_i \neq x_i)= \epsilon n}\mper 
\]
Similarly,
\[
\cA_{rob} \sststile{2}{} \Set{ \sum_{i \leq n} (1-m_i) \cdot \1 (y_i = x_i) \leq \sum_{i \leq n}(1-m_i) = \epsilon n}\mper
\]
Thus, 
\[
\cA_{rob} \sststile{2}{} \Set{ \sum_{i \leq n} m_i \cdot \1 (y_i = x_i) \geq \sum_{i \leq n} \Paren{m_i + (1-m_i)} \Paren{ \1 (y_i = x_i) + \1(y_i \neq x_i)} \geq n - 2 \epsilon n }\mper
\]
\end{proof}

\begin{proof}[Proof of Lemma~~\ref{lem:intersection-bounds-from-separation-robust}]
Observe that using  $\cA_{rob} \sststile{}{} \Set{m_i' \leq 1}$ for every $i$, and $\cA_{rob}\sststile{}{} \Set{ \sum_{r \leq k} w(C_r) = 1}$ we have:
\begin{align*}
\cA_{rob} &\sststile{s(\delta)\log(\kappa)}{w,w',m'} \Biggl\{ \sum_{r \neq r'} w'(C_r)w'(C_{r'})= \frac{k^2}{n^2} \sum_{r \neq r'} \sum_{i \in C_r, j\in C_{r'}} w_i w_j m_i' m_j'\\
&\geq \frac{k^2}{n^2} \sum_{r \neq r'} \sum_{i \in C_r, j\in C_{r'}} w_i w_j - 2\frac{k^2}{n^2} \sum_{r \neq r'} \sum_{i \in C_r, j\in C_{r'}} w_i w_j (1-m_i) \\
&\geq \frac{k^2}{n^2} \sum_{r \neq r'} \sum_{i \in C_r, j\in C_{r'}} w_i w_j - 2\frac{k^2}{n^2} \sum_{r \neq r'} \sum_{i \in C_r, j\in C_{r'}} w_i  (1-m_i) \\
&\geq \frac{k^2}{n^2} \sum_{r \neq r'} \sum_{i \in C_r, j\in C_{r'}} w_i w_j - 2\frac{k}{n} \sum_{r \neq r'} \sum_{i \in C_r, j\in C_{r'}} (1-m_i)\\
&\geq \frac{k^2}{n^2} \sum_{r \neq r'} \sum_{i \in C_r, j\in C_{r'}} w_i w_j - 2\frac{k}{n} \sum_{r \neq r'} \sum_{i \in C_r, j\in C_{r'}} (1-m_i)\\
&= \frac{k^2}{n^2} \sum_{r \neq r'} \sum_{i \in C_r, j\in C_{r'}} w_i w_j - 2k\epsilon \Biggr\}\mper
\end{align*}

Rearranging yields: 
\[
\cA_{rob} \sststile{s(\delta)\log(\kappa)}{w} \Set{\sum_{r \neq r'} w(C_r) w(C_{r'}) \leq \sum_{r \neq r'} w'(C_r) w'(C_{r'}) + 2k\epsilon}\mper
\]

Plugging in the bound from Lemma~\ref{lem:intersection-bounds-from-separation-robust-proxy} completes the proof.
\end{proof}

\subsection{Proof of the Simultaneous Proxy  Intersection Bounds}

We prove Lemma~\ref{lem:intersection-bounds-from-separation-robust-proxy} with a proof strategy that is essentially same as the one  employed in the proofs of Lemmas~\ref{lem:intersection-bounds-from-spectral-separation},~\ref{lem:intersection-bounds-from-mean-separation} and~\ref{lem:intersection-bounds-from-rel-frob-sep}. We will start with constraints stated in terms of the $X'$ variables and use Lemma~\ref{lem:replace-X-by-X'} at appropriate places to transition  into $X$ variables. At that point, we can plug in our argument from the previous section without change.  

We will do the case of spectral separation in detail to illustrate why this strategy works essentially syntactically. 
\begin{lemma}[Simultaneous Proxy Intersection Bounds from Spectral Separation]\label{lem:intersection-bounds-from-spectral-separation-robust}
Suppose there exists a $v$ such that $v^{\top}\Sigma(r') v > \Delta_{\spe} v^{\top}\Sigma(r') v$. Let $B = \max_{i \leq k} \frac{v^{\top} \Sigma(i)v}{v^{\top}\Sigma(r')v} \leq \kappa$.

Then, whenever $\Delta_{\spe} \gg Cs/\delta$, 
\[
\cA_{rob} \sststile{O(s\log (2B))}{w'} \Set{w'(C_r)w'(C_{r'}) \leq  O(\sqrt{\delta}) }\mper
\]


\end{lemma}
Observe, as in the previous section, that $B = 1$ when $k= 2$.

As in the previous section, we start by proving a lower-bound on the variance of $\cD_w$ in the direction $v$ where $\Sigma(r)$ and $\Sigma(r')$ are spectrally separated. This gives us:

\begin{lemma}[Large Intersection Implies High Variance, Spectral Separation]
\label{lem:spectral-lower-bound-intersection-clustering-robust}
\begin{equation}
\cA_{rob} \sststile{4s}{} \Biggl\{w'(C_{r'}) w'(C_r) \Paren{v^{\top} \Paren{\Sigma(r) + \Sigma(r')} v^{\top}}^{s} \leq \Paren{\frac{2}{\delta^{2}}}^s \Paren{v^{\top} \Sigma(w) v}^s + C \delta \Paren{v^{\top} \Paren{\Sigma(r) + \Sigma(r')} v^{\top}}^s \Biggr\}     
\end{equation}
\end{lemma}

\begin{proof}


We know from Lemma~\ref{lem:typical-samples-good} that two-sample-centered points from both $C_r$ and $C_{r'}$ (note that these are subsets of the original  uncorrupted sample $X$) are $2s$-certifiably $(\delta,C\delta)$-anti-concentrated. Using Definition~\ref{def:certifiable-anti-concentration}, thus yields: 

\begin{multline}
\cA_{rob} \sststile{4s}{} \Biggl\{ \frac{k^4}{n^4} \sum_{i_1,i_2 \in C_r, j_1,j_2 \in C_{r'}} w'_{i_1} w'_{i_2} w'_{j_1} w'_{j_2} \Iprod{x_{i_1}-x_{i_2} -x_{j_1}+x_{j_2},v}^{2s}  \\ \geq \delta^{2s} w'(C_r)^2w'(C_{r'})^2 \Paren{v^{\top} 2(\Sigma(r) + \Sigma(r')) v^{\top}}^s\\ -  \delta^{2s}\frac{k^4}{n^4} \sum_{i_1, i_2 \in C_r, j_1,j_2 \in C_{r'
}} w'_{i_1} w'_{i_2} w'_{j_1} w'_{j_2} q_{\delta,2(\Sigma(r)+\Sigma(r'))}^2 (x_{i_1}-x_{i_2} -x_{j_1}+x_{j_2}, v) \Biggr\} \label{eq:anti-conc-app-clusters-outlier}
\end{multline}

Using that $\cA_{rob} \sststile{4}{} \Set{w'_{i_1} w'_{i_2} w'_{j_1} w'_{j_2} \leq 1}$ for every $i_1,i_2,j_1,j_2$ and using $2s$-certifiable $(\delta,C\delta)$-anti-concentration of $x_{i_1}-x_{i_2} -x_{j_1}+x_{j_2}$ and invoking Definition~\ref{def:certifiable-anti-concentration}, we have:

\begin{multline}
\cA_{rob} \sststile{4s}{w', \Sigma } \Biggl\{\frac{k^4}{n^4} \sum_{i_1, i_2 \in C_r, j_1,j_2 \in C_{r'
}}  w'_{i_1} w'_{i_2} w'_{j_1} w'_{j_2} q_{\delta,2(\Sigma(r)+\Sigma(r'))}^2 (x_{i_1}-x_{i_2} -x_{j_1}+x_{j_2}, v)\\ \leq \frac{k^4}{n^4} \sum_{i_1, i_2 \in C_r, j_1, j_2 \in C_{r'
}} q_{\delta,2(\Sigma(r)+\Sigma(r'))}^2 (x_{i_1}-x_{i_2} -x_{j_1}+x_{j_2}, v) \leq  C\delta \Paren{v^{\top} 2(\Sigma(r) + \Sigma(r')) v}^s \Biggr\}
\end{multline}
 
Plugging in the above bound in \eqref{eq:anti-conc-app-clusters} gives:

\begin{multline}
\cA_{rob} \sststile{}{} \Biggl\{\frac{k^4}{n^4} \sum_{i_1, i_2 \in C_r, j_1, j_2 \in C_{r'
}} w'_{i_1} w'_{i_2} w'_{j_1} w'_{j_2} \Iprod{x_{i_1}-x_{i_2} -x_{j_1}+x_{j_2}, v}^{2s} \\ \geq \delta^{2s} \Paren{w'(C_r)^2w'(C_{r'})^2 - C\delta} \Paren{v^{\top} 2(\Sigma(r)+ \Sigma(r')) v^{\top}}^s \Biggr\} 
\end{multline}

Rearranging thus yields:
\begin{multline}
\cA_{rob} \sststile{4s}{} \Biggl\{\frac{1}{\delta^{2s}} \frac{k^4}{n^4} \sum_{i_1, i_2 \in C_r, j_1, j_2 \in C_{r'
}} w'_{i_1} w'_{i_2} w'_{j_1} w'_{j_2} \Iprod{x_{i_1}-x_{i_2} -x_{j_1}+x_{j_2}, v}^{2s}  + C\delta \Paren{v^{\top} 2(\Sigma(r) + \Sigma(r')) v^{\top}}^s\\ \geq w'(C_r)^2w'(C_{r'})^2 \Paren{v^{\top} 2(\Sigma(r)+\Sigma(r')) v^{\top}}^s \Biggr\} \label{eq:spectral-lower-bound-1-robust} 
\end{multline}

So far in the proof, the only change (compared to the proof of Lemma~\ref{lem:spectral-lower-bound-intersection-clustering}) in the proof has been that we work with the subset indicated by $w_i'$. 

The key additional step we observe now is the following consequence of $\cA_{rob} \sststile{}{} \Set{w_i'(x_i-x_i') = 0}$ (Lemma~\ref{lem:replace-X-by-X'}). 
\[
\cA_{rob} \sststile{4}{} \Set{ w'_{i_1} w'_{i_2} w'_{j_1} w'_{j2} \Iprod{x'_{i_1}-x'_{i_2} -x'_{j_1}+x'_{j_2}, v} = w'_{i_1} w'_{i_2} w'_{j_1} w'_{j2} \Iprod{x_{i_1}-x_{i_2} -x_{j_1}+x_{j_2}, v} }\mper
\]
Using further that $w_i \geq w_i'$, we have:

\begin{align*}
\cA_{rob} &\sststile{4s}{} \Biggl\{ \Paren{\frac{4cs}{\delta^{2}}}^s \Paren{v^{\top} \Sigma(w) v}^s \geq \frac{1}{\delta^{2s}} \frac{k^4}{n^4} \sum_{i_1, i_2, j_1, j_2 \in [n]} w_{i_1} w_{i_2} w_{j_1} w_{j_2} \Iprod{x'_{i_1}-x'_{i_2} -x'_{j_1}+x'_{j_2}, v}^{2s}\\
&\geq \frac{1}{\delta^{2s}} \frac{k^4}{n^4} \sum_{i_1, i_2, j_1, j_2 \in [n]} w'_{i_1} w'_{i_2} w'_{j_1} w'_{j_2} \Iprod{x'_{i_1}-x'_{i_2} -x'_{j_1}+x'_{j_2}, v}^{2s}\\ 
&\geq \frac{1}{\delta^{2s}} \frac{k^4}{n^4} \sum_{i_1, i_2, j_1, j_2 \in [n]} w'_{i_1} w'_{i_2} w'_{j_1} w'_{j_2} \Iprod{x_{i_1}-x_{i_2} -x_{j_1}+x_{j_2}, v}^{2s}\\ 
&\geq \frac{1}{\delta^{2s}} \frac{k^4}{n^4} \sum_{i_1, i_2 \in C_r, j_1, j_2 \in C_{r'}} w'_{i_1} w'_{i_2} w'_{j_1} w'_{j_2} \Iprod{x_{i_1}-x_{i_2} -x_{j_1}+x_{j_2}, v}^{2s}   \Biggr\} \mper
\end{align*}

Plugging in the upper bound above in \eqref{eq:spectral-lower-bound-1} and canceling out a copy of $2^s$ from both sides gives the lemma.

\end{proof}

The basic spectral upper bound also follows by simply shifting to the proxy variables $w_i'$. This yields us the following analog of Lemma~\ref{lem:anti-conc-upper-bound-w}:

\begin{lemma}[Spectral Upper Bound via Anti-Concentration] \label{lem:anti-conc-upper-bound-w-robust}
\begin{equation}
\cA_{rob} \sststile{4s}{} \Biggl\{ \Paren{w'(C_r)^2 - C\delta} \Paren{v^{\top} \Sigma(w) v^{\top}}^s \leq \Paren{\frac{Cs}{\delta^{2}}}^s \Paren{v^{\top} \Sigma(r) v}^s \Biggr\} 
\end{equation}

\end{lemma}

\begin{proof}

Our constraint system $\cA_{rob}$ allows us to derive that two-sample-centered points indicated by $w$ are $2s$-certifiably $(\delta,C\delta)$-anti-concentrated with witnessing polynomial $p_{\cD}$. Using Definition~\ref{def:certifiable-anti-concentration} and summing up over all $n$ after multiplying throughout by $w_i'w_j'$ yields:

\begin{multline}
\cA_{rob} \sststile{4s}{} \Biggl\{   \delta^{2s} w'(C_r)^2 \Paren{v^{\top} \Sigma(w) v^{\top}}^s \\\leq \frac{k^2}{n^2} \sum_{i,j \in C_r} w'_i w'_j \Iprod{ \frac{1}{\sqrt{2}} \Paren{x'_i - x'_j}, v}^{2s} +\delta^{2s}\frac{k^2}{n^2} \sum_{i \neq j \in C_r} w'_i w'_j q_{\delta,\Sigma(w)}^2 \Paren{ \frac{1}{\sqrt{2}} \Paren{x'_i - x'_j}, v}   \Biggr\} \label{eq:main-bound-spectral-upper}
\end{multline}

Using that $\cA_{rob} \sststile{2}{} \Set{ w_i' w_j' \Paren{(x_i' -x_j') - (x_i-x_j)} = 0}$ (two applications of Lemma~\ref{lem:replace-X-by-X'}) yields:
\begin{multline}
\cA_{rob} \sststile{4s}{\Sigma,w'} \Biggl\{   \delta^{2s} w'(C_r)^2 \Paren{v^{\top} \Sigma(w) v^{\top}}^s \\\leq \frac{k^2}{n^2} \sum_{i,j \in C_r} w'_i w'_j \Iprod{ \frac{1}{\sqrt{2}} \Paren{x_i - x_j}, v}^{2s} +\delta^{2s}\frac{k^2}{n^2} \sum_{i \neq j \in C_r} w'_i w'_j q_{\delta,\Sigma(w)}^2 \Paren{ \frac{1}{\sqrt{2}} \Paren{x_i - x_j}, v}   \Biggr\} \label{eq:main-bound-spectral-upper}
\end{multline}

Using that $\cA_{rob} \sststile{2}{} \Set{w'_i w'_j \leq 1}$ for every $i,j$, using that $\cA_{rob}$ derives $2s$-certifiable $(\delta,C\delta)$-anti-concentration of $w$-samples and invoking Definition~\ref{def:certifiable-anti-concentration}, we have:

\begin{equation}
\begin{split}
\cA_{rob} \sststile{4s}{} \Biggl\{\frac{k^2}{n^2} \sum_{i \neq j \in C_r} w'_i w'_j q_{\delta,\Sigma(w)}^2 \Paren{\frac{1}{\sqrt{2}} \Paren{x_i - x_j}, v} & \leq \frac{k^2}{n^2} \sum_{i \neq j \in [n]} w'_i w'_j q_{\delta,\Sigma(w)}^2 \Paren{ \frac{1}{\sqrt{2}} \Paren{x_i - x_j}, v} \\
& \leq  C\delta \Paren{v^{\top} \Sigma(w) v}^s \Biggr\} 
\end{split}
\end{equation}

Further, using that $\cA_{rob} \sststile{2}{} \Set{ w'_i w'_j \leq 1 }$ for all $i,j$ and relying on the certifiable Sub-gaussianity of $C_r$, we have:

\begin{equation}
\cA_{rob} \sststile{4s}{\Sigma,w'} \Biggl\{\frac{k^2}{n^2} \sum_{i,j \in C_r} w'_i w'_j \Iprod{ \frac{1}{\sqrt{2}} \Paren{x_i - x_j}, v}^{2s}  \leq \frac{k^2}{n^2} \sum_{i,j \in C_r} \Iprod{ \frac{1}{\sqrt{2}} \Paren{x_i - x_j}, v}^{2s} = \Paren{Cs}^s \Paren{v^{\top} \Sigma(r) v}^s\Biggr\} 
\end{equation}

Combining the last two bounds with \eqref{eq:main-bound-spectral-upper} thus yields:

\begin{equation}
\cA_{rob} \sststile{4s}{} \Biggl\{ w'(C_r)^2 \Paren{v^{\top} \Sigma(w) v^{\top}}^s \leq \frac{1}{\delta^{2s}} \Paren{Cs}^s \Paren{v^{\top} \Sigma(r) v}^s + C \delta  \Paren{v^{\top} \Sigma(w) v^{\top}}^s \Biggr\} 
\end{equation}

\end{proof} 

Finally, we must translate the rough spectral upper bounds we had in Lemma~\ref{lem:rough-bound-spectral-upper}. Yet again, the proof goes through essentially with only syntactic changes.

\begin{lemma}[Rough Spectral Upper bound on $\Sigma(w)$] \label{lem:rough-bound-spectral-upper-robust}
\begin{equation}
\cA_{rob} \sststile{4s}{} \Biggl\{ \Paren{v^{\top} \Sigma(w) v^{\top}}^{s} \leq (2Ck)^{s+1} \Paren{Cs}^{s} \sum_{r \leq k} \Paren{v^{\top} \Sigma(r) v}^{s} \Biggr\} 
\end{equation}
\end{lemma}
 \begin{proof}
For ease of exposition, we drop the variable and degree specifications since they are clear from context.
As before, we start by invoking our constraints to conclude:

\begin{multline}
\cA_{rob} \sststile{}{} \Biggl\{   \tau^{2s} \sum_{r \leq k} w'(C_r)^2 \Paren{v^{\top} \Sigma(w) v^{\top}}^s \\ \leq \frac{k^2}{n^2} \sum_{r\leq k} \sum_{i,j \in C_r} w'_i w'_j \Iprod{ \frac{1}{\sqrt{2}} \Paren{x'_i - x'_j}, v}^{2s} +\tau^{2s}\frac{k^2}{n^2} \sum_{r \leq k} \sum_{i \neq j \in C_r} w'_i w'_j q_{\tau,\Sigma(w)}^2 (\frac{1}{\sqrt{2}} \Paren{x'_i - x'_j}, v)   \Biggr\} \label{eq:main-bound-spectral-upper}
\end{multline}

We invoke Lemma~\ref{lem:replace-X-by-X'} to conclude:
\begin{multline}
\cA_{rob} \sststile{}{} \Biggl\{   \tau^{2s} \sum_{r \leq k} w'(C_r)^2 \Paren{v^{\top} \Sigma(w) v^{\top}}^s \\ \leq \frac{k^2}{n^2} \sum_{r\leq k} \sum_{i,j \in C_r} w'_i w'_j \Iprod{ \frac{1}{\sqrt{2}} \Paren{x_i - x_j}, v}^{2s} +\tau^{2s}\frac{k^2}{n^2} \sum_{r \leq k} \sum_{i \neq j \in C_r} w'_i w'_j q_{\tau,\Sigma(w)}^2 (\frac{1}{\sqrt{2}} \Paren{x_i - x_j}, v)   \Biggr\} \label{eq:main-bound-spectral-upper}
\end{multline}

The second term on the RHS can be upper bounded just as in the proof of Lemma~\ref{lem:anti-conc-upper-bound-w} to yield:

\begin{equation}
\begin{split}
\cA_{rob} \sststile{}{} \Biggl\{\frac{k^2}{n^2} \sum_{r \leq k} \sum_{i \neq j \in C_r} w'_i w'_j q_{\tau,\Sigma(w)}^2\Paren{\Iprod{ \frac{1}{\sqrt{2}} \Paren{x_i - x_j}, v}} & \leq \frac{k^2}{n^2} \sum_{i \neq j \in [n]} w'_i w'_j q_{\tau,\Sigma(w)}^2 \Paren{\Iprod{ \frac{1}{\sqrt{2}} \Paren{x_i - x_j}, v}} \\
& \leq  C\tau \Paren{v^{\top} \Sigma(w) v}^s \Biggr\} 
\end{split}
\end{equation}

The first term can be also be upper bounded - this time in terms of the Covariances of all the $k$ components.

\begin{equation}
\begin{split}
\cA_{rob} \sststile{}{} \Biggl\{\frac{k^2}{n^2} \sum_{r \leq k} \sum_{i,j \in C_r} w'_i w'_j \Iprod{ \frac{1}{\sqrt{2}} \Paren{x_i - x_j}, v}^{2s}  & \leq \sum_{r \leq k} \frac{k^2}{n^2} \sum_{i,j \in C_r} \Iprod{ \frac{1}{\sqrt{2}} \Paren{x_i - x_j}, v}^{2s}\\
&  = \Paren{Cs}^s \sum_{r \leq k} \Paren{v^{\top} \Sigma(r) v}^s\Biggr\} 
\end{split}
\end{equation} 

We can now combine the two estimates above to yield:

\begin{equation}
\cA_{rob} \sststile{}{} \Biggl\{ \Paren{\sum_{r \leq k} w'(C_r)^2 - C\tau} \Paren{v^{\top} \Sigma(w) v^{\top}}^s  \leq \frac{1}{\tau^{2s}} \Paren{Cs}^s \sum_{r \leq k} \Paren{v^{\top} \Sigma(r) v}^s \Biggr\} \label{eq:main-bound-spectral-upper}
\end{equation}

So far the argument closely follows the proof of Lemma~\ref{lem:anti-conc-upper-bound-w}. 

We now observe (note the change in the bound compared to the proof of Lemma~\ref{lem:rough-bound-spectral-upper})
\[
\cA_{rob} \sststile{}{} \Set{ \sum_{r \leq k} w'(C_r)^2 \geq \frac{1}{k} \Paren{\sum_{r \leq k} w'(C_r)}^2}\mper
\]
Now, 
\begin{equation}
\begin{split}
\cA_{rob} \sststile{}{} \Biggl\{ \Paren{\sum_{r \leq k} w'(C_r)}^2 = \Paren{ \frac{k}{n} \sum_{i \leq n } w_i m_i'}^2 & = \Paren{ \frac{k}{n} \sum_{i \leq n } w_i }^2 - \Paren{ \frac{k}{n} \sum_{i \leq n } w_i (1-m_i') }^2 \\
& \geq \Paren{ \frac{k}{n} \sum_{i \leq n } w_i }^2 - \Paren{ \frac{k}{n} \sum_{i \leq n }(1-m_i') }^2 \\
& \geq \Paren{ \frac{k}{n} \sum_{i \leq n } w_i }^2 - k^2 \epsilon^2 \\
& \geq 1-k^2 \epsilon^2 \Biggr\} \mper
\end{split}
\end{equation}
Thus,
\[
\cA_{rob} \sststile{}{} \Set{ \sum_{r \leq k} w'(C_r)^2 \geq \frac{1}{k} \Paren{\sum_{r \leq k} w'(C_r)}^2 \geq 1/k - k \epsilon^2}\mper
\]
Thus, as long as $\tau \ll \frac{1}{2k}$, we can derive:
\begin{equation}
\cA_{rob} \sststile{}{} \Biggl\{ \Paren{v^{\top} \Sigma(w) v^{\top}}^{s} \leq k^{s+1} \Paren{Cs}^{s} \sum_{r \leq k} \Paren{v^{\top} \Sigma(r) v}^{s} \Biggr\} \label{eq:main-bound-spectral-upper}
\end{equation}
This concludes the proof.
\end{proof}

The argument for combining the upper and lower-bounds above proceeds exactly the same as in Section~\ref{sec:clustering}.

\paragraph{Proxy Intersection Bounds from Mean and Relative Frobenius Separation.} 
The proof of the other two intersection bounds follows via similar strategy yielding:
\begin{lemma}[Simultaneous Proxy Intersection Bounds from Mean Separation]\label{lem:intersection-bounds-from-mean-separation-robust}
Suppose there exists a $v \in \R^d$ such that $\Iprod{\mu(r) -\mu(r'),v}_2^2 \geq \Delta^2_m v^{\top} \Paren{\Sigma(r) +\Sigma(r')} v$. 

Then, whenever $\Delta_m \gg Cs/\delta$, 
\[
\cA_{rob} \sststile{O(s\log  \kappa)}{w'} \Set{w'(C_r)w'(C_{r'}) \leq  O(\sqrt{\delta}) }\mper
\]
For the special case of $k=2$, whenever $\Delta_m \gg \Theta(1)$, 
\[
\cA_{rob} \sststile{O(s)}{w'} \Set{w'(C_1)w'(C_{2}) \leq  O(\sqrt{\delta}) }\mper
\]
\end{lemma}

\begin{lemma}[Simultaneous Proxy Intersection Bounds from Relative Frobenius Separation]
Suppose $\Norm{\Sigma(r')^{-1/2}\Sigma(r)\Sigma(r')^{-1/2} - I}_F^2 \geq \Delta_{cov}^2 \Paren{ \Norm{\Sigma(r')^{-1/2}\Sigma(r)^{1/2}}_{op}^4}$ for $\Delta_{cov} \gg C/\delta^2$. Then, 
\[
\cA_{rob} \sststile{O(s\log \kappa)}{w'} \Set{w'(C_r) w'(C_{r'}) \leq O(\delta^{1/3})}\mper
\]
For the special case of $k=2$, we have:
\[
\cA_{rob} \sststile{O(s)}{w'} \Set{w'(C_1) w'(C_{2}) \leq O(\delta^{1/3})}\mper
\] \label{lem:intersection-bounds-from-rel-frob-sep-robust}
\end{lemma}

Combining the above  three bounds yields Lemma~\ref{lem:intersection-bounds-from-separation-robust}.

\section{Fully Polynomial Algorithm via Recursive Partial Clustering} \label{sec:better-algo}
In this section, we describe our fully polynomial time algorithm and prove Theorem~\ref{thm:main-clustering-general-main}.

\begin{theorem}[Precise form of Theorem~\ref{thm:main-clustering-general-main}] \label{thm:main-section}
Let $\eta,\epsilon  \leq k^{-\Omega(k)}$. 
Let $\Delta \geq \poly(\eta/2^k)^k$.
Let $X$ be an i.i.d. sample from $\Delta$-separated mixture of $k$ reasonable distributions $\{\cD_r\}_{r \leq k}$ with parameters $\{\mu(r),\Sigma(r)\}_{r\leq k}$ with true clusters $C_1, C_2, \ldots, C_k$ of size $n/k$ each. 
Let $Y$ be an $\epsilon$-corruption of $X$.
Then, there exists an algorithm, that with probability $\geq 0.99$ over the draw of the sample and its random choices, takes input $Y$ and outputs a clustering $\hat{C}_1, \hat{C}_2,\ldots, \hat{C}_k$ such that there exists a permutation  $\pi:[k]\rightarrow [k]$ satisfying:
\[
\min_{i\leq k}\frac{k}{n}|\hat{C}_i \cap C_{\pi(i)}| \geq 1-O(k^{O(k)}(\eta+\epsilon))\mper
\]
\end{theorem}

\paragraph{Discussion} In Section~\ref{sec:outlier-robust-clustering}, we proved that our simple rounding (Algorithm~\ref{algo:rounding-for-pseudo-distribution-robust}) of any pseudo-distribution $\tzeta$ of degree $\Omega( s(\poly(\eta/k)) \log (\kappa))$ consistent with $\cA_{rob}$ produces an approximately correct clustering of any $\epsilon$-corruption $Y$ of a good sample $X$. In this section, we will establish two somewhat curious technical facts about Algorithm~\ref{algo:rounding-for-pseudo-distribution-robust} and the constraints $\cA_{rob}$ to show Theorem~\ref{thm:main-clustering-general-main}.

\begin{enumerate}
	\item \emph{All is not lost in constant degree} (Lemma~\ref{lem:partial-cluster-recovery}). When the rounding in Algorithm~\ref{algo:rounding-for-pseudo-distribution-robust} is run on a pseudo-distribution $\tzeta$ of degree $O(s(\poly(\eta/k)))$ consistent with $\cA_{rob}$, it still contains non-trivial information about the true clusters and in particular can be used to construct a \emph{partial clustering}.
	\item \emph{Verification can be done in constant degree} (Lemma~\ref{lem:soundness-verification-main}). While we cannot show that degree $O( s(\poly(\eta/k)))$ is enough to \emph{find} a clustering, we will prove that it is enough to \emph{verify} a purported approximate clustering.
\end{enumerate}
These facts let us use a slightly more complicated recursive clustering algorithm combined with a verification subroutine to obtain an outlier-robust clustering algorithm with no dependence on the spread $\kappa$ in the running time. 


\paragraph{Algorithm} Our algorithm is the following recursive clustering subroutine that we invoke with the input corrupted sample $Y$ and outlier parameter $\epsilon$. The base case of the recursion uses a verification subroutine that confirms if a subset of $n/k$ samples is close to a true cluster. The main recursive step employs the exact same rounding of the pseudo-distribution that we used in  Algorithm~\ref{algo:rounding-for-pseudo-distribution-robust}. 


\begin{mdframed}
  \begin{algorithm}[Recursive Partial Clustering]
    \label{algo:polynomial-rounding-for-pseudo-distribution-robust}\mbox{}
    \begin{description}
    \item[Given:]
        A subsample $Y' \subseteq Y$ of size $j n/k$ for $j \in [k]$. A outlier parameter $\tau > 0$ and a accuracy parameter $\eta > 0$. 
    \item[Output:]
      A partition of $Y'$ into an approximately correct clustering $\hat{C}_1, \hat{C}_2, \ldots, \hat{C}_j$.
    \item[Operation:]\mbox{}
    \begin{enumerate}
    	\item \textbf{Base Case:} If $|Y'| = n/k$, accept if Algorithm~\ref{algo:verification-subroutine} applies to $Y'$ with outlier parameter $\tau$ accepts. Otherwise output fail.
   		\item \textbf{SDP Solving:} Find a pseudo-distribution $\tzeta$ satisfying $\cA_{rob}$ minimizing $\Norm{\pE[w]}_2^2$ with number of components set to $j$ and outlier parameter set to $\tau$. If no such pseudo-distribution exists, output fail.
      	\item \textbf{Rounding:} Let $M = \pE_{w \sim \tzeta} [ww^{\top}]$. 
       		\begin{enumerate} \item Choose a uniformly random row $i$ of $M$.
        		\item Choose $\ell = O(k \log (k/\eta))$ rows of $M$ uniformly at random and independently. 
       			\item For each $i \leq \ell$, let $\hat{C}_i$ be the indices of the columns $j$ such that $M(i,j) \geq \eta/\poly(k)$.
            \item Let $\hat{C}_{\ell+1} = [n] \setminus \cup_{i \leq \ell} \hat{C}_i$.
       		\end{enumerate}
    	\item \textbf{Brute-Force Search Over Partial Clusterings:} For each subset $S \subseteq [\ell+1]$, recursively run two instances of Algorithm~\ref{algo:polynomial-rounding-for-pseudo-distribution-robust} with inputs $\cup_{i \in S} \hat{C}_i$, $\cup_{i \not \in S} \hat{C}_i$ respectively with outlier parameters $\eta + O(k^3\tau)$ for both runs.
    	\item If either run fails, output fail and return. Otherwise output the union of clusters returned by the two runs of the algorithm.
      \end{enumerate}
    \end{description}
  \end{algorithm}
\end{mdframed}

\paragraph{Analysis of Algorithm.}
The analysis of our algorithm is based on the following two key pieces. 
The first shows that Algorithm~\ref{algo:rounding-for-pseudo-distribution-robust}, when run with a pseudo-distribution $\tzeta$ of degree $O(s(\poly(\eta/k)))$ consistent with $\cA_{rob}$ recovers a \emph{partial} clustering of the input sample. An (approximate) partial clustering is a non-trivial split of $Y$ into (approximate) unions of clusters. 

\begin{definition}[Partial Clustering]
A $\tau$-approximate partial clustering of $Y  =  C_1 \cup C_2 \cup \ldots C_k \subseteq \R^d$ described by a partition of $Y$ into $P_1 \cup P_2$ such that there exists $S \subseteq [k]$, $0 < |S| < k$ satisfying $\frac{|P_1 \cap \cup_{i \in S} C_i|}{|\cup_{i \in S} C_i|},\frac{|P_2 \cap \cup_{i \not \in S} C_i|}{|\cup_{i \not \in S} C_i|} \geq 1-\tau$. 
\end{definition}

The following lemma analyzes the output of Algorithm~\ref{algo:polynomial-rounding-for-pseudo-distribution-robust} when run with a $\tau$-corrupted mixture of $k' \leq k$ reasonable distributions. We will use it to analyze all runs of Algorithm~\ref{algo:polynomial-rounding-for-pseudo-distribution-robust}.


\begin{lemma}[Outlier-Robust Partial Cluster Recovery]
\label{lem:partial-cluster-recovery}
Let $X$ be a good sample from a $\Delta$-separated mixture of reasonable distributions with parameters $\{\mu(r),\Sigma(r)\}_{r \leq k}$ and true clusters $C_1, C_2, \ldots, C_{k'}$ of size $\frac{n}{k}$ each.
Let $Y$ be a $\tau$-corruption of $X$.
Then, whenever $\Delta \geq \poly(\eta/{k'})^{k'}$, Algorithm~\ref{algo:polynomial-rounding-for-pseudo-distribution-robust} with probability at least $1-2^{-\Omega(k)}$ recovers a clustering $\hat{C}_1, \hat{C}_2,\ldots, \hat{C}_{k'}$ such that there exists a partition $G_S \cup G_L = [k]$ such that for $P_1 = \cup_{j \in G_S} \hat{C}_j$ and $P_2 = \cup_{j \in G_L} \hat{C}_j$ form a $\eta+O(k^3\tau)$-approximate partial clustering of $Y$.
\end{lemma}

The next step is a \emph{verification subroutine} that, in polynomial (degree depending only on $k,\eta$) time verifies if a given subset of $n/k$ samples intersects in a true cluster in $(1-\tau)$ fraction of points. 

\begin{lemma}[Verification Subroutine] \label{lem:soundness-verification-main}
Let $X$ be a good sample from a $\Delta$-separated mixture of reasonable distribution with parameters $\{\mu(r),\Sigma(r)\}_{r \leq k}$ and equal-size true clusters $C_1, C_2, \ldots, C_k$.
Let $Y$ be a $\tau$-corruption of $X$.
Let $\hat{C} \subseteq Y$ be such that $\max_{j \leq k} \frac{k}{n}  |\hat{C} \cap C_j| < 1-2k\sqrt{\tau}$.
Then, Algorithm~\ref{algo:verification-subroutine} rejects on input $\hat{C}$.
On the other hand, if $\exists r \leq k$ such that $\frac{k}{n} |\hat{C} \cap C_r| \geq 1-\tau$,  Algorithm~\ref{algo:verification-subroutine} accepts on input $\hat{C}$.
\end{lemma}

We can complete the analysis of Algorithm~\ref{algo:polynomial-rounding-for-pseudo-distribution-robust} and prove Theorem~\ref{thm:main-section} using the above results:

\begin{proof}[Proof of Theorem~\ref{thm:main-section}]
We run Algorithm~\ref{algo:polynomial-rounding-for-pseudo-distribution-robust} with input $Y$ and initial outlier parameter $\tau = \epsilon$. Let's track the outlier parameters in the recursive calls - in each recursive call, $\tau \rightarrow \eta + O(k\tau)$.  Since the depth of our recursive calls is at most $k$, $\tau = O(k^k \eta + k^{3k} \epsilon)$ throughout the algorithm. 

Let's bound the running time of the algorithm. The base case requires running the verification algorithm that needs $n^{O(s(\poly(\tau/C))}$ time for $\tau = O(k^k(\eta + \epsilon))$. Each run of the algorithm makes at most $2^{k}$ recursive calls to instances with number of components reduced by at least $1$ and needs to solve an SDP that needs $n^{O(s(\eta/k))}$ time. Thus, the running time follows the recurrence: $T(j) \leq 2^k T(j-1) + n^{O(s(\eta/k))}$. Thus  the running time $T(k) \leq 2^{k^2} T(1) = 2^{k^2} n^{O(s(\poly(\eta/k))}$. 

Finally, let's confirm the correctness of the procedure. First, we show that if the algorithm doesn't fail, then it outputs a correct approximate clustering $\hat{C}_1, \hat{C}_2, \ldots, \hat{C}_k$ of $Y$. It's immediate that the algorithm always produces a partitioning of $Y$ into subsets of size  $n/k$ each. Further, each $\hat{C}_i$ must cause Algorithm~\ref{algo:verification-subroutine} to accept (base case of Algorithm~\ref{algo:polynomial-rounding-for-pseudo-distribution-robust}). From Lemma~\ref{lem:soundness-verification-main}, it must hold for each $i$, $\hat{C}_i$ some cluster $C_{\pi(i)}$ in $1-\tau$ fraction of the  $n/k$ samples for $\tau = O(k^{2k} \eta + k^{2k} \epsilon)$. Finally, observe that if $\tau  \ll 1/k$ then,  $C_{\pi(i)} \neq  C_{\pi(j)}$ for $i \neq j$. Thus, $\pi$ must be a permutation of $[k]$. This finishes the proof.

What remains is to argue that when run with $\epsilon$-corruption $Y$ of a good sample $X$, Algorithm~\ref{algo:polynomial-rounding-for-pseudo-distribution-robust} does not output fail with probability at least $0.99$. For this, we need to exhibit a choice of $S \subset[k]$ for each recursive call for which the algorithm does not fail. Observe, our algorithm never outputs fail if the input $Y'$ intersects $(1-\tau)$ fraction of samples in some union of true clusters. This is guaranteed by Lemma~\ref{lem:partial-cluster-recovery} with probability at least $1-2^{-\Omega(k)}$. By a union bound, this guarantee holds for the output of all rounding steps incurred by making the choices of $S$ above with  probability at least $0.99$. Thus, we must arrive at subsets $\hat{C}_i$ that are $(1-\tau)$-intersecting with some true cluster for $\tau = O(k^{k} \eta + k^{k} \epsilon) \ll 1/k^2$. By the completeness of our verification subroutine (Lemma~\ref{lem:soundness-verification-main}), all $\hat{C}_i$  produced via these  choices cause the verification algorithm to accept. 
This completes the proof.
\end{proof}

\subsection{Partial Cluster Recovery}
In this section, we prove Lemma~\ref{lem:partial-cluster-recovery}. 
The crux of the proof is the following intersection bound that finds a bipartition of clusters and proves that the simultaneous intersection of $\hat{C}$ (searched for in $\cA_{rob}$ via $w$-variables) with the two pieces of the bipartition is small. Note that this gets  us a \emph{weaker} guarantee than the inter-cluster  simultaneous intersection bounds proven in Sections~\ref{sec:clustering} and ~\ref{sec:outlier-robust-clustering} with the upshot that the degree of the SoS proof here does not depend on $\kappa$, the spread of the mixture. 

\begin{lemma}[Simultaneous Intersections Bounds Across Bipartition] \label{lem:sim-intersection-bipartition-bound}
Let $X, Y$ be as in the setting of Lemma~\ref{lem:partial-cluster-recovery} with true clusters $C_1, C_2, \ldots, C_{k'}$ with $\eta = O(1/k)$ and $\Delta = \Delta_{rob}^k = \poly(\eta/k')^{k'}$ where  $\Delta_{rob}$ is the separation requirement in Lemma~\ref{lem:intersection-bounds-from-separation-robust}.
There exists a partition $S \cup L =[k']$ such that $|S| < k'$ satisfying:
\[
\cA \sststile{}{} \Set{ \sum_{r \in S, r' \in L} w(C_r) w(C_{r'}) \leq  O(k^2\delta^{1/3} + k \eta) }\mper
\]
\end{lemma}

\begin{proof}

We break the proof into two cases.

\textbf{Case 1}: No pair of clusters $C_{r},C_{r'}$ is spectrally separated.
In this case, for every direction $v$, either $v^{\top} \Sigma(i) v = 0$ for all $i \leq k'$ or $\frac{v^{\top} \Sigma(r) v}{v^{\top} \Sigma(r') v} \leq \Delta \leq O(s(\poly(\eta/k))^k)$ for all $r,r'$. Thus, in particular, the spread $\kappa \leq \Delta$. Applying Lemma~\ref{lem:intersection-bounds-from-separation-robust} and plugging in the upper bound on $\kappa$ immediately yields that for every $1 \leq r < r' \leq k'$
\Anote{the corruption in this section is being defined as $\eta$. }
\[
\cA \sststile{O( {k'}^2 s^2 \poly \log (s))}{w} \Set{\sum_{r \neq r'} w(C_r) w(C_{r'}) \leq O(k'\tau) + \eta}\mper
\]

Thus, in this case, we recover every cluster approximately and thus can set $S$ and $L$ to be any non-trivial partition (that is, both $S$ and $L$ are non-empty) and finish the proof.

\textbf{Case 2}:
There exist $r,r'$ such that $C_r$ and $C_{r'}$ that are spectrally separated.
Then there is a direction $v$ such that $\Delta_{rob}^k v^{\top} \Sigma(r) v \leq  v^{\top} \Sigma(r') v$.
Consider an ordering of the true clusters along the direction $v$, renaming cluster indices if needed, such that $v^{\top} \Sigma(1) v \leq v^{\top}\Sigma(2)v \leq \ldots v^{\top}\Sigma(k')v$.
Then, clearly, $v^{\top}\Sigma(k')v \geq \Delta_{rob} v^{\top}\Sigma(r)v$.

Let $j \leq k'$ be the largest integer such that $\Delta_{rob}v^{\top}\Sigma(j)v \leq v^{\top}\Sigma(j+1)v$.
Observe that since we are in Case 2, such a $j$ exists. Further, observe that since $j$ is defined to be the largest index which incurs separation $\Delta_{rob}$, all indices in $[j,k']$ have spectral bound at most $\Delta_{rob}$ and thus $\frac{v^{\top} \Sigma(k')v}{v^{\top}\Sigma(j)v} \leq \Delta_{rob}^{{k'}}$.
Applying Lemma~\ref{lem:intersection-bounds-from-spectral-separation-robust} with the above direction $v$ to every $r <j$ and $r' \geq j$ and observing that the parameter $B$ in each case is at most $\frac{v^{\top} \Sigma(k')v}{v^{\top}\Sigma(j)v} \leq \Delta_{rob}^{k'}$ yields:

\[
\cA \sststile{ O({k'}^2 s^2 \poly \log (s) )}{} \Set{w(C_r)w(C_{r'}) \leq O(k'\tau) + \eta}\mper
\]

Adding up the above inequalities over all $r \leq j-1$ and $r' \geq j+1$ and taking $S = [j-1]$, $T = [k']\setminus [j-1]$ yields the claim.
\end{proof}

We are now ready to prove Lemma~\ref{lem:partial-cluster-recovery}.

\begin{proof}[Proof of Lemma~\ref{lem:partial-cluster-recovery}]

We will prove that whenever $\Delta \geq \Delta_{rob}^k = \poly(\eta/k)^k$, Algorithm~\ref{algo:polynomial-rounding-for-pseudo-distribution-robust}, when  run with input $Y$ recovers a collection $\hat{C}_1, \hat{C}_2,\ldots, \hat{C}_{\ell}$ of subsets of indices such that there is a partition $S \cup L = [\ell]$, $0 < |S| <\ell$ satisfying:
\begin{equation}
\min\Biggl\{ \frac{k}{n} |\hat{C}_i \cap \cup_{j \in S} C_j|, \frac{k}{n} |\hat{C}_i \cap \cup_{j \in L} C_j| \Biggr\} \leq \eta + O(k^3 \tau)\mper \label{eq:partial-recovery-clusterwise}
\end{equation}

This suffices to complete the proof:
Split $[\ell]$ into two groups $G_S, G_L$ as follows.
For each $i$, let $j = \argmax_{r \in [\ell]} \frac{k}{n} |\hat{C}_i \cap C_r|$.
If $j \in S$, add it to $G_S$, else add it to $G_L$.
Observe that this process is well-defined. To see this, suppose $j \in S$. Let $j' \in L$. Then, applying Lemma~\ref{lem:partial-cluster-recovery} and using that $\eta + O(k\tau) \ll 1/k$ and that $\frac{k}{n} |\hat{C}_i \cap \cup_{r \in S} C_r| \geq \frac{k}{n} |\hat{C}_i \cap C_j| \geq 1/k$, we have that: $\frac{k}{n} |\hat{C}_i \cap \cup_{j' \in L} C_{j'}| < 1/k$.

We are now ready to verify the first claim. The second follows immediately from the first.
For each $i \in G_S$, we have that $\frac{k}{n} |\hat{C}_i \cap \cup_{j \in L} C_j| \leq \eta + O(k\tau)$.
Adding up these inequalities for all $i \in S$ yields that $\frac{k}{n} |P_1 \cap \cup_{j \in L} C_j| \leq |S| \Paren{\eta + O(k\tau)}$. Using that $|P_1| = |S| \frac{n}{k}$ and $S, L$ form a partition of $[k]$ completes the proof.

We now go ahead and establish \eqref{eq:partial-recovery-clusterwise}. Let $\tzeta$ be a pseudo-distribution satisfying $\cA$ of degree $\poly(k/\eta)$. Let $M = \pE_{\tzeta}[ww^{\top}]$. Reasoning similarly as in the proof of Theorem~\ref{thm:main-clustering-section}, we have:

\begin{enumerate}
  \item $1/k \geq M(i,j) \geq 0$ for all $i,j$,
  \item $M(i,i) = 1/k$ for all $i$,
  \item $\E_{j \sim [n]} M(i,j) = \frac{1}{k^2}$ for every $i$.
\end{enumerate}
For an $\eta'$ to be chosen later, call an entry of $M$ large if it exceeds $\eta'/k^2$. For each $i$, let $B_i$ be the set of large entries in row $i$ of $M$. Then, using (3) and (1) above gives that $|B_i| \geq (1-k\eta')n/k$ for each $1 \leq  i \leq n$. Next, call a row $i$ ``good'' if $\frac{k}{n} \min \{ \abs{\cup_{r \in L} C_r \cap B_i}, \abs{\cup_{r' \in S} C_{r'} \cap B_i}\} \leq 100 k^2\eta' + O(k^3 \tau)$. Let us estimate the fraction of rows of $M$ that are good. 

Towards that goal, let's apply Lemma~\ref{lem:sim-intersection-bipartition-bound} with $\eta = \eta'/2k$ and $\delta = {\eta'}^3/8k^6$. Then, using Fact~\ref{fact:sos-completeness}, we obtain $\sum_{r \in S, r' \in L}\E_{i \in C_r} \E_{j \in C_{r'}} M(i,j) \leq \sum_{r'\neq r} \E_{i \in C_r} \E_{j \in C_{r'}} \pE[w_i w_j] = \pE[w(C_r)w(C_{r'})] \leq \eta'+O(k\tau)$. Using Markov's inequality $1-1/100k^2$ over the uniformly random choice of $i$, $\E_{ j \in C_{r'}} M(i,j) \leq 100 k^2\eta' + O(k^3 \tau)$. Thus, $1-1/100 k^2$ fraction of the rows of $M$ are good.

Next, let $R$ be the set of $100 k \log k/\eta'$ rows sampled in the run of the algorithm and set $\hat{C}_i = B_i$ for every $i \in R$. The probability that all of them are good is then at least $(1-1/100k^2)^{k \log k/\eta'} \geq 1-\eta'\log k/100k$. Let's estimate the probability that $|\cup_{i \in R} \hat{C}_i| \geq (1-1/k^{10})n$. The chance that a given point $t \in B_i$ for a uniformly random $B_i$ is at least $(1-k\eta')/k$. Thus, the chance that $t \not \in \cup_{i \in R}B_i$ is at most $(1-1/2k)^{100 k \log k/\eta'} \leq \eta'/k^{50}$. Thus, the expected number of $t$ that are not covered by $\cup_{i \in R} \hat{C}_i$ is at most $n \eta' /k^{50}$. Thus, by Markov's inequality, with probability at least $1-1/k^{10}$, $1-\eta'/k^{40}$ fraction of $t$ are covered in $\cup_{i \in R} \hat{C}_i$. 

Let's now condition on the events that 1) each of the $100 k \log k/\eta'$ rows $R$ sampled are good and 2) $|\cup_{i \in R} \hat{C}_i| \geq (1-\eta'/k^{40})n$.  By the above computations and a union bound, this event happens with probability at least $1-\eta'/k^{10}$. Let $\hat{C}_{\ell+1} =[n] \setminus \cup_{i \leq \ell} \hat{C}_i$ be the set of indices that are not covered in $\cup_{i \in R} \hat{C}_i$. Then, $\cup_{i \leq \ell+1} \hat{C}_i$ is a partition of $[n]$.

We will show that the following way of grouping this partition into two buckets: $R_L = R \cap \cup_{i \in L} C_i$ and $R_S = R \setminus R_L$ satisfies the requirements of the lemma. To see this, note that $\abs{\cup_{i \in R_L} \hat{C}_i \cap \cup_{i \in S} C_i} \leq n/k 100 k^3 \eta' +O(k^3\tau)$. Similarly, $\abs{\cup_{i \in R_S} \hat{C}_i \cup P \cap \cup_{i \in L} C_i} \leq n/k 100 k^3 \eta' + |P| \leq n/k (100 k^3 \eta' + \eta' k^{-40})+ O(k^3\tau)$. 

Choosing $\eta' \leq \eta/k^{10}$ completes the proof. 

\end{proof}

\subsection{Verification Algorithm} \label{sec:verification}
In this section, we prove Lemma~\ref{lem:soundness-verification-main}. 
We first describe our verification algorithm that involve computing (if one exists) a pseudo-distribution consistent with a system of constraints that verifies the properties of being close to a reasonable distribution for a given input subset $\hat{C}$ of size $n/k$ of $Y$.

We first describe the verification constraint system $\cV = \cV(\hat{C})=\cV_1\cup \cV_2 \cup\cV_3 \cup \cV_4  \cup \cV_5$ that is closely related to those used in Sections~\ref{sec:clustering} and~\ref{sec:outlier-robust-clustering}. Covariance constraints introduce a matrix valued indeterminate intended to be the square root of $\Sigma$.
\begin{equation}
\text{Covariance Constraints: $\cV_1$} = 
  \left \{
    \begin{aligned}
      &
      &\Pi
      &=UU^{\top}\\
      &
      &\Pi^2
      &=\Sigma\\
    \end{aligned}
  \right \}
\end{equation}
The intersection constraints force that $X'$ be close to $X$.
\begin{equation}
\text{Intersection Constraints: $\cV_2$} = 
  \left \{
    \begin{aligned}
      &\forall i\in [n'],
      & m_i^2
      & = m_i\\
      &&
      \textstyle\sum_{i\in[n']} m_i
      &= (1-\tau) n'\\
      &\forall i \in [n'],
      &m_i (y_i-x'_i)
      &= 0
    \end{aligned}
  \right \}
\end{equation}
The parameter constraints create indeterminates to stand for the covariance $\Sigma$ and mean $\mu$ of $\hat{C}$ (indicated by $w$).
\begin{equation}
\text{Parameter Constraints: $\cV_3$} = 
  \left \{
    \begin{aligned}
      &
      &\frac{1}{n'}\sum_{i = 1}^{n'} \Paren{x'_i-\mu}\Paren{x'_i-\mu}^{\top}
      &= \Sigma\\
      &
      &\frac{1}{n'}\sum_{i = 1}^{n'} x'_i
      &= \mu\\
    \end{aligned}
  \right \}
\end{equation}


Finally, we enforce certifiable anti-concentration and hypercontractivity of $\hat{C}$.
\begin{equation}
\text{Certifiable Anti-Concentration : $\cV_4$} =
  \left \{
    \begin{aligned}
      &
      &\frac{1}{{n'}^2}\sum_{i,j=  1}^{n'} q_{\tau/C,2\Sigma}^2\left(\Paren{x'_i-x'_j},v\right)
      &\leq 2^{s(\tau/C)} \tau \Paren{v^{\top}\Sigma v}^{s(\tau/C)}
     \end{aligned}
    \right\}
 \end{equation}

\begin{equation}
\text{Certifiable Hypercontractivity : $\cV_5$} = 
  \left \{
    \begin{aligned}
     &\forall j \leq 2s(\tau/C),
     &\frac{1}{{n'}^2} \sum_{i,\ell \leq n'} Q(x'_i-x'_\ell)^{2j}
     &\leq (Cj)^{2j}2^{2j}\Norm{\Pi Q \Pi}_F^{2j}
    \end{aligned}
  \right \}
\end{equation}
\text{Certifiable Bounded Variance: $\cV_6$} = 
\begin{equation}
  \left \{
    \begin{aligned}
     &\forall j \leq 2s,
     &\frac{k^2}{n^2} \sum_{i,\ell \leq n} w_i w_\ell \Paren{Q(x'_i-x'_\ell)-\frac{k^2}{n^2} \sum_{i,\ell \leq n} w_i w_\ell Q(x'_i-x'_\ell)}^{2}
     &\leq C \Norm{\Pi Q \Pi}_F^{2}\mper
    \end{aligned}
  \right \}
\end{equation}

\begin{mdframed}
  \begin{algorithm}[Verification Subroutine]
    \label{algo:verification-subroutine}\mbox{}
    \begin{description}
    \item[Given:]
        A purported cluster $Y = \hat{C}$ of size $n'= \frac{n}{k}$.
    \item[Output:]
        Accept or Reject.
    \item[Operation:]\mbox{} Accept iff $\exists$ a pseudo-distribution $\tzeta$ of degree $4s(\tau/C)$ consistent with $\cV(\hat{C})$.
    \end{description}
  \end{algorithm}
\end{mdframed}

\paragraph{Analysis of Verification Subroutine}

Let $m_i' = m_i \cdot \1 (y_i = x_i)$ for every $i$. 
Define $m'(C_i)  = \frac{k}{n} \sum_{j \in C_i} m'_j$ for every $i$. 

Our proof of Lemma~\ref{lem:soundness-verification-main} will rely on the following three lemmas that give a degree $O(s(\tau/C))$ refutation of $\cV(\hat{C})$ whenever $\hat{C}$ intersects at least two clusters appreciably. The proofs follow the same conceptual plan of combining an upper and lower bound on the variance of $v^{\top}\Sigma v$ as in Sections~\ref{sec:clustering} and ~\ref{sec:outlier-robust-clustering}. The key difference,  as we suggested earlier, is  that the degree of the proof is a fixed constant (instead of growing with $\log \kappa$). The proof  exploits the fact that in the verification setting, $\hat{C}$  is \emph{not} a variable in our constraint system.

\begin{lemma}[SoS Refutation from Simultaneous Intersection with Spectrally Separated Components]  \label{lem:spectral-separated-refutation}
Let $X$ be a good sample from a $\Delta$-separated reasonable distribution with parameters $\{\mu(r),\Sigma(r)\}_{r \leq k'}$ and true clusters $C_1, C_2, \ldots, C_{k'}$ of size $\frac{n}{k}$ each.
Let $Y$ be a $\tau$-corruption of $X$.
Let $\hat{C} \subseteq Y$ be a subset of size $\frac{n}{k}$.
Suppose $C_r,C_{r'}$ are $\Delta$-spectrally separated and $\frac{k}{n} |\hat{C} \cap C_r|, \frac{k}{n} |\hat{C} \cap C_{r'}|\geq 2\sqrt{\tau}$. Then, whenever $\Delta \geq \frac{1}{\tau^6}$, 
Then,
\[
\cV(\hat{C}) \sststile{4s(\tau/C)}{} \Set { -1 \geq 0}\mper
\]

\end{lemma}

\begin{lemma}[SoS Refutation from Simultaneous Intersection with Mean Separated Components]\label{lem:mean-separated-refutation}
Let $X$ be a good sample from a $\Delta$-separated reasonable distribution with parameters $\{\mu(r),\Sigma(r)\}_{r \leq k'}$ and true clusters $C_1, C_2, \ldots, C_{k'}$ of size $\frac{n}{k}$ each.
Let $Y$ be a $\tau$-corruption of $X$.
Let $\hat{C} \subseteq Y$ be a subset of size $\frac{n}{k}$.
Suppose $C_r,C_{r'}$ are $\Delta$-mean separated and $\frac{k}{n} |\hat{C} \cap C_r|, \frac{k}{n} |\hat{C} \cap C_{r'}|\geq 2\sqrt{\tau}$. Then, whenever $\Delta \geq \frac{1}{\tau^6}$, 
Then,
\[
\cV(\hat{C}) \sststile{4s(\tau/C)}{} \Set { -1 \geq 0}\mper
\]

\end{lemma}

\begin{lemma}[SoS Refutation from Simultaneous Intersection with Frobenius Separated Components]\label{lem:frob-separated-refutation}
Let $X$ be a good sample from a $\Delta$-separated reasonable distribution with parameters $\{\mu(r),\Sigma(r)\}_{r \leq k'}$ and true clusters $C_1, C_2, \ldots, C_{k'}$ of size $\frac{n}{k}$ each.
Let $Y$ be a $\tau$-corruption of $X$.
Let $\hat{C} \subseteq Y$ be a subset of size $\frac{n}{k}$.
Suppose $C_r,C_{r'}$ are $\Delta_{cov}$-relative Frobenius separated and $\frac{k}{n} |\hat{C} \cap C_r|, \frac{k}{n} |\hat{C} \cap C_{r'}|\geq 2\sqrt{\tau}$. Then, whenever $\Delta_{cov} \geq \frac{1}{\tau^6}$, 
Then,
\[
\cV(\hat{C}) \sststile{4s(\tau/C)}{} \Set { -1 \geq 0}\mper
\]
\end{lemma}

\begin{proof}[Proof of Lemma~\ref{lem:soundness-verification-main}]

Let $j$ be the maximizer of $|\hat{C} \cap C_r|$ over all $r \leq k'$. Then, $|\hat{C} \cap C_j| \geq 1/k$.
Let $j'$ be the maximizer of $|\hat{C} \cap C_r|$ over all $r \neq j$. Then, $|\hat{C} \cap C_{j'}|\leq |\hat{C} \cap C_j|$.
Then, observe that $\frac{k}{n} |\hat{C} \cap C_{j'}| \geq 2k\sqrt{\tau}/k \geq 2 \sqrt{\tau}$.

Applying Lemmas~\ref{lem:spectral-lower-bound-intersection-clustering},~\ref{lem:mean-separated-refutation} and~\ref{lem:frob-separated-refutation} for each of the three possible ways that $C_i$ and $C_j$ could be separated, we obtain that:

\[
\cV \sststile{4s(\tau/C)}{} \Set{-1  \geq 0}\mper
\]

This immediately implies that there's no degree $\geq 4s(\tau/C)$ pseudo-distribution  $\tzeta$ consistent with  $\cV(\hat{C})$ -for if there was one, then the above inequality yields a contradiction. This completes the proof of the first part. 

For the  second  part, observe that setting $X'$ to be the cluster closest (and thus  $1-\tau$-intersecting) to $\hat{C}$ immediately completes  the proof.

\end{proof}

\paragraph{Sum-of-Squares Refutation of Reasonableness of Bad Clusters}
We now prove Lemmas~\ref{lem:spectral-separated-refutation},~\ref{lem:mean-separated-refutation} and~\ref{lem:frob-separated-refutation}. The proof of these lemmas closely resembles our proofs of the simultaneous intersection bounds in Sections~\ref{sec:clustering} and ~\ref{sec:outlier-robust-clustering}. So it may appear somewhat confusing as to how we can get the SoS  proofs  to work in degrees that do not depend on $\kappa$. The key difference is that, informally speaking, here we already ``know'' that two clusters have large intersection with  a purported bad cluster $\hat{C}$ (which is given to us, not a variable) and our goal is to obtain a contradiction from the axioms that $\hat{C}$ satisfies $\cV$ in low-degree SoS. Such a difference, while inconsequential in ``ordinary math'', is key to obtaining the stronger degree bounds that do not depend on $\kappa$ in this section. 

We will use the following result in all the three proofs. 

\begin{lemma}[Matching with Original Uncorrupted Samples] \label{lem:intersection-X-X'-verification}
Suppose $\frac{1}{n'} |\hat{C} \cap C_r|,\frac{1}{n'} |\hat{C} \cap C_{r'}|\geq 2 \sqrt{\tau}$.
Let $m'(C_r) = \frac{1}{n'} \sum_{i \leq C_r} m_i'$. 
Then,
\[
\cV \sststile{4s(\tau/C)}{} \Set{m'(C_r)^2 \geq \frac{k}{n}|C_{r} \cap \hat{C}| - 2\tau \geq 2\tau}\mper
\] 

\end{lemma}
\begin{proof}
Reasoning as in Lemma~\ref{lem:intersection-X-X'}, we obtain that for any subset $\hat{C}' \subseteq \hat{C}$, we have: 
Then, 
\[
\cV(\hat{C}) \sststile{2}{m'} \Set{\frac{1}{n'}\sum_{i \in \hat{C}'} m_i' \geq |\hat{C}'|-2\tau }\mper
\]

Applying this to subsets $\hat{C}' = \hat{C} \cap C_r$ yields:
\[
\cV \sststile{4s(\tau/C)}{} \Set{m'(C_r)^2 \geq \frac{k}{n}|C_{r} \cap \hat{C}| - 2\tau \geq 2\tau}\mper
\] 

\end{proof}



\begin{proof}[Proof of Lemma~\ref{lem:spectral-separated-refutation}]
WLOG, assume $\Delta v^{\top}\Sigma(r)v \leq v^{\top} \Sigma(r')v$ for some $v \in \R^d$. The proof follows by from combining certifiable anti-concentration constraints $\cV_4$, certifiable anti-concentration of $C_r$ and  Lemma~\ref{lem:intersection-X-X'-verification}. We will use $\cV$ to denote $\cV(\hat{C})$ in the proof below.

Using certifiable anti-concentration of $C_{r'}$:
\begin{equation} \label{eq:verification-anti-conc-1}
\cV \sststile{4s(\tau/C)}{} \Biggl\{ (C\tau)^{2s} \Paren{m'(C_{r'})^2 - \tau} \Paren{v^{\top} \Sigma(r') v^{\top}}^s \leq  \frac{1}{{n'}^2} \sum_{i,j \leq C_{r'}} m_i'm_j' \iprod{x_i'-x_j',v}^{2s} \leq  \Paren{v^{\top} \Sigma(m) v}^s \Biggr\} 
\end{equation}

Similarly, using certifiable anti-concentration constraints $\cV_4$:
\begin{equation} \label{eq:verification-anti-conc-2}
\cV \sststile{4s(\tau/C)}{} \Biggl\{ \Paren{m'(C_r)^2 - \tau} \Paren{v^{\top} \Sigma(m) v^{\top}}^s \leq \Paren{\frac{1}{\tau^2}}^s \Paren{v^{\top} \Sigma(r) v}^s \Biggr\} 
\end{equation}

Plugging in the estimates from Lemma~\ref{lem:intersection-X-X'-verification} in \eqref{eq:verification-anti-conc-1} and \eqref{eq:verification-anti-conc-2}, and rearranging yields:

\[
\cV \sststile{4s(\tau/C)}{} \Set{\tau^2 (C\tau)^{2s} \Paren{v^{\top} \Sigma(r') v^{\top}}^s  \leq \tau\Paren{v^{\top} \Sigma(m) v}^s \leq \Paren{\frac{1}{\tau^2}}^s \Paren{v^{\top} \Sigma(r) v}^s}\mper
\]

Dividing throughout by $\Paren{v^{\top} \Sigma(r) v}^s$ yields:
\[
\cV \sststile{4s(\tau/C)}{} \Set{\tau^2 (C\tau)^{4s} \Delta^{s}\leq 1}\mper
\]

Using that $\Delta^s \geq 2 \frac{1}{\tau^{6}}$ and subtracting out $1$ from both sides above yields:
\[
\cV \sststile{4s(\tau/C)}{}  \Set{-1 \geq 0}\mper
\]

\end{proof}

The proof of Lemma~\ref{lem:mean-separated-refutation} follows via a similar argument as above. 
We now proceed to the proof of Lemma~\ref{lem:frob-separated-refutation}. 

\begin{proof}[Proof of Lemma~\ref{lem:frob-separated-refutation}]

As in the proof of Lemma~\ref{lem:Large-Intersection-Implies-High-Variance-frobenius}, for the sake of the analysis, we first apply the linear transformation $y_i \rightarrow \Sigma(r')^{-1/2}y_i$. Let $Q = \Sigma(r)-I$.

From an argument similar to Lemma~\ref{lem:Large-Intersection-Implies-High-Variance-frobenius}, we can obtain:
\begin{equation}\label{eq:lower-bound-frob-verification}
\begin{split}
\cV \sststile{}{} \Biggl\{ 2\E_{X'} \paren{Q-\E_{X'}Q}^2 + 2  \E_{C_r} (Q-\E_{C_r}Q)^2 & + 2 \E_{C_{r'}} (Q-\E_{C_{r'}}Q)^2 \\
& \geq m'(C_{r})^2 m'(C_{r'})^2 \Norm{\Sigma(r')^{-1/2}\Sigma(r)\Sigma(r')^{-1/2}-I}_F^4 \Biggr\}
\end{split}
\end{equation}

Reasoning as in Section~\ref{sec:intersection-bounds-frobenius}, and using Lemma~\ref{lem:mean-variance-general-polynomials}:
$\E_{C_r} (Q-\E_{C_r}Q)^2 \leq (C-1) \Norm{\Sigma(r')^{-1/2}\Sigma(r)^{1/2}Q\Sigma(r)^{1/2}\Sigma(r')^{-1/2}}_F^2 \leq \Norm{\Sigma(r')^{-1/2}\Sigma(r)^{1/2}}^2_{op} \Norm{Q}_F^2$. Similarly, $\E_{C_{r'}} (Q-\E_{C_r}Q)^2 \leq \Norm{Q}_F^2$.

For the upper bound on $\E_{X'} \paren{Q-\E_{X'}Q}^2$, our proof is similar to that of Lemma~\ref{lem:variance-upper-bound-hypercontractivity} but leverages the argument in the proof of Lemma~\ref{lem:spectral-separated-refutation} to obtain a degree bound independent of $\kappa$ (without relying on the uniform polynomial approximator for the threshold):

From our bounded-variance constraints, we have:
\begin{equation}
\cA \sststile{4}{\Pi,m} \Set{ \E_{X'} (Q-\E_{X'} Q)^2  \leq C \Norm{\Pi Q\Pi}_F^2}\mper \label{eq:variance-polynomials}
\end{equation}

We will now apply Lemma~\ref{lem:contraction-property} in order to bound the RHS above. Towards that, reasoning as in Lemma~\ref{lem:spectral-separated-refutation}, we have:
\[
\cA \sststile{4s(\tau/C)}{} \Biggl\{ \Paren{v^{\top} \Sigma(X')v}^s \leq \frac{1}{\tau^{2s+2}} \Paren{v^{\top} \Sigma(r) v}^s \Biggr\}\mper
\]

Substituting $v \rightarrow \Sigma(r')^{\dagger/2}v$ yields:
\[
\cA \sststile{4s(\tau/C)}{} \Biggl\{ \Paren{v^{\top} \Sigma(r')^{\dagger/2}\Sigma(X')\Sigma(r')^{\dagger/2}v}^s \leq \frac{1}{\tau^{2s+2}} \Paren{v^{\top}\Sigma(r')^{\dagger/2} \Sigma(r) \Sigma(r')^{\dagger/2} v}^s \Biggr\}\mper
\]

Proceeding as in the proof of Lemma~\ref{lem:variance-upper-bound-hypercontractivity}, we can now obtain:
\begin{equation}\label{eq:upper-bound-frobenius-verification}
\cA \sststile{4s(\tau/C)}{} \Set{ \E_{X'} (Q-\E_{X'} Q)^{2s(\tau)}   \leq \frac{1}{\tau^{2s+2}} \Norm{\Sigma(r')^{-1/2}\Sigma(r)\Sigma(r')^{-1/2}-I}_F^2}\mper
\end{equation}

Combining \eqref{eq:lower-bound-frob-verification} and \eqref{eq:upper-bound-frobenius-verification} and the SoS almost triangle inequality (Fact~\ref{fact:sos-almost-triangle}) we obtain:

\begin{equation*}
\begin{split}
\cV \sststile{4s(\tau/C)}{}\Biggl\{ m'(C_{r})^{2s(\tau/C)} m'(C_{r'})^{2s(\tau/C)} & \Norm{\Sigma(r')^{-1/2}\Sigma(r)\Sigma(r')^{-1/2}-I}_F^{4s(\tau/C)} \\
& \leq 2^{3s} \frac{1}{\tau^{2s+2}}\Norm{\Sigma(r')^{-1/2}\Sigma(r)\Sigma(r')^{-1/2}-I}_F^2 \Biggr\}
\end{split}
\end{equation*}
Using the separation condition with the fact that $\Delta \geq \frac{3}{\tau^6}$ yields via an argument similar to that in the proof of Lemma~\ref{lem:spectral-separated-refutation}:

\[
\cV \sststile{4s(\tau/C)}{}  \Set{ -1 \geq 0}\mper
\]

\end{proof}


\section{Outlier-Robust Covariance Estimation in Frobenius Distance} \label{sec:param_recovery}

In this section, we give an outlier-robust algorithm for estimating covariances in relative Frobenius distance (i.e. Frobenius distance after putting one of the distribution in isotropic position). Our algorithm is same as the one employed in~\cite{DBLP:journals/corr/abs-1711-11581} to obtain outlier-robust algorithms for estimating mean and covariance in spectral distance for all certifiably Sub-gaussian distributions. 

Our stronger error bounds hold for distributions with certifiable hypercontractive degree 2 polynomials. This is a strictly stronger assumption (and thus a smaller class of distributions) than certifiable subgausianity considered in ~\cite{DBLP:journals/corr/abs-1711-11581}. As pointed out in ~\cite{DBLP:journals/corr/abs-1711-11581} (see discussion in the last paragraph of page 6 for a simple counter-example), certifiable Sub-gaussianity is provably insufficient to obtain the stronger relative Frobenius errors guarantees. 

Our proof approach is similar to that of~\cite{DBLP:journals/corr/abs-1711-11581}  - the key difference being that we rely on certifiable hypercontractivity (instead of the weaker certifiable Sub-gaussianity) and use the contraction lemma (Lemma~\ref{lem:contraction-property}).  


\begin{theorem}[Robust  Parameter Estimation for Certifiably Hypercontractive and Bounded-Variance Distributions] \label{thm:param-estimation-main}
Fix an $\epsilon > 0$ small enough fixed constant so that $Ct\epsilon^{1-4/t} \ll 1$\footnote{This notation means that we needed $Ct\epsilon^{1-2/t}$ to be at most $c_0$ for some absolute constant $c_0 > 0$}. For every even $t \in \N$, there's an algorithm that takes input $Y$ be an $\epsilon$-corruption of a sample $X$ of size $n \geq n_0 = d^{O(t)}/\epsilon^2$ from a $2t$-certifiably $C$-hypercontractive distribution with certifiably $C$-bounded variance with unknown mean $\mu_*$ and covariance $\Sigma_*$ respectively and in time $n^{O(t)}$ outputs an estimate $\hat{\mu}$ and $\hat{\Sigma}$ satisfying:
\begin{enumerate}
\item $\Norm{\Sigma^{-1/2}(\mu_*-\hat{\mu})}_2 \leq O(Ct)^{1/2} \epsilon^{1-1/t}$, 
\item $(1-\eta) \Sigma_* \preceq \hat{\Sigma} \preceq (1+\eta) \Sigma_*$ for $\eta \leq O(Ck) \epsilon^{1-2/t}$, and, 
\item $\Norm{\Sigma_*^{-1/2} \hat{\Sigma} \Sigma_*^{-1/2}-I}_F \leq (Ct) O(\epsilon^{1-1/t})$.
\end{enumerate}
In particular, by choosing $t = O(\log(1/\epsilon))$ results in the error bounds of $\tilde{O}(\epsilon)$ in all the three inequalities above. 
\end{theorem}

We consider the following system $\cA\seteq\cA_{Y,\e}$ of quadratic equations in scalar-valued variables $w_1,\ldots,w_n$ and vector-valued variables $x'_1,\ldots,x'_n$,
\begin{equation}
  \cA_{Y,\e}\colon
  \left \{
    \begin{aligned}
      &&
      \textstyle\sum_{i=1}^n w_i
      &= (1-\e) \cdot n\\
      &\forall i\in [n].
      & w_i^2
      & =w_i \\
      &
      &\Pi 
      &=UU^{\top}\\
      &
      &\Pi^2
      &=\Sigma\\
      &\forall i\in [n].
      & w_i \cdot (y_i - x'_i)
      & = 0\\
      &
      &\frac{1}{n} \sum_{i \leq n}  x'_i
      &=\mu\\
      &
      &\frac{1}{n}\sum_{i \leq n} (x'_i -\mu)(x'_i-\mu)^{\top}
      &=\Sigma\\
      &
      &\Paren{\frac{1}{n} \sum_{i \leq n} \left((x'_i-\mu)^{\top}Q(x'_i-\mu)\right)^{2t}} 
      &\leq (Ct)^{2t} \Paren{\frac{1}{n} \sum_{i \leq n} ((x'_i-\mu)^{\top}Q(x'_i-\mu))^2}^{t}\\
      &
      &\Paren{\frac{1}{n} \sum_{i \leq n} \left((x'_i-\mu)^{\top}Q(x'_i-\mu)\right)^{2}} 
      &\leq C \Norm{\Pi Q \Pi}_F^2\mper
    \end{aligned}
  \right \}
\end{equation}

\begin{mdframed}
  \begin{algorithm}[Parameter Estimation Algorithm]
    \label[algorithm]{alg:moment-estimation-program}\mbox{}
    \begin{description}
    \item[Given:]
      $\e$-corrupted sample $Y=\set{y_1,\ldots,y_n}\subseteq \R^d$ of a certifiably hypercontractive distribution $D_0$ over $\R^d$
    \item[Estimate:]
      Mean $\mu_*$ and Covariance $\Sigma_*$ of $D_0$.
    \item[Operation:]\mbox{}
      \begin{enumerate}
      \item 
        Find a level-$O(t)$ pseudo-distribution $\tzeta$ that satisfies $\cA_{Y,\e}$.
      \item
        Output estimates $\hat{\mu} = \pE[\mu]$ and $\hat{\Sigma} = \pE[\Sigma]$.
      \end{enumerate}
    \end{description}    
  \end{algorithm}
\end{mdframed}

\paragraph{Analysis of Algorithm}
Corollaries 4.6 and 4.7 in~\cite{DBLP:journals/corr/abs-1711-11581} show the following low-degree sum-of-squares proofs of certifiability of mean and covariance under spectral distance. 

\begin{equation} \label{eq:covariance-estimation}
\cA_{Y,\e} \sststile{O(t)}{\Sigma, u} \Set{ (1-\eta)u^{\top}\Sigma_* u \leq  \Iprod{u,\Sigma u} \leq (1+\eta) u^{\top}\Sigma_* u} \mcom
\end{equation}
for some $\eta \leq O(Ct) \epsilon^{1-2/t}$, and, 
\begin{equation} \label{eq:mean-estimation}
\cA_{Y,\e} \sststile{O(t)}{\mu, u} \Set{\Iprod{u,\mu-\mu_*} \leq \eta \iprod{u,\Sigma_*u}^{1/2}}\mcom
\end{equation}
for some $\eta = O(\sqrt{Ct} \epsilon^{1-1/t})$. 

We will rely on these to show:
\begin{lemma}[Certifiability in Relative Frobenius Distance]
\label{lem:parameter_prox}
For any $t \in \mathbb{N}$, 
\begin{equation} \label{eq:covariance-frob-estimation}
\cA_{Y,\epsilon} \sststile{4t}{\Sigma} \Biggl\{ \Norm{\Sigma_*^{-1/2}\Sigma \Sigma_*^{-1/2}}_F^{2} \leq \eta \Biggr\}
\end{equation}
where $\eta = \Paren{(Ct)^{2} O(\epsilon^{4-4/t}) + (Ct)^{2} O(\epsilon^{2-4/t})}$. 
\end{lemma}



We now conclude with proving the parameter proximity lemma: 

\begin{proof}[Proof of Lemma \ref{lem:parameter_prox}]

To show \eqref{eq:covariance-frob-estimation}, we begin
by applying the linear transformation $y \rightarrow \Sigma_*^{-1/2}y$ so as to simplify notation.

In the following, we use that $\sum_{i\leq n} (x_i' - \mu )^{\top} Q (\mu-\mu_*) = \sum_{i\leq n} (\mu-\mu_* )^{\top} Q (x_i'-\mu)  = 0$ and apply the the SoS Cauchy-Schwarz inequality (Fact~\ref{fact:sos-holder}) and guarantee for the mean estimation above (note that we are in the space  where $\Sigma_* = I$ after the affine  transform), to obtain:

\begin{equation}
\label{eqn:bound_mean}
\begin{split}
\cA_{Y,\epsilon} \sststile{4t}{Q, \mu} \Biggl\{ \frac{1}{n} \sum_{i\leq n} \Paren{(\mu-\mu_*)^{\top} Q(\mu-\mu_*)}^{2t}&\leq \Norm{\mu-\mu_*}_2^{4t} \Norm{Q}_F^{2t}\\
&\leq {(Ct)}^{2t} O(\epsilon^{4t-4}) \|Q\|^{2t}_F \Biggr\}\mper
\end{split}
\end{equation}
where the last inequality follows from the mean closeness bound in \eqref{eq:mean-estimation}. 
Using that $\Sigma$ is the covariance of $X'$ while $I$, the covariance of $X$ along with the SoS almost triangle inequality and the bound in \eqref{eq:mean-estimation}, we have:
\begin{equation}
\begin{split}
\cA \sststile{4t}{\mu,w, Q} \Biggl\{  & \Iprod{\Sigma-I,Q}^{2t} \\
&= \Paren{\frac{1}{n} \sum_{i\leq n} \Paren{Q(x'_i-\mu) - Q(x_i-\mu_*)}}^{2t}\\
&\leq 2^{2t} \Paren{\frac{1}{n} \sum_{i\leq n} \Paren{Q(x'_i-\mu_*) - Q(x_i-\mu_*)}}^{2t}  + 2^{2t}  \Paren{Q(\mu-\mu_*)}^{2t}\\
&\leq 2^{2t} \Paren{\frac{1}{n} \sum_{i\leq n} \Paren{Q(x'_i-\mu_*) - Q(x_i-\mu_*)}}^{2t}  + 2^{2t} (Ct)^{2t} O(\epsilon^{4t-4}) \Norm{Q}_F^{2t}\\
&= 2^{4t} \Paren{\frac{1}{n} \sum_{i\leq n} (1-w_i) \Paren{Q(x'_i-\mu_*)}}^{2t} + 2^{4t} \Paren{\frac{1}{n} \sum_{i\leq n} (1-w_i) Q(x_i-\mu_*)}^{2t} \\
& + 2^{2t} (Ct)^{2t} O(\epsilon^{4t-4}) \Norm{Q}_F^{2t} \Biggr\}\mper \\
\end{split}
\label{eq:main-calc-parama-estimation}
\end{equation}

Applying SoS Hölder's inequality to the first term above, using that $\cA_{Y,\e} \sststile{}{} \Set{(1-w_i)^2 = (1-w_i)}$, along with the certifiable hypercontractivity, we obtain
\begin{align*}
\cA_{Y,\epsilon} \sststile{4t}{\mu, w} \Biggl\{ \Paren{\frac{1}{n} \sum_{i\leq n} (1-w_i) \Paren{Q(x'_i-\mu_*)}}^{2t} &\leq \Paren{\frac{1}{n} \sum_{i\leq n} (1-w_i)^{2t}}^{2t-1} \Paren{\frac{1}{n} \sum_{i\leq n} \Paren{Q(x'_i-\mu_*)}^{2t}}\\
&\leq \epsilon^{2t-1} (Ct)^{2t} \Paren{\frac{1}{n} \sum_{i\leq n} \Paren{Q(x'_i-\mu_*)}^2}^{t}\\
&\leq \epsilon^{2t-1} (Ct)^{2t} \Norm{\Sigma_*^{-1/2} \Pi Q \Pi \Sigma_*^{-1/2}}^{2t}_F\\
&\leq \epsilon^{2t-1} (Ct)^{2t} t^t (Ct)^{2t} \epsilon^{2t-4} \Norm{Q}_F^{2t} \Biggr\}\mcom
\end{align*}
where in the third inequality, we invoked bounded variance constraint and in the 4th inequality we invoked Lemma~\ref{lem:contraction-property} along with \eqref{eq:covariance-estimation}. 

Similarly, we can bound $\Paren{\frac{1}{n} \sum_{i\leq n} (1-w_i) (x_i-\mu_*)^{\top} Q(x_i-\mu_*)}^{2t}$ using certifiable hypercontractivity of $X$ (the samples from the true distribution) as follows: 
\begin{align*}
\cA_{Y,\epsilon} \sststile{4t}{\mu, w, Q} \Biggl\{ \Paren{\frac{1}{n} \sum_{i\leq n} (1-w_i) \Paren{Q(x_i-\mu_*)}}^{2t} &\leq \Paren{\frac{1}{n} \sum_{i\leq n} (1-w_i)^{2t}}^{2t-1} \Paren{\frac{1}{n} \sum_{i\leq n} \Paren{Q(x_i-\mu_*)}^{2t}}\\
&\leq \epsilon^{2t-1} (Ct)^{2t} \Paren{\frac{1}{n} \sum_{i\leq n} \Paren{Q(x_i-\mu_*)}^2}^{t}\\
&\leq \epsilon^{2t-1} (Ct)^{2t} \Norm{ Q  }^{2t}_F\Biggr\}\mcom
\end{align*}
Plugging into \eqref{eq:main-calc-parama-estimation} and applying \eqref{eq:covariance-estimation} , we get 
\begin{align*}
\cA_{Y,\epsilon} \sststile{4t}{\Sigma, Q} \Biggl\{ \Iprod{\Sigma-I,Q}^{2t} &\leq  \Paren{(Ct)^{2t} O(\epsilon^{4t-4}) + (Ct)^{t} O(\epsilon^{2t-1})} \Norm{Q}_F^{2t} \Biggr\}\mcom
\end{align*}
Substituting $Q = \Sigma^{-1/2}\Sigma \Sigma^{-1/2}-I$ and using \eqref{eq:covariance-estimation} again, 
\begin{align*}
\cA_{Y,\epsilon} \sststile{4t}{\Sigma} \Biggl\{ \Norm{\Sigma^{-1/2}\Sigma \Sigma^{-1/2}-I}_F^{4t} &\leq  \Paren{(Ct)^{2t} O(\epsilon^{4t-4}) + (Ct)^{t} O(\epsilon^{2t-1})} \Norm{\Sigma^{-1/2}\Sigma \Sigma^{-1/2}-I}_F^{2t} \Biggr\}\mcom
\end{align*}

Applying Lemma~\ref{lem:cancellation-sos} with $a = \Norm{\Sigma^{-1/2}\Sigma \Sigma^{-1/2}-I}_F^{2t}$ yields the lemma.

\end{proof}

It's easy to finish the proof of Theorem~\ref{thm:param-estimation-main} from here.

\begin{proof}[Proof of Theorem~\ref{thm:param-estimation-main}.]
Then, by an argument similar to proof of Theorem 1.2 in~\cite{DBLP:journals/corr/abs-1711-11581}, $\pE[\Sigma]$ satisfies the third guarantee in Theorem~\ref{thm:param-estimation-main}. Let $\tzeta$  be the degree-$O(\ell)$ pseudo-distribution output by our algorithm above. Then, our estimator for the covariance is simply $ \hat{\Sigma} = \expecf{\tzeta}{\Sigma}$. From Lemma \ref{lem:parameter_prox} it follows that 
\[
\cA_{Y,\epsilon} \sststile{4t}{\Sigma, Q} \Biggl\{ \Iprod{\Sigma-I,Q}^{2t} \leq  \eta \Norm{Q}_F^{2t} \Biggr\}
\]
where $\eta =  ((Ct) \epsilon^{1-4/t})$.
Therefore, for any $Q$, we have, $\expecf{\tzeta}{\Iprod{\Sigma-I,Q}^{2t}} \leq \eta \Norm{Q}^{2t}_F$. Then, using Cauchy-Schwarz for pseudo-distributions we have 
\begin{equation}
\label{eqn:rounding_final}
\begin{split}
\Paren{\Iprod{\Sigma_*^{-1/2} \expecf{\tzeta}{\Sigma} \Sigma_*^{-1/2}-I,Q}}^{2} & = \Paren{\expecf{\tzeta}{\Iprod{\Sigma_*^{-1/2} \Sigma \Sigma_*^{-1/2}-I,Q}}}^2 \\
&\leq \expecf{\tzeta}{\Iprod{\Sigma_*^{-1/2} \Sigma \Sigma_*^{-1/2}-I,Q}^{2}}  \\
& \leq \eta \Norm{Q}^{2}_F
\end{split}
\end{equation}
Setting $Q = \Sigma_*^{-1/2} \Sigma \Sigma_*^{-1/2}-I$, yields the claim. 
\end{proof}


\section{Reasonable Distributions} \label{sec:reasonable-distributions}
In this section, we recall known results that imply that Gaussian distributions and affine transforms of uniform distribution on the unit sphere are reasonable. 
\paragraph{Certifiable Hypercontractivity of Degree 2 Polynomials}

\begin{definition}[Certifiable Hypercontractivity]
Let $\cD$ be a distribution on $\R^d$. For an even $h$, $\cD$ is said to have $h$-certifiably $C$-hypercontractive degree $2$ polynomials if for $P$ - a $d \times d$ matrix-valued indeterminate, 
\[
\E_{x \sim \cD} \iprod{P,x^{\otimes 2}}^h \leq (Ch)^{h} (\E x^{\top}Px^2)^{h/2}\mper
\]
\end{definition}

Gaussian distributions satisfy $h$-certifiable $1$-hypercontractive. 

We will need the following result that follows from~\cite{MR3376479-Kauers14}:
\begin{fact}[Hypercontractivity of Degree-$2$ Polynomials of Gaussians]
\label{fact:certifiable_hypercontractivity_gaussians}
The standard normal distribution, $\cN(0,I)$, is $h$-certifiable $1$-hypercontractive. 
\end{fact}


Since this is a fact about degree 2 polynomials, as stated, non-standard Gaussian distributions do not have certifiably hypercontractive degree 2 polynomials.

\begin{lemma}[Certifiable Hypercontractivity Under Sampling]
\label{lem:cert_hyper_sampling}
Let $\cD$ be a $1$-sub-gaussian, $h$-certifiably $c$-hypercontractive distribution over $\mathbb{R}^d$. Let $\cS$ be a set of $n = \Omega( (hd)^{8h})$ i.i.d. samples from $\cD$. Then, with probability at least $1-1/\poly(n)$, the uniform distribution on $\cS$ is $h$-certifiably $(2c)$-hypercontractive.
\end{lemma}
\begin{proof}
Since $\cD$ is $h$-certifiably $c$-hypercontractive, 
\begin{equation*}
\sststile{2h}{P} \Set{ \E_{x \sim \cD} \left[  \left\langle P, x^{\otimes 2} \right\rangle^{h} \right] \leq (ch)^h \|P\|^h_F }
\end{equation*}
Since for any matrices $M$ and $N$, $\langle M, N \rangle^{h} = \langle M^{\otimes h} , N^{\otimes h} \rangle$ using the substitution rule, 
\begin{equation}
\label{eqn:original_hypercontractive}
\sststile{2h}{P} \Set{   \left\langle P^{\otimes h} , \E_{x \sim \cD} \left[ x^{\otimes 2h} \right]  \right\rangle \leq (ch)^h \|P\|^h_F }
\end{equation}
Let $\cD'$ be the uniform distribution over samples from $\cD$. Then,  
\begin{equation*}
\E_{x \sim \cD'} \left[  \left\langle P, x^{\otimes 2} \right\rangle^{h} \right]  = \left\langle P^{\otimes h} , \E_{x \sim \cD'} \left[ x^{\otimes 2h} \right]  \right\rangle 
\end{equation*}
Let $M =  \E_{x \sim \cD'} \left[ x^{\otimes 2h} \right]    - \E_{x \sim \cD} \left[ x^{\otimes 2h} \right] $.
Therefore, assuming that $\|M\|_{2} \leq (ch)^h$, using Fact \ref{fact:operator_norm}  with the substitution rule, we can conclude
\begin{equation}
\label{eqn:bounding_operator_norm}
\sststile{2h}{P} \Set{ \left| \left\langle P^{\otimes h}, M \right\rangle\right| \leq (ch)^h \|P\|^{h}_F }
\end{equation}
Observe, we can then rewrite \eqref{eqn:original_hypercontractive} as follows : 
\begin{equation*}
\sststile{2h}{P} \Set{   \left\langle P^{\otimes h} , \E_{x \sim \cD'} \left[ x^{\otimes 2h} \right] -M \right\rangle \leq (ch)^h \|P\|^h_F }
\end{equation*}
Rearranging and using \ref{eqn:bounding_operator_norm}, we can conclude 
\begin{equation*}
\sststile{2h}{P} \Set{   \left\langle P^{\otimes h} , \E_{x \sim \cD'} \left[ x^{\otimes 2h} \right] \right\rangle \leq 2(ch)^h \|P\|^h_F }
\end{equation*}
Therefore, it remains to show $\|M\|_{2} \leq (ch)^h$. Let $x^{(1)},x^{(2)}, \ldots x^{(n)}$ be $n$ iid samples from $\cD$. Then, observe
\begin{equation*}
M_{i_1,\ldots, i_{2h}} = \left[\E_{x\sim\cD'} x^{\otimes 2h} \right]_{i_1, \ldots i_{2h} }-  \left[\E_{x\sim\cD} x^{\otimes 2h} \right]_{i_1, \ldots i_{2h} } = \frac{1}{n} \sum_{\ell\in[n]} \Paren{x^{(\ell)}_{i_1}x^{(\ell)}_{i_2}\ldots x^{(\ell)}_{i_{2h}}  -  \E_{x\sim \cD}  \left[ x_{i_1} x_{i_2}\ldots x_{i_{2h}} \right]}\mper
\end{equation*}

Let $Z_{\ell} = \Paren{x^{(\ell)}_{i_1}x^{(\ell)}_{i_2}\ldots x^{(\ell)}_{i_{2h}}}$. Then,  $M_{i_1,\ldots, i_{2h}}$ is an average of independent random variables $\bar{Z}_{\ell} = Z_{\ell} -\E[Z_{\ell}]$ for $\ell \in [n]$. We will estimate moments of $\sum_{\ell \leq n} \bar{Z}_\ell$ in order to order to obtain upper bounds on the deviation probabilities. 

Towards that we observe the following:
$\E \left[\Paren{\frac{1}{n} \sum_{\ell \in [n]} \bar{Z}_{\ell}}^{2t}\right] = \frac{1}{n^{2t}} \sum_{r_1, r_2,\ldots, r_{2t}} \E\left[\prod_{j \in [2t]} \bar{Z}_{r_j}\right]$. If $\E[\prod_{j \in [2t]} \bar{Z}_{r_j}] \neq 0$, then, each $\bar{Z}_{r_j}$ must appear even number of times in the product. Thus, the number of distinct $\bar{Z}_{r_j}$ in the product are at most $t$. Thus, the number of non-zero terms in the above sum is at most $n^{t} (2t)^{2t}$.
Next, for any non-zero term in the above sum, using the AM-GM inequality,
\begin{equation}
\label{eqn:prod_bound_am_gm}
\E\left[\prod_{i \in [2t]} \bar{Z}_{r_j}\right] \leq \frac{1}{(2t)^{2t}} \E\left[ \left( \sum_{i \in [2t]} \bar{Z}_{r_j} \right)^{2t} \right] \leq \frac{1}{(2t)} \sum_{i \in [2t]} \E[\bar{Z}_{r_j}^{2t}] 
\end{equation}
By Jensen's inequality, $(\E[Z_{r_i}])^{2t} \leq \E[Z_{\ell}^{2t}]$ and thus  $\E\left[\bar{Z}_{r_j} ^{2t}\right]  \leq 2^{2t} (\E[Z_{r_j}^{2t}] +  (\E[ Z_{r_j})^{2t}]) \leq 2^{2t+1} \E[Z_{r_j}^{2t}]$. Then,
\begin{equation}
\label{eqn:term_wise_bound}
\begin{split}
\E\left[Z_{r_j}^{2t}\right] = \E\left[\Paren{x^{(r_j)}_{i_1}x^{(r_j)}_{i_2}\ldots x^{(r_j)}_{i_{2h}}}^{2t}\right] & \leq \E\left[ \Paren{\frac{1}{2h} \sum_{k \in[2h]} \Paren{x^{(r_i)}_{i_k}}^{2h}}^{2t}\right] \\
& \leq \frac{1}{2h} \sum_{k\in [2h]}  \E\left[\Paren{x^{(r_i)}_{i_k}}^{4ht}\right] \\
& \leq (4ht)^{2ht}
\end{split}
\end{equation}
 where the first inequality uses the AM-GM inequality, the second uses Jensen's inequality and the final inequality uses the $1$-subgaussianity of $x^{(r_j)}_{i_j}$. 
Combining \eqref{eqn:prod_bound_am_gm} and  \eqref{eqn:term_wise_bound} 
\begin{equation*}
\E \left[\Paren{\frac{1}{n} \sum_{\ell \in [n]} \bar{Z}_{\ell}}^{2t}\right] \leq \frac{1}{2t n^{2t} } \cdot n^t (2t)^{2t} \cdot (4ht)^{2ht} \leq n^{-t} (2t)^{2t-1} (4ht)^{2ht}
\end{equation*}
Using Chebyshev's inequality, 
\begin{equation*}
\Pr\left[ \Big| \frac{1}{n} \sum_{\ell \in [n]} \bar{Z}_{\ell} \Big| > \eta \right] \leq \frac{ \E \left[\Paren{\frac{1}{n} \sum_{\ell \in [n]} \bar{Z}_{\ell}}^{2t} \right]  }{\eta^{2t}} \leq \frac{(2t)^{2t-1} (4ht)^{2ht}}{\eta^{2t} n^{t} }
\end{equation*}
Setting $t > 2h \log d$ and $\eta = (ch/d^2)^{h}$ yields that whenever $n \geq n_0 = \Omega\left( \frac{d^{4h}}{c^{2h}} h^{9h} \log^{2h+2}(d) \right)$,  $|M_{i_1,i_2,\ldots i_{2h}}| \leq \eta$ with probability at least $1-1/d^{4h}$.  By a union bound over the $d^{2h}$ entries of $M$, we have that all entries of $M$ are at most $\eta$ with probability at least $1 - d^{-2h}$. We can then easily bound the operator norm of $M$ by $d^{2h} \cdot (ch/d^2)^h = (ch)^h$, which completes the proof. 


\end{proof}


\paragraph{Certifiable Anti-Concentration}

\begin{lemma}[Certifiable Anti-Concentration of Gaussians, Theorem 5.5 \cite{bakshi2020list}]
\label{lem:cert_anti_conc_gaussian}
Given $0<\delta \leq 1/2$, there exists $s = O\left(\frac{\log^{5}(1/\delta)}{\delta^{2}} \right)$ such that the Gaussian distribution and the uniform distribution on the unit sphere is $s$-certifiably $(C, \delta)$-anti-concentrated. 
\end{lemma}


\begin{lemma}[Certifiable Anti-Concentration under Sampling, Lemma 5.8 \cite{bakshi2020list}]
\label{lem:cert_anti_conc_sampling}
Let $\cD$ be $s$-certifiably $(c, \delta)$-anti-concentrated Sub-Exponential distribution over $\mathbb{R}^d$. Let $\cS$ be a set of $n = \Omega( (sd\log(d))^{s})$ i.i.d. samples from $\cD$. Then, with probability at least $1-1/\poly(n)$, the uniform distribution on $\cS$ is $s$-certifiably $(2c,\delta)$-anti-concentrated.
\end{lemma}

\paragraph{Bounded Variance of Degree-$2$ Polynomials.}

Recall that we say that a zero mean distribution $\cD$ with covariance $\Sigma$ has certifiably $C$-bounded variance degree $2$ polynomials if $\sststile{2}{Q} \Set{\E_{x \sim \cD} (x^{\top}Qx-\E_{x \sim \cD} x^{\top}Qx)^2 \leq C \Norm{\Sigma^{1/2}Q\Sigma^{1/2}}_F^2}.$ 

\begin{lemma}[Bounded Variance of Degree 2 Polynomials of 4-wise independent distributions]
Let $\cD$ be an isotropic, 4-wise independent distribution on $\R^d$. Then, $\cD$ has certifiably $3$-bounded variance degree $2$ polynomials. That is, 
\[
\sststile{2}{Q} \Set{\E_{\cD} \Paren{x^{\top}Qx - \E_{\cD} x^{\top}Qx}^2 \leq 3 \Norm{Q}_F^2}\mper
\]
\end{lemma}
\begin{proof}
By viewing $xx^{\top}$ and $I \in \R^{d \times d}$ as $d^2$ dimensional vectors, and using that $\E_{y \sim \cD} (yy^{\top}-I)(yy^{\top}-I)^{\top} \preceq 3I \otimes I$ for any $4$-wise independent, isotropic distribution, we have:
\begin{multline}
\sststile{2}{Q} \Biggl\{ \E_{\cD} \Paren{x^{\top}Qx - \E_{\cD} x^{\top}Qx}^2 = \E_{\cD} \Iprod{xx^{\top}-I,Q}^2 \leq \Norm{\E_{x \sim \cD}(xx^{\top}-I)(xx^{\top}-I)^{\top}}_2 \Norm{Q}_F^2 \\\leq 3\Norm{I \otimes I}_2 \Norm{Q}_F^2 = 3 \Norm{Q}_F^2 \Biggr\}\mper
\end{multline}




\end{proof}

The uniform distribution on $\sqrt{d}$-radius sphere in $d$ dimensions is not $4$-wise independent. However, the above proof only requires that $\E (y^{\otimes 2}-I)(y^{\otimes 2}-I)^{\top} \preceq C I \otimes I$. For the uniform distribution on the sphere, notice that $i,j,k,\ell$-th entry of this matrix is non-zero iff the indices are in have two repeated indices and in that case, by negative correlation of the $x_i^2$ and $x_j^2$ on the sphere, it holds that $\E x_i^2 x_j^2 \leq 1$. Thus, $\E (y^{\otimes 2}-I)(y^{\otimes 2}-I)^{\top} \preceq 3 I \otimes I$ for $y$ uniformly distribution on the $\sqrt{d}$-radius unit sphere. The above proof thus also yields:

\begin{corollary}
Let $y$ be uniform on $\sqrt{d}$-radius sphere in $d$ dimensions. Then, $y$ has certifiably $3$-bounded variance degree $2$ polynomials. 
\end{corollary}

\begin{lemma}[Linear Invariance] \label{lem:linear-invariance-bounded-variance}
Let $x$ be a random variable with an isotropic distribution $\cD$ on $\R^d$ with certifiably C-bounded variance degree $2$ polynomials. Let $A \in \R^{d \times d}$ be an arbitrary $d \times d$ matrix. Then, the random variable $x'=Ax$ also has certifiably C-bounded variance degree $2$ polynomials. 
\end{lemma}

\begin{proof}
The covariance of $x'$ is $AA^{\top} =\Sigma$, say. Let $\Sigma^{1/2}$ be the PSD square root of $\Sigma$. The proof follows by noting that ${x'}^{\top}Qx' = (Ax)^{\top}Q(Ax) = x^{\top} (A^{\top}QA) x^{\top}$ and that $\Norm{A^{\top}QA}_F^2 = \tr(A^{\top}QAA^{\top}QA) = \tr(AA^{\top}QAA^{\top}Q) = \tr(\Sigma Q\Sigma Q) = \tr(\Sigma^{1/2}Q\Sigma^{1/2} \Sigma^{1/2}Q\Sigma^{1/2}) = \Norm{\Sigma^{1/2}Q\Sigma^{1/2}}_F^2$. 
\end{proof}

\begin{lemma}[Bounded Variance Under Sampling] \label{lem:bounded-variance-under-sampling}
Let $\cD$ be have degree $2$ polynomials with certifiably C-bounded variance and be $8$-certifiably $C$-subgaussian. Let $X$ be an i.i.d. sample from $\cD$ of size $n \geq n_0 = O(C^4) d^{16}$. Then, with probability at least $0.99$ over the draw of $X$, the uniform distribution on $X$ has degree $2$ polynomials with certifiable $2C$-bounded variance. 
\end{lemma}

\begin{proof}
Using Lemma~\ref{lem:linear-invariance-bounded-variance}, we can assume that $\cD$ is isotropic. Arguing as in the proof of Lemma~\ref{lem:cert_hyper_sampling}, it is enough to upper-bound the spectral norm $\Norm{\frac{1}{n}\sum_i (x_i^{\otimes 2}-I)(x_i^{\otimes 2}-I)^{\top}-\E_{x \sim \cD} (x^{\otimes 2}-I)(x^{\otimes 2}-I)^{\top}}_2$ by $C$ (with probability $0.99$ over the draw of $X$). We do this below:

By applying certifiable $C$-bounded variance property to $Q = vv^{\top}$ where $e_i$ are standard basis vectors in $\R^d$, we have that $\E (\iprod{x_i,v}^2-\E \iprod{x_i,v}^2)^2 \leq C \Norm{v}_2^4$ and thus, $\E \iprod{x_i,v}^4 \leq (1+C) \Norm{v}_2^4$. By an application of the AM-GM inequality, we know that for every $i,j,k,\ell$, $(\iprod{x,e_i}^2\iprod{x,e_j}^2 \iprod{x,e_k}^2 \iprod{x,e_\ell})^2 \leq \iprod{x,e_i}^8 + \iprod{x,e_j}^8 + \iprod{x,e_k}^{8} + \iprod{x,e_\ell}^8$. Thus, the variance of every entry of the matrix $\E x^{\otimes 4}$ is bounded above by $4(8C)^4 = O(C^4)$. Thus, by Chebyshev's inequality, any given entry of $\frac{1}{n} x_i^{\otimes 4}-\E_{x \sim \cD} x^{\otimes 4}$ is upper-bounded by $O(C^2)d^4/\sqrt{n}$ with probability at least $1-1/(100d^4)$. By a union bound, all entries of this tensor are upper-bounded by $O(C^2)d^4/\sqrt{n}$ with probability at least $0.99$. Thus, the Frobenius norm of this tensor is at most $d^8 O(C^2)/\sqrt{n}$. Since $n \geq n_0 = O(C^4)d^{16}$, this bound is at most $C/2$. Thus, we obtain that with probability at least $0.99$, $\Norm{\frac{1}{n}\sum_i (x_i^{\otimes 2}-I)(x_i^{\otimes 2}-I)^{\top}-\E_{x \sim \cD} (x^{\otimes 2}-I)(x^{\otimes 2}-I)^{\top}}_2 \leq 2\Norm{\frac{1}{n}\sum_i x_i^{\otimes 4}-\E_{x \sim \cD} x^{\otimes 4}}_F \leq C$. 

\end{proof}

The above three lemmas immediately yield that Gaussian distributions, linear transforms of uniform distribution on unit sphere, discrete product sets such as the Boolean hypercube and any 4-wise independent zero-mean distribution has certifiably $C$-bounded variance degree $2$ polynomials. 
\section{Sum-of-Squares Toolkit} \label{sec:sos-toolkit}
In this section,  we give low-degree SoS proofs of some inequalities that we use repeatedly in our arguments. 

The following is an SoS version of the following simple matrix analytic inequality: for any matrices $A,B$, $\Norm{AB}_F^2 \leq \Norm{A}_{op}^2 \Norm{B}_F^2$. We give a constant degree SoS proof of this inequality (with $O(1)$ factor loss) by relying on certifiable hypercontractivity of Gaussians. 

\begin{lemma}[Contraction and Frobenius Norms]
Let $A,B$ be $d \times d$ matrix valued indeterminates. Let $\beta$ be a scalar-valued indeterminate.
Then, 
\[
\Set{\beta \Paren{v^{\top}A^{\top}Av}^{t} \preceq \Delta \norm{v}_2^{2t}} \sststile{}{} \Set{ \beta \Norm{AB}_F^{2t} \leq \Delta t^t \Norm{B}_F^{2t}} \mcom
\]
and
\[
\Set{\beta \Paren{v^{\top}AA^{\top}v}^{t} \preceq \Delta \norm{v}_2^{2t}} \sststile{}{} \Set{ \beta\Norm{BA}_F^{2t} \leq \Delta t^t \Norm{B}_F^{2t}} \mcom
\]
\label{lem:contraction-property}
\end{lemma}
\begin{proof}
We prove the first conclusion. The proof of the second one is similar.

We start by observing that for any matrix valued indeterminate $M$, $\sststile{2}{M} \Set{ \Norm{M}_F^2 = \E_g \Norm{Mg}_2^2}$ where the expectation is with respect to $g \sim\cN(0,I)$.

We thus have: 
\begin{equation}
\begin{split}
\Set{\beta \Paren{v^{\top}A^{\top}Av}^{t} \leq \Delta \norm{v}_2^{2t}} \sststile{}{} \Biggl\{\beta \Paren{\Norm{AB}_F^2}^t = \Paren{\E_g \Norm{ABg}_2^2}^t
&\leq \beta \E_g \Norm{ABg}_2^t\\
&= \E_g \Paren{(Bg)^{\top} \Paren{\beta A^{\top} A} (Bg)}^t\\
&\leq \Delta \E_g \Norm{Bg}_2^{2t}\\
&\leq t^t \Delta \Paren{\E_g \Norm{Bg}_2^2}^t\\
&= t^t \Delta \Norm{B}_F^{2t} \Biggr\}\mper
\end{split}
\end{equation}  
Here, the first inequality follows by using the SoS Hölder's inequality, the second one uses the constraint satisfied by $A^{\top} A$ with the substituting $v  = Bg$ and finally,  the last inequality relies on certifiable hypercontractivity of quadratic forms of Gaussians.
This completes the proof.

\end{proof}

The following two lemmas allow us to ``cancel out'' common factors from both sides of an inequality in low-degree SoS. 

\begin{lemma}[Cancellation within SoS, Constant RHS] \label{lem:cancellation-SoS-constant-RHS}
Let $a$ be an indeterminate. Then, 
\[
\Set{a^{2t} \leq 1} \sststile{2t}{a} \Set{a^2 \leq 1}\mper
\]
\end{lemma}
\begin{proof}
Applying the SoS AM-GM inequality (Fact~\ref{fact:sos-am-gm}) with $f_1 = a^2$, $f_2 = \ldots = f_t = 1$, we get:
\[
\sststile{2t}{a} \Set{a^2 \leq a^{2t}/t + 1-1/t} \mper
\]
Thus, 
\[
\Set{a^{2t} \leq 1} \sststile{2t}{a} \Set{a^2 \leq 1/t + 1-1/t = 1}\mper
\]
\end{proof}

\begin{lemma}[Cancellation Within SoS]
Let $a,C$ be indeterminates. Then, 
\[
\Set{a \geq 0 } \cup \Set{a^t \leq C a^{t-1}} \sststile{2t}{a,C} \Set{a^{2t} \leq C^{2t}}\mper
\]
\label{lem:cancellation-sos}
\end{lemma}
\begin{proof}
We first prove the case of $t = 2$. We have:

\[
\sststile{2}{a,C} \Set{a^2 = (a-C/2 + C/2)^2 \leq 2 (a-C/2)^2 + 2 (C/2)^2 }\mper
\]
And,
\[
\Set{a^2 \leq C a} \sststile{2}{a,C}  \Set{(a-C/2)^2 \leq C^2/4}\mper
\]
Thus,
\[
\Set{a^2 \leq C a} \sststile{2}{a,C} \Set{a^2 \leq C^2}\mper
\]

Consider now the general case. Iteratively using $\Set{a^t \leq C a^{t-1}}$ yields:
\[
\Set{a \geq 0 } \cup \Set{a^t \leq C a^{t-1}} \sststile{2t}{a,C} \Set{a^{2t} \leq a^{t-2} a^t  C^2 \leq a^{t-3} a^t  C^3 \ldots \leq a^t C^t}\mper
\]
Applying the special case of $t = 2$ above to the indeterminate $a^t$ now yields:
\[
\Set{a \geq 0} \Set{a^t \leq C a^{t-1}} \sststile{2t}{a,C} \Set{a^{2t} \leq C^{2t}}\mper
\]
\end{proof}

\section*{Acknowledgment}
We thank Boaz Barak, Ryan O'Donnell, Venkat Guruswami, Rajesh Jayaram, Gautam Kamath, Roie Levin, Jerry Li, Pedro Paredes, Nicolas Resch and David Woodruff for illuminating discussions related to this project. We thank Sam Hopkins for suggesting that the techniques from~\cite{bakshi2020list} might be relevant for Outlier-Robust Clustering and Misha Ivkov and Peter Manohar for pointing out typos in a previous version of this paper. 

We thank an anonymous reviewer for pointing out an issue with the proof of a lemma in a previous version of this paper that bounds the variance of degree $2$ polynomials (such lemmas now appear in Section~\ref{sec:reasonable-distributions}).



\phantomsection
  \addcontentsline{toc}{section}{References}
  \bibliographystyle{amsalpha}
  \bibliography{bib/mathreview,bib/dblp,bib/custom,bib/scholar,bib/custom2}  
\appendix

\section{Total Variation vs Parameter Distance for Gaussian Distributions} \label{Sec:tv-vs-param-gaussian}
\begin{proposition}[Parameter Closeness Implies TV Closeness for Gaussian Base Model] \label{prop:tv-vs-param-for-gaussians}
Fix $\Delta > 0$ and let $\mu,\mu'$ and $\Sigma, \Sigma' \succ 0$ satisfy:
\begin{enumerate}
\item \textbf{Mean Closeness: } for all $v \in \R^d$, $\Norm{\Paren{\mu - \mu'},v}^2_2 \leq \Delta^2 v^{\top} (\Sigma + \Sigma')v$.
\item \textbf{Spectral Closeness: } for all $v \in \R^d$ $\frac{1}{\Delta^2} v^{\top} \Sigma v \leq v^{\top} \Sigma' v \leq \Delta^ 2 v^{\top} \Sigma(r')v$. 
\item \textbf{Relative Frobenius Closeness: } $\Norm{ \Sigma^{\dagger/2} \Sigma' \Sigma^{\dagger/2} -I}_F^2 \leq \Delta^2 \cdot \Norm{ \Sigma^{\dagger} \Sigma'}^2_2$.
\end{enumerate}
Then, $\dtv(\cN(\mu,\Sigma), \cN(\mu',\Sigma')) \leq 1- \exp(-O(\Delta^2 \log \Delta))$.
\end{proposition}

\begin{proof}[Proof of Lemma~\ref{prop:tv-vs-param-for-gaussians}] 

We will work with the distributions after applying the transformation $x \rightarrow \Sigma^{-1/2}x$ to the associated random variables. Since $\dtv$ is invariant under affine transformations, this is WLOG. The transformation produces distributions $\cN(\mu_1, I)$ and $\cN(\Sigma^{-1/2} \mu', \Sigma^{-1/2} \Sigma' \Sigma^{-1/2})$ for $\mu_1 = \Sigma^{-1/2} \mu$, $\mu_2 =\Sigma^{-1/2} \mu' $ and $\Sigma_2 = \Sigma^{-1/2} \Sigma' \Sigma^{-1/2}$.

We will first bound the Hellinger distance between the two distributions above. Recall that $h= h(\cN(\Sigma^{-1/2} \mu, I),\cN(\Sigma^{-1/2} \mu', \Sigma^{-1/2} \Sigma' \Sigma^{-1/2}))$ satisfies:
\[
h(\cN(\mu_1,I), \cN(\mu_2,\Sigma_2))^2 =1 - \frac{\det(\Sigma_2)^{1/4}}{\det\left( \frac{I+\Sigma_2}{2}\right)^{\frac{1}{2}}} \exp\left(-\frac{1}{8} (\mu_1 - \mu_2)^{\top}\Paren{\frac{I + \Sigma_2}{2}}^{-1} (\mu_1 - \mu_2) \right)\mper
\]
We will estimate the RHS of the expression above to bound the Hellinger distance.

From the mean closeness condition, we have:

\[
\iprod{\mu_1-\mu_2,v}= \langle \mu - \mu', \Sigma^{-1/2} v \rangle \leq \sqrt{\log 1/\eta} \sqrt{ v^{\top} (I + \Sigma_2) v } \mper
\]

Plugging in $v = \Paren{\frac{I + \Sigma_2}{2}}^{-1} (\mu_1 - \mu_2)$ gives:
\[
\iprod{\mu_1-\mu_2,\frac{I + \Sigma_2}{2}^{-1} (\mu_1 - \mu_2)} \leq  2/\eta \sqrt{ v^{\top} \Paren{\frac{I + \Sigma_2}{2}}^{-1} v } \mcom
\]
or, \[ \iprod{\mu_1-\mu_2,\Paren{\frac{I + \Sigma_2}{2}}^{-1} (\mu_1 - \mu_2)} \leq 4 1/\eta^2  \mper\]

And thus,
\[
\exp\left(-\frac{1}{8} (\mu_1 - \mu_2)^{\top} \Paren{ \frac{I + \Sigma_2}{2}}^{-1} (\mu_1 - \mu_2) \right) \geq \exp{\Paren{-1/2\eta^2}}\mper
\]

Thus, we have:
\[
h \leq 1- \frac{\det(\Sigma_2)^{\frac{1}{4}}}{\det \Paren{ \frac{\II+ \Sigma_2}{2} }^{1/2}} \exp{\Paren{-1/2\eta^2}}\mper
\]

Let $\lambda_1 \geq \lambda_2 \geq \cdots \lambda_d$ be eigenvalues of $\Sigma_2$.
From the spectral closeness condition, observe that each $\frac{1}{\eta} \geq \lambda_1 \geq \cdots \lambda_d \geq \eta$.

Then,
\[
\frac{\det(\Sigma_2)^{\frac{1}{4}}}{\det \Paren{ \frac{\II+ \Sigma_2}{2} }^{1/2}} = \frac{\Pi_{i \leq d} \lambda_i^{1/4}} { \Pi_{i \leq d} \Paren{\frac{1+\lambda_i}{2}}^{1/2}}\mper
\]
Thus,
\begin{equation}
\log(1/(1-h)) \leq \frac{1}{2} \log (1/\eta) + \frac{1}{2} \sum_{i\in [d]}\log\left(\frac{1+\lambda_i}{2\sqrt{\lambda_i}}\right)\mper \label{eq:hellinger-bound}
\end{equation}

We break the second term in the RHS above based on the magnitude of the eigenvalues $\lambda_i$s.
Let's first bound the contribution to this term coming from eigenvalues $\lambda_i \geq 1.5$ - let's call these the \emph{large} eigenvalues of $\Sigma_2$.

Next, observe that the Relative Frobenius Closeness condition gives us that $\|\II - \Sigma_2 \|^2_F \leq (1/\eta^2)$. Thus, $\sum_{i \in [d]}(1- \lambda_i)^2 = \|\II - \Sigma_2 \|^2_F \leq (1/\eta^2)$, the number of large eigenvalues is at most $4/\eta^2$. Further, for every large eigenvalue $\lambda_i$, $1+\lambda_i \leq 2\lambda_i$. Thus,

\begin{equation*}
\sum_{i: \lambda_i \text{ is large }} \log\left(\frac{1+\lambda_i}{2\sqrt{\lambda_i}}\right) \leq \sum_{i\in \cE} \log\left(\sqrt{\lambda_i} \right) \leq \frac{2}{\eta}\cdot \log(1/\eta)
\end{equation*}
where the last step uses that $\lambda_i \leq 1/\eta$.

Let's now consider all the remaining \emph{small} eigenvalues that satisfy $\eta \leq\lambda_i < 1.5$. Then, we can write $\lambda_i = 1+\beta_i$ such that $-(1-\eta)\leq \beta_i \leq 0.5$. Then, we have
\begin{equation*}
\begin{split}
\sum_{i:\lambda_i \leq 1.5} \log\left(\frac{1+\lambda_i}{2}\right) + \frac{1}{2}\log\left(\frac{1}{\lambda_i}\right) &= \sum_{i \in \cE'} \log\left(1+\frac{\beta_i}{2}\right) - \frac{1}{2}\log\left(1+\beta_i\right) \\
&\leq \sum_{i:\lambda_i \leq 1.5} \frac{\beta_i}{2} - \frac{\beta_i}{2} + \frac{\beta_i^2}{4}\\
& = \sum_{i:\lambda_i \leq 1.5} \frac{(1-\lambda_i)^2}{4} \leq \frac{1}{4\eta^2}
\end{split}
\end{equation*}
using the bound $\sum_{i} (1-\lambda_i)^2 \leq \frac{1}{\eta^2}$ in the last inequality.
Plugging this estimate back in \eqref{eq:hellinger-bound} yields $h \geq 1- \exp (-O(1/\eta^2 \log (1/\eta))$.

To finish the proof, we observe that $\dtv(p,q) \leq h(p,q) \sqrt{2-h(p,q)} \leq  1-\exp (-O(1/\eta^2 \log (1/\eta))$.

\end{proof}

\section{Typical Samples are Good with High Probability}
\label{sec:feasibility}

\begin{proof}[Proof of Lemma~\ref{lem:typical-samples-good}] 
We begin with the empirical mean condition. 
For any fixed $\ell$, $C_{\ell}$ contains samples from a $1$-Sub-gaussian distributions and thus it follows from Fact \ref{fact:empirical_mean} that with probability at least $1-(1/\delta)$,
\begin{equation*}
\Iprod{\mu_{\ell} - \hat{\mu}_{\ell} , \Sigma^{\dagger/2}_{\ell} v}^2 = v^{\top} \Sigma^{\dagger/2}_{\ell}(\mu_{\ell} - \hat{\mu}_{\ell})(\mu_{\ell} - \hat{\mu}_{\ell})^{\top}\Sigma^{\dagger/2}_{\ell} v \leq  \left( \frac{ k r + \log(1/\delta) k}{n} \right)  v^T v
\end{equation*} 
Since $n_0= \Omega( (k\log(rk) + kr))$, we can substitute $v \to \Sigma^{1/2}_{\ell}v$ to get 
\begin{equation*}
\Iprod{\mu_{\ell} - \hat{\mu}_{\ell} , \Sigma^{\dagger/2}_{\ell} \Sigma^{1/2}_{\ell} v}^2  \leq  1.01  v^T\Sigma_{\ell} v
\end{equation*} 
Observe, $\Iprod{\mu_{\ell} - \hat{\mu}_{\ell} , \Sigma^{\dagger/2}_{\ell} \Sigma^{1/2}_{\ell} v} = \Iprod{\Sigma^{\dagger/2}_{\ell} \Sigma^{1/2}_{\ell} (\mu_{\ell} - \hat{\mu}_{\ell}),  v} = \Iprod{\mu_{\ell} - \hat{\mu}_{\ell},  v}$, where the last equality follows from observing that $\mu_{\ell} - \hat{\mu}_{\ell}$ lies in the subspace spanned by $\Sigma_{\ell}$. 
Union bound over failure events for all  $\ell \in [k]$ and thus with probability at least $1-1/\poly(k)$, for all $\ell \in [k]$, $\Iprod{\mu_{\ell} - \hat{\mu}_{\ell} ,  v}^2  \leq  1.01  v^T\Sigma_{\ell} v$.

Similarly,
using Fact \ref{fact:empirical_cov} for i.i.d. samples from a $1$-Sub-gaussian distribution, it follows that for a fixed $\ell \in [k]$, with probability at least $1-1/d^{10}$,
\[
\left(1- c\sqrt{\frac{rk\log(k)}{n}}\right)\Sigma_{\ell} \preceq \hat{\Sigma}_{\ell} \preceq \left(1+c\sqrt{\frac{rk\log(k)}{n}}\right) \Sigma_{\ell}\] 
for fixed constants $c$. Union bounding over $\ell \in [k]$, and observing that $n_0 = \Omega(rk\log(k)/2^{2s})$ with probability at least $1-1/k^{8}$ for all $\ell \in [k]$, 
\begin{equation}
\label{eqn:sample_complexity_cov}
\left(1-\frac{1}{2^{2s}}\right)\Sigma_{\ell} \preceq \hat{\Sigma}_{\ell} \preceq \left(1+\frac{1}{2^{2s}}\right) \Sigma_{\ell}
\end{equation}
for any $s>2$, 
which concludes the empirical covariance condition.
\noindent By definition of a ``nice'' distribution, we know that the points in $C_{\ell}$ are drawn i.i.d. from a $s$-certifiably $(C,\delta)$-anti-concentrated distribution denoted by $\cD(\mu_\ell, \Sigma_{\ell})$ and thus for all $\eta$,
\[
\sststile{2s}{v} \Set{
\expecf{x,y \sim \cD(\mu_\ell, \Sigma_{\ell}) }{ q^2_{\eta, \Sigma_{\ell}} \Paren{\iprod{ x-y,v} }} \leq C\eta \Paren{v^{\top}\Sigma_{\ell} v}^{s}}
\]
Consider the substitution $v \to \Sigma^{\dagger/2} v$. Then,
\[
\sststile{2s}{v} \Set{
\expecf{x,y \sim \cD(\mu_\ell, \Sigma_{\ell}) }{ q^2_{\eta, \Sigma_{\ell}} \Paren{\Iprod{\Sigma^{\dagger/2}_{\ell}(x-y),v} }} \leq C\eta \Norm{v}^{2s}_2}
\] 
Since $q_{\eta, \hat{\Sigma}}$ is a degree-$s$ even polynomial, $q^2_{\eta, \hat{\Sigma}}(z) = \sum_{i\in[s]}c_{i}z^{2i}$ and thus using the substitution rule, 
\Anote{still assuming even-ness of q.}
\begin{equation}
\label{eqn:cert_anit_conc_dist}
\sststile{2s}{v} \Set{\sum_{j \in [s]} c_i \left\langle \E_{ x, y \sim \cD(\mu_\ell, \Sigma_{\ell})} \Paren{\Sigma^{\dagger/2}_{\ell}(x-y)}^{\otimes 2j} , v^{\otimes 2j} \right\rangle \leq C\eta \Norm{v}^{2s}_2} 
\end{equation}
Let $\cD$ be the true distribution and $\cD'$ be the uniform distribution over $n$ samples from $\cD$.
We can rewrite the above expression by adding and subtracting $\E_{ x, y \sim \cD'} \Paren{\Sigma^{\dagger/2}_{\ell}(x-y)}^{\otimes 2j}$ as follows:
\begin{equation}
\label{eqn:cert_anit_conc_dist2}
\begin{split}
\sststile{2s}{v}\Bigg\{  \frac{k^2}{n^2}\sum_{\substack{ i\neq j \in C_\ell}}  q^2_{\eta, \hat{\Sigma}(r)} \Paren{x_i-x_j,\Sigma^{\dagger/2} v } & \leq   \sum_{j \in [s]} c_i \left\langle \E_{ x, y \sim \cD} \Paren{\Sigma^{\dagger/2}_{\ell}(x-y)}^{\otimes 2j} - \E_{ x, y \sim \cD'} \Paren{\Sigma^{\dagger/2}_{\ell}(x-y)}^{\otimes 2j} , v^{\otimes 2j} \right\rangle \\
& + C\eta \Norm{v}^{2s}_2 \Bigg\}
\end{split}
\end{equation}
By definition of a reasonable distribution, we know that $\Sigma^{\dagger/2}(x-y)$ is certifiably hypercontractive (and thus subgaussian with covariance bounded by identity).
Then, using concentration of polynomials of sub-exponential random variables, for all $i_1,i_2 \in [d^j]$, 
\begin{equation*}
\begin{split}
 \Pr_{x\sim \cD}\Bigg[ \Big|\expecf{x,y\sim\cD(\mu_\ell, \Sigma_{\ell})}{((x-y)^{\otimes j})_{i_1} ((x-y)^{\otimes j})_{i_2}} & -  \expecf{x,y\sim\cD(\mu_\ell, \hat{\Sigma}_{\ell})}{((x-y)^{\otimes j})_{i_1} ((x-y)^{\otimes j})_{i_2}}\Big|  > \epsilon \Bigg] \\
& \leq \exp\left(-\left(\frac{ \epsilon n }{\expec{x,y}{((x-y)^{\otimes j})_{i_1} ((x-y)^{\otimes j})_{i_2}}^2}\right)^{\frac{1}{2s}} \right)
\end{split}
\end{equation*}

Setting $\epsilon = \expecf{x,y\sim\cD(\mu_\ell, \Sigma_{\ell})}{((x-y)^{\otimes j})_{i_1} ((x-y)^{\otimes j})_{i_2}}/2^{2s}$, and union bounding over $d^s$ entries, we can bound error probability by 
$d^{2s} \exp\left(-\left(\frac{ n }{(2d)^{O(s)}}\right)^{\frac{1}{2s}} \right)$.
Therefore, setting $n = \Omega( (s d\log(d))^s)$ suffices and substituting $v \to \Sigma^{1/2}v$, we have  with probability $1-1/\poly(d)$,
\begin{equation}
\label{eqn:cert_anit_conc_dist3}
\sststile{2s}{v} \Set{ \frac{k^2}{n^2}\sum_{\substack{ i\neq j \in C_\ell}}  q^2_{\eta, \hat{\Sigma}(r)} \Paren{x_i-x_j,v } \leq  \left(1 + \frac{1}{2^{2s}}\right)^s  \sum_{j \in [s]} c_i \left\langle \E_{ x, y \sim \cD(\mu_\ell, \Sigma_{\ell})} (x-y)^{\otimes 2j}, v^{\otimes 2j} \right\rangle + C\eta \Paren{v^{\top}\Sigma_{\ell} v}^{2s}_2} 
\end{equation}
Applying the definition of certifiable anti-concentration again, and using the spectral closeness from Eqn \eqref{eqn:sample_complexity_cov}, we can conclude 
\begin{equation}
\label{eqn:cert_anit_conc_dist3}
\sststile{2s}{v} \Set{ \frac{k^2}{n^2}\sum_{\substack{ i\neq j \in C_\ell}}  q^2_{\eta, \hat{\Sigma}(r)} \Paren{x_i-x_j,v } \leq  10C\eta \Paren{v^{\top} \hat{\Sigma}_{\ell} v}^{2s}_2} 
\end{equation}
\Anote{ The trick of first substituting $\Sigma^{\dagger/2}v$ and then substituting $\Sigma^{1/2}v$, then observing that $\Sigma^{\dagger/2}\Sigma^{1/2} (x_1 - x_2) = (x_1 - x_2)$ leads to a sos proof of the above without breaking v into the component in the subspace etc.  }
A similar proof applies to $4$-tuples and yields the second property for anti-concentration. 



Since for all $\ell \in [k]$, $\cD(\mu_\ell, \Sigma_{\ell})$ is also $s$-certifiably $C$-hypercontractive, 
\begin{equation}
\label{eqn:cert_hyper_contractivity_dist}
\sststile{2s}{Q} \Set{ \E_{x, y \sim \cD(\mu_\ell, \Sigma_{\ell})} \left[  ( (x-y)^{\top} Q (x-y) )^{s} \right] \leq (Cs)^s \E_{x \sim \cD(\mu_\ell, \Sigma_{\ell})} \left[  ((x-y)^{\top} Q (x-y))^2 \right]^{s/2} }
\end{equation}
Substituting $Q = \Sigma^{\dagger/2} Q \Sigma^{\dagger/2}$ and observing $(x-y)^{\top} \Sigma^{\dagger/2} Q \Sigma^{\dagger/2} (x-y) = \Iprod{ \Sigma^{\dagger/2}(x-y)(x-y)^{\top}\Sigma^{\dagger/2}, Q} = \Iprod{ \Paren{\Sigma^{\dagger/2}(x-y)}^{\otimes 2}, Q} $, we have

\begin{equation}
\label{eqn:cert_hyper_contractivity_dist2}
\sststile{2s}{Q} \Set{ \E_{x,y \sim \cD(\mu_\ell, \Sigma_{\ell})} \left[ \left( \left\langle (\Sigma^{\dagger/2}(x-y) )^{\otimes 2} ,Q \right\rangle  \right)^{s} \right] \leq (Cs)^s \E_{x \sim \cD(\mu_\ell, \Sigma_{\ell})} \left[  ((x-y)^{\top} \Sigma^{\dagger/2} Q \Sigma^{\dagger/2} (x-y))^2 \right]^{s/2} }
\end{equation}
Observing that $\expecf{x,y\sim \cD}{(x-y)}=0$, we can apply Lemma \ref{lem:mean-variance-general-polynomials} to derive
\begin{equation}
\label{eqn:cert_hyper_contractivity_dist2}
\sststile{2s}{Q} \Set{   \left( \left\langle \E_{x,y \sim \cD(\mu_\ell, \Sigma_{\ell})} (\Sigma^{\dagger/2}(x-y) )^{\otimes 2s} ,Q^{\otimes s} \right\rangle  \right)  \leq (Cs)^{2s} \|  Q  \|^2_F  }
\end{equation}
Let $\cD$ represent the true distribution and $\cD'$ represent the uniform distribution over pairs $(x_i,x_j)$ sampled from $\cD$. Then, adding and subtracting $\Iprod{\E_{x,y \sim \cD'} (\Sigma^{\dagger/2}(x-y) )^{\otimes 2s} , Q^{\otimes s}}$, we have

\begin{equation}
\label{eqn:cert_hyper_contractivity_dist2}
\sststile{2s}{Q} \Set{  \frac{k^2}{n^2} \sum_{\substack{ i\neq j \in C_\ell}} \left((x-y)^{\top} \Sigma^{\dagger/2} Q \Sigma^{\dagger/2} (x-y)\right)^s   \leq |\Delta| + (Cs)^{2s} \| Q  \|^2_F  }
\end{equation}
where $\Delta = \Iprod{\E_{x,y \sim \cD'} (\Sigma^{\dagger/2}(x-y) )^{\otimes 2s} , Q^{\otimes s}} - \Iprod{\E_{x,y \sim \cD} (\Sigma^{\dagger/2}(x-y) )^{\otimes 2s} , Q^{\otimes s}}$. Using Lemma \ref{lem:mean-variance-general-polynomials}, we can bound $\Delta$ by $C^s \Norm{Q}^{2s}_F$, to obtain

\begin{equation}
\label{eqn:cert_hyper_contractivity_dist2}
\sststile{2s}{Q} \Set{  \frac{k^2}{n^2} \sum_{\substack{ i\neq j \in C_\ell}} \left((x-y)^{\top} \Sigma^{\dagger/2} Q \Sigma^{\dagger/2} (x-y)\right)^s   \leq  (2Cs)^{2s} \| Q  \|^2_F  }
\end{equation}
Substituting $Q \to \Sigma^{1/2}_{\ell} Q \Sigma^{1/2}_{\ell}$, and observing that $\Sigma^{1/2}_{\ell}\Sigma^{\dagger/2}_{\ell}(x_i - x_j) = (x_i - x_j)$, we can conclude 
\begin{equation}
\label{eqn:cert_hyper_contractivity_dist2}
\sststile{2s}{Q} \Set{  \frac{k^2}{n^2} \sum_{\substack{ i\neq j \in C_\ell}} \left((x-y)^{\top}  Q  (x-y)\right)^s   \leq  (2Cs)^{2s} \| \Sigma^{1/2}_{\ell} Q \Sigma^{1/2}_{\ell}  \|^2_F  }
\end{equation}

A similar argument holds for $4$-tuples of samples. The final claim about certifiably bounded variance property follows by a similar bound on the empirical moments of the distribution along with Lemma~\ref{lem:bounded-variance-under-sampling}. This concludes the proof.  
\end{proof}
\newcommand{\thr}{\mathsf{thr}}

\section{Polynomial Approximators for Thresholds} \label{sec:poly-approx-threshold}

We will use elementary approximation theory to construct the polynomial.  
\begin{fact}[Jackson's Theorem] \label{fact:jackson}
Let $f:[-1,1]\rightarrow \R$ be continuous. 
Let the modulus of continuity of $f$ be defined as $\omega(\delta) = \sup_{x,y \in [-1,1]} \Set{|f(x)-f(y)| \leq \delta}$ for every $\delta > 0$.
Then, for every $b$, there's a degree $b$ polynomial $p$ such that for every $x \in [-1,1]$, 
\[ 
|p(x)-f(x)| \leq 6 \omega(1/b)\mper
\]
\end{fact}

The following lemma gives an ``amplifying polynomial'' as in ~\cite{DBLP:journals/eccc/DiakonikolasGJSV09} and is an easy consequence of Chernoff bounds.

\begin{fact}[Claim 4.3 in \cite{DBLP:journals/eccc/DiakonikolasGJSV09}] \label{fact:amplification}
Let $A_q(u) = \sum_{j\geq q/2} {q \choose j} \Paren{\frac{1+u}{2}}^j  \Paren{\frac{1-u}{2}}^{q-j}$. Then, $A_q$ is a degree $q$ polynomial that satisfies:
\begin{enumerate}
\item $A_q(u) \in [1-e^{q/6},1]$ for all $u \in [3/5,1]$,
\item $A_q(u) \in [0,e^{-q/6}]$ for all $u \in [-1,-3/5]$,
\item $A_q(u) \in [0,1]$ for all $u \in [-1,1]$.
\end{enumerate}

\end{fact}

\begin{proof}[Proof of Lemma~\ref{lem:poly-approximate-threshold}]
Let $\thr:[0,1]\rightarrow [0,1]$ be any function that is $0$ on $[0,c]$, $1$ on $[2c,1]$

Consider the piecewise linear function $f:[0,1] \rightarrow [0,1]$ such that $f(x) = 0$ whenever $|x| \leq c$, $f(x) = 1$ for $|x|\geq 2c$ and $f(x) = \frac{(x-c)}{c}$ otherwise. Then, $f$ is continuous. Further, the modulus of continuity, $\omega(\delta)$ for $f$ is at most $\frac{1}{c\delta}$.

Taking $q = 25/c$ and applying Fact~\ref{fact:jackson} yields a polynomial $J(t)$ of degree at most $q$ such that:
\[
\max_{t \in [-1,1]} |J(t)-f(t)| \leq 1/4\mper
\]
We now "amplify"  this polynomial to get the final construction. 

Let $p(t) = \Paren{A_r(8/5J(t)-4/5)}^2$ for $r = 15 \log(1/\eta)$. Then, the argument of $A_r$ in $p(t)$ lies in $[3/5,1]$ whenever $t \geq 2c$ and in $[-1,-3/5]$ whenever $t \in [0,c]$. Thus, applying Fact~\ref{fact:amplification} yields that:
\[
\sup_{t\in [0,c]\cup [2c,1]} |p(t) - \thr(t)| \leq 2 e^{-r/6} \leq \eta\mper
\]
\end{proof} 

\section{TV-Close Subgaussian Distributions with Arbitrarily Far Parameters} \label{sec:subgaussian-no-TV-vs-param}

We give a simple example of a pair of (one-dimensional) subgaussian distributions that are $(1-\eta)$-close in TV-distance for some $\eta <1/2$ while have an arbitrarily separated variances. 

For $i =1 ,2$, let $\cD_i$ be the distribution on $\R$ that outputs $0$ with probability $\eta<1/2$ and a sample from Gaussian $\cN(0,\sigma_i^2)$ otherwise. Observe that $\cD_1,\cD_2$ are clearly $2$-subgaussian: $\E_{\cD_i} x^{2} = (1-\eta)\sigma_i^2$ while for every $t$, $\E_{\cD_i} x^{2t} \leq \Paren{\frac{1}{(1-\eta)}}^{t} \Paren{\E_{\cD_i} x^2}^t$. Thus, both $\cD_1, \cD_2$ are $C = \frac{1}{(1-\eta)} \leq 2$-subgaussian. Further, since $\Pr_{\cD_i} [x=0] \geq \eta$, it's immediate that $\dtv(\cD_1, \cD_2)\leq (1-\eta)$. However, since we can choose $\sigma_1,\sigma_2$ arbitrary, the variances of $\cD_1,\cD_2$ are arbitrarily far.

Observe, however, that both $\cD_1,\cD_2$ are \emph{not} anti-concentrated in the construction above. 
Observe, further that when $\eta$ gets close to $1$ (instead of $\leq 1/2$), the constant $C$ in Sub-gaussianity blows-up. Thus, if we fix $C$ before-hand and look at all C-subgaussian distributions, then we can hope to prove TV-closeness implies parameter closeness when TV distance is small enough but not when it's close to $1$.



\end{document}